\documentclass[final,3p,times,authoryear, number,comma,square]{elsarticlePLF} 
\usepackage{amssymb,color}
\usepackage{a4wide}
\usepackage{amsmath,amsthm}
\usepackage{amssymb}
\newtheorem{definition}{Definition}[section]
\newtheorem{theorem}[definition]{Theorem}
\newtheorem{proposition}[definition]{Proposition}
\newtheorem{corollary}[definition]{Corollary}
\newtheorem{lemma}[definition]{Lemma}
\newtheorem{remark}[definition]{Remark}

\numberwithin{equation}{section}
\newcommand \bei {\begin{itemize}}
\newcommand \eei {\end{itemize}}

\newcommand \str {\mathcal E}
\newcommand \Rbf {\mathbf R}
\newcommand \Wbf {\mathbf W}
\newcommand \Sbf {\mathbf S}
\newcommand \eps \epsilon 
\newcommand \Rcal {\mathcal{R}} 
\newcommand \Pcal {\mathbb{P}}

\newcommand \Uh {\widehat U}

\newcommand \supp {\text{supp}} 

\newcommand \Rr {\Rbf^\rightarrow}
\newcommand \Rl {\Rbf^\leftarrow}
\newcommand \Sr {\Sbf^\rightarrow} 
\newcommand \Sl {\Sbf^\leftarrow}
\newcommand \Wavr {\Wbf^\rightarrow}
\newcommand \Wl {\Wbf^\leftarrow}

\newcommand \Mcal {\mathcal M}
\newcommand \Hcal {\mathcal H}
\newcommand \gb  {g_0}
\newcommand \mub {\mu_0}
\newcommand \ub  {u_0} 
\newcommand \Mb {M_0} 
\newcommand \Vb {V_0} 
\newcommand \ab {a_0} 
\newcommand \bb {b_0}

\newcommand \Mh {\widehat M}

\newcommand \nh {\widehat n}

\newcommand \D {\sharp}

\newcommand \Lip 	{\text{Lip}}
\newcommand \ra 	\rangle
\newcommand \la 	\langle

\newcommand \del   	\partial
\newcommand \RR	{\mathbb R}
\newcommand \be   	{\begin{equation}}
\newcommand \ee   	{\end{equation}} 
\newcommand{\bel}[1]{\be\label{#1}}
\newcommand{\G}[3]{\ensuremath{\Gamma^{#1}_{#2#3}}} 
\newcommand{\T}[2]{\ensuremath{T^{#1}{}_{#2}}} %
\newcommand \Ma {M_\sharp}
\newcommand \Va {V_\sharp}
\newcommand \aaa {a_\sharp}
\newcommand \ba {b_\sharp}
\newcommand \za {z_\sharp}
\newcommand \wa {w_\sharp}
\newcommand \Ms {M^{(0)}}
\newcommand \Vs {V^{(0)}}
\newcommand \as {a^{(0)}}
\newcommand \bs {b^{(0)}}
\newcommand \mus {\mu^{(0)}}
\newcommand \Mp {M^{(1)}}
\newcommand \Vp {V^{(1)}}
\newcommand \ap {a^{(1)}}
\newcommand \bp {b^{(1)}}

\begin{document}

\begin{frontmatter}

\hfill {\bf Journal de Math\'ematiques Pures et Appliqu\'ees (2014)} 

\

\

\

\title{The formation of trapped surfaces in spherically-symmetric Einstein-Euler spacetimes 
with bounded variation}

\author{Annegret Y. Burtscher$^{1}$ and Philippe G. LeFloch$^2$} 

\address{$^1$ Fakult\"at f\"ur Mathematik, Universit\"at Wien, Oskar-Morgenstern-Platz 1, 
1090 Wien, Austria.  \\
 Email : {\it annegret.burtscher@univie.ac.at.}
\\
$^2$ Laboratoire Jacques-Louis Lions \& Centre National de la Recherche Scientifique, 
Universit\'e Pierre et Marie Curie, 
\\
4 Place Jussieu, 75252 Paris, France. 
Email : {\it contact@philippelefloch.org.}
}

\begin{abstract}
We study the evolution of a self-gravitating compressible fluid in spherical symmetry and 
we prove the existence of weak solutions with bounded variation for the Einstein-Euler equations of general relativity.
We formulate the initial value problem in Eddington-Finkelstein coordinates and prescribe
spherically symmetric data on a characteristic initial hypersurface. 
We introduce here a broad class of initial data which contain no trapped surfaces, 
and we then prove that their Cauchy development contains trapped surfaces. 
We therefore establish the {\it formation of trapped surfaces} in weak solutions to the Einstein equations.
This result generalizes a theorem by Christodoulou for regular vacuum spacetimes (but without symmetry restriction). Our method of proof relies on a generalization of the "random choice" method for nonlinear hyperbolic systems 
and on a detailled analysis of the nonlinear coupling between the Einstein equations and the relativistic Euler equations in spherical symmetry. 
\end{abstract}


\end{frontmatter}

\noindent{\bf R\'esum\'e}

Nous \'etudions l'\'evolution d'un fluide compressible auto-gravitant en sym\'etrie radiale et nous d\'emon\-trons un r\'esultat d'existence de solutions faibles \`a variation born\'ee pour les \'equations d'Einstein-Euler de la relativit\'e g\'en\'erale. 
Nous formulons le probl\`eme de Cauchy en coordonn\'ees d'Eddington-Finkelstein et prescrivons
 des donn\'ees \`a sym\'etrie radiale sur une hypersurface initiale caract\'eristique.  
Nous introduisons ici une classe de donn\'ees initiales qui ne contiennent pas de surfaces pi\'eg\'ees, 
et nous 
d\'emontrons alors que leur d\'eveloppement de Cauchy contient des surfaces pi\'eg\'ees. 
Nous \'etablissons ainsi un r\'esultat de {\sl formation de surfaces pi\'eg\'ees} dans les solutions faibles des \'equations d'Einstein. Ce r\'esultat g\'en\'eralise un th\'eor\`eme de Christodoulou pour les espaces-temps r\'eguliers sans mati\`ere (mais sans restriction de sym\'etrie). 
Notre m\'ethode de preuve s'appuie sur une g\'en\'eralisation de la m\'ethode "random choice" pour les syst\`emes hyperboliques nonlin\'eaires et sur une analyse fine du couplage nonlin\'eaire entre les \'equations d'Einstein et les \'equations d'Euler relativistes en sym\'etrie radiale.


  
\newpage 

\section{Introduction}
\label{sec:1}

We are interested in the problem of the gravitational collapse of compressible matter under the assumption of 
spherical symmetry. When the matter evolves under its self-induced gravitational field,  
two distinct behaviors can be observed: a dispersion of the matter in future timelike directions, or 
a collapse of the matter and the formation of a trapped surface and, under certain conditions, a black hole~\cite{HawkingPenrose,Penrose,Wald}. 
The collapse problem in spherical symmetry was extensively investigated by Christodoulou and followers in the past twenty years, 
under the assumption that the matter is represented by a scalar field \cite{Christo0,Christo1} or is driven by a kinetic
equation like Vlasov equation; cf.~Andreasson \cite{Andreasson}, Andreasson and Rein \cite{AR}, and Rendall \cite{Rendall1,rendall08} and the references cited therein. Furthermore, the problem of the generic formation of trapped surfaces in vacuum spacetimes without symmetry was solved by Christodoulou in the pioneering work \cite{Christo-book}. 

In recent years,  the second author together with collaborators \cite{BLSS,GrubicLeFloch,LeFloch,LeFlochRendall,LeFlochStewart1,LeFlochStewart2} 
has initiated the mathematical study of self-gravitating compressible fluids and constructed classes of spacetimes with weak regularity whose curvature is defined in the sense of distributions \cite{LeFlochMardare}. Global existence results have been established for several classes of solutions to the Einstein equations with symmetry. LeFloch and Stewart \cite{LeFlochStewart2} proposed a mathematical theory of the characteristic initial value problem for plane-symmetric spacetimes with weak regularity, while LeFloch and Rendall \cite{LeFlochRendall} and  Grubic and LeFloch \cite{GrubicLeFloch} constructed a global foliation for the larger class of weakly regular spacetimes with Gowdy symmetry. Furthermore, LeFloch and Smulevici \cite{LeFlochSmulevici} developped the theory of weakly regular, vacuum  spacetimes with $T^2$ symmetry.  

The present paper is motivated by Christodoulou's work \cite{Christo-book} on trapped surface formation and, by building upon the mathematical technique \cite{BLSS,LeFloch,LeFlochRendall,LeFlochStewart2}, we are able to construct a large class of spherically-symmetric Einstein-Euler spacetimes which have bounded variation and exhibit trapped surface formation. 
We thus consider matter spacetimes $(\Mcal,g)$ (with bounded variation) satisfying the Einstein equations 
\bel{eq:Einstein}
 G^{\alpha\beta} = 8 \pi T^{\alpha\beta}
\ee
understood in the distributional sense (see Section~\ref{sec:3}, below), when the geometry described by the Einstein tensor $G^{\alpha\beta}$ is coupled
to the matter content governed by the energy-momentum tensor 
\bel{eq:EM}
T^{\alpha\beta} = (\mu + p) u^\alpha u^\beta + p \, g^{\alpha\beta}. 
\ee
Here, all Greek indices take values $0, \ldots, 3$ and implicit summation over repeated indices is used. 
According to the Bianchi identities satisfied by the geometry, \eqref{eq:Einstein}-\eqref{eq:EM} imply the Euler equations
\bel{euler}
\nabla_\alpha T^{\alpha\beta} = 0. 
\ee
In \eqref{eq:EM}, $\mu$ denotes the mass-energy density of the fluid and $u^\alpha$ its velocity vector, which is 
normalized to be of unit norm $u^\alpha u_\alpha = -1$, while the pressure function $p=p(\mu)$ is assumed to 
depend linearly on $\mu$, that is, 
\bel{eq:pressure}
 p = k^2 \mu. 
\ee
The constant $k \in (0,1)$ represents the speed of sound, while the light speed is normalized to unit.
 
In the present paper, we thus investigate the class of spherically symmetric spacetimes governed by the Einstein-Euler
equations \eqref{eq:Einstein}-\eqref{euler}, and after formulating the initial value problem with data posed on a spacelike
hypersurface, we establish several results concerning their local and global geometry. 
The main challenge overcome is coping with the weak regularity of the spacetimes under consideration, which is 
necessary since shock waves are expected to form in the fluid even if the initial data are smooth (cf.~Rendall and St{\aa hl} \cite{rendallstahl}). 
Our main result is now stated, in which we are able to identify a large class of initial data leading to the formation of trapped
two-spheres.

\begin{theorem}[A class of spherically-symmetric Einstein-Euler spacetimes with bounded variation]
\label{th:main} 
By solving the initial value problem from a class\footnote{specified explicitly in  {Corollary~\ref{thm:trapped} and} Proposition~\ref{prop67} below} 
 of initial data set $(\Hcal, \gb, \mub, \ub)$ with spherical symmetry and bounded variation, prescribed on a hypersurface $\Hcal \subset \Mcal$, 
one obtains a class of Einstein-Euler spacetimes $(\Mcal,g, \mu, u)$ with bounded variation satisfying
\eqref{eq:Einstein}--\eqref{eq:pressure}, together with the following conditions: 
\bei 

\item[1.] The spacetime is a spherically symmetric, future development of the initial data set.

\item[2.] The initial hypersurface does not contain trapped spheres. 

\item[3.]  The spacetime contains trapped spheres. 

\eei 
\end{theorem}

The notion of {\it spacetime with bounded variation\/} used in the above theorem will be presented in Section~\ref{sec:3}. 
Observe that to establish the above theorem we need not construct the maximal development of the given initial data set, but
solely to establish that the solution to the Einstein-Euler system exists in a ``sufficiently large'' time interval within which
trapped surfaces have formed. 

An outline of this paper is as follows. In Section~\ref{sec:2}, we express the spacetime metric in generalized Eddington-Finkelstein
coordinates and we write the Einstein-Euler equations for spherically symmetric solutions as a first-order partial differential system
which, later in Section~\ref{sec:4}, will be shown to be hyperbolic.  Our choice of coordinates guarantees that the trapped region of
the development can be reached in the chosen coordinates. For instance, let us illustrate this choice with the Schwarzschild metric
which, in the Eddington-Finkelstein coordinates considered in the present work, reads ($m>0$ representing the mass of the black hole)
\bel{eq:104} 
g = - \big( 1 - 2m/r\big) \, dv^2 + 2 dvdr + r^2 \big( d\theta^2 + \sin^2 \theta \, d\varphi^2 \big), 
\ee
in which $(v , r) \in [0, +\infty) \times (0 +\infty)$ and the variable $(\theta, \varphi)$ parametrizes the two-spheres. The coefficients are
regular everywhere except at the center $r=0$ (where the curvature blows up~\cite{HawkingEllis})   
and these coordinates
allow us to ``cross'' the horizon $r=2m$ and ``enter'' the trapped region $r < 2m$. In contrast, in the so-called Schwarzschild
coordinates, we have (with $v = t + r + 2m \ln (r- 2m)$) 
$$
g = - \big( 1 - 2m/r\big) \, dt^2 + \big( 1 - 2m/r\big) ^{-1} \, dr^2 + r^2 \big( d\theta^2 + \sin^2 \theta \, d\varphi^2 \big), 
$$ 
and the metric coefficients suffer an (artificial) singularity around $r=2m$, so that these latter coordinates can not be used for our
purpose. The generalized Eddington-Finkelstein coordinates mimic \eqref{eq:104} for more general spacetimes (cf.~below). 
  
In Section~\ref{sec:3}, we follow \cite{LeFloch,LeFlochRendall,LeFlochStewart1} and introduce a definition of solutions to the Einstein-Euler
system. We then perform a ``reduction'' of this system, by eliminating certain redundancies in the ``full'' Einstein-Euler system
and arrive at a well-chosen set of ``essential equations''. Throughout the regularity of the solutions is specified and the
equivalence between the original system and the reduced one is established within the class of solutions with bounded variation. 

Before we can proceed with the study of general solutions to the coupled Einstein-Euler system, we investigate a special class of
solutions and, in Section~\ref{sec:5}, we analyze the class of static spacetimes, which are  
described by a system of ordinary differential equations associated with a suitably reduced version of the Einstein-Euler system.
Here, we rely on earlier work by Rendall and Schmidt \cite{rendallschmidt} and Ramming and Rein \cite{RRein}
who, however, assumed a different choice of coordinates.  

In Section~\ref{sec:4}, we investigate the (homogeneous version of the) Euler equations on a fixed background and, specifically, we
solve the so-called Riemann problem for the Euler equations in Eddington-Finkelstein coordinates. Since shock waves are expected to
form in finite time, it is natural to investigate initial data that consist of two constant states separated by an initial jump
discontinuity. In this ``ideal'' situation, the solutions to the Euler system (after having neglected the coupling with the Einstein
equations) can be given in closed form. This Riemann problem,  in turn, is fundamental in building a general solutions with arbitrary
initial data, as we explain in Section~\ref{sec:6}, below. 
 
Our key contribution in the present work is the identification of a large class of untrapped initial data whose Cauchy development contains trapped surfaces (arising therefore during the evolution). 
In Section~\ref{sec:6}, we introduce the class of initial data of interest and we state a precise version of our main result for the
Einstein-Euler system in spherical symmetry. We rely on the random choice method (for which we refer to
\cite{Dafermos,Glimm,LeFloch-book} and, more specifically, Smoller and Temple \cite{SmollerTemple} fas far as the relativistic fluid equations is concerned). The Riemann solutions serve as building blocks in order to approximate general solutions and the compactness of these approximate solutions follow from a uniform bound on their total variation. Only {\it local-in-time} existence results via
the random choice method were established earlier, however in other coordinates or under different symmetry assumptions, by Groah and Temple \cite{GroahTemple} and by LeFloch et al. \cite{BLSS,GrubicLeFloch,LeFlochStewart1}. 
Our result is a ``semi-global'' existence result, in the sense that we are able to control the time of existence of the solutions until a trapped surface forms. For clarity in the presentation, all technical estimates are postponed to Section~\ref{sec:7}. 


\section{The Einstein-Euler system in spherical symmetry} 
\label{sec:2}

\subsection{Einstein equations in Eddington-Finkelstein coordinates}
\label{sec:21}

We impose spherical symmetry and express the spacetime metric in the following {\it generalized Eddington-Finkel\-stein coordinates\/}
(following~\cite{Eddington,Finkelstein}):
\bel{eq:metric}
g = -a b^2 \, dv^2 + 2 b \, dv dr + r^2 \, \big( d\theta^2 + \sin^2\theta d\varphi^2 \big). 
\ee
Here, the time variable $v$ lies in some interval $[v_0, v_*]$ and the radius $r$ belongs to some interval  {$[0,r_0)$},   
 while 
$(\theta, \varphi)$ are standard coordinates on the two-sphere. 
The spacetime geometry is described by two metric coefficients such that $a=a(v,r)$ may change sign, but $b=b(v,r)$ remains positive,
and we require the following regularity condition at the center: 
\bel{eq:center}
 \lim_{r\to0} (a,b)(v,r) = (1,1) \qquad \text{ for all relevant } v. 
\ee
In view of \eqref{eq:metric}, the metric and its inverse read 
$$
 (g_{\alpha\beta}) = \left( \begin{array}{cccc}
                             -a b^2 & b & 0 & 0 \\
			          b & 0 & 0 & 0 \\
                                  0 & 0 & r^2 & 0 \\
				  0 & 0 & 0 & r^2 \sin^2 \theta
                            \end{array}
 \right), 
\qquad 
 (g^{\alpha\beta}) = \left( \begin{array}{cccc}
                                  0 & \frac{1}{b} & 0 & 0 \\
			          \frac{1}{b} & a & 0 & 0 \\
                                  0 & 0 & \frac{1}{r^2} & 0 \\
				  0 & 0 & 0 & \frac{1}{r^2 \sin^2 \theta}
                            \end{array}
 \right)
$$
and, therefore, the (non-vanishing) Christoffel symbols 
$\Gamma^\alpha_{\beta \gamma}
= \frac{1}{2}  \, g^{\alpha\delta} \big( \frac{\del g_{\gamma \delta}}{\del x^\beta} + \frac{\del
g_{\beta \delta}}{\del x^\gamma} - \frac{\del g_{\beta\gamma}}{\del x^\delta} \big)
$  
of the connection $\nabla$ read  
\bel{Christoffel}
\begin{array}{lll}
\G{0}{0}{0} = \frac{b_v}{b} + \frac{1}{2} a_r b + a b_r,
\hskip2.cm 
& \G{0}{2}{2} = - \frac{r}{b},
\hskip2.cm 
& \G{0}{3}{3} = - \frac{r}{b} \sin^2 \theta,
\\
\G{1}{0}{0} = - \frac{1}{2} a_v b + \frac{1}{2}a a_r b^2 + a^2 b b_r,
& \G{1}{0}{1} = - \frac{1}{2} a_r b - a b_r,
& \G{1}{1}{1} = \frac{b_r}{b}, 
\\
\G{1}{2}{2} = - r a,
& \G{1}{3}{3} = - r a \sin^2 \theta,
& \G{2}{1}{2} = \frac{1}{r},
\\
\G{2}{3}{3} = - \sin \theta \cos \theta,
& \G{3}{1}{3} = \frac{1}{r},
& \G{3}{2}{3} = \cot \theta.
\end{array}
\ee
Elementary calculations also yield the non-vanishing components of the Einstein tensor: 
\bel{Gab}
\aligned
 G^{00}= \; & \frac{2 b_r}{r b^3}, 
\hskip6.cm
 G^{01} =  \frac{1}{r^2b^2} \Big( r a_r b + ab-b + 2 r a b_r \Big),
\\
 G^{11} = \; & \frac{1}{r^2 b} \Big( a^2b - ab + r a a_r b + 2 r a^2 b_r - r a_v \Big), 
\hskip2.cm
G^{33} = \; (\sin \theta)^{-2} \, G^{22}.
\\
 G^{22} = \; & \frac{1}{2 r^3 b^3} \Big( 2 a_r b^3 + r a_{rr} b^3 + 3r a_r b^2 b_r + 2 a b^2 b_r
+ 2 r a b^2 b_{rr} - 2 r b_v b_r + 2 r b b_{vr} \Big). 
\endaligned
\ee

In the coordinates \eqref{eq:metric} under consideration, the Einstein equations \eqref{eq:Einstein} 
are equivalent to the two ordinary differential equations 
\begin{align}
b_r & = 4 \pi \, r b^3 \,  T^{00},  
\label{eq:Einstein1} 
\\
r a_r b + ab-b + 2 r a b_r & = 8 \pi \, r^2 b^2 \,  T^{01}, 
\label{eq:Einstein2}
\end{align}
and the two partial differential equations
\begin{align}
a^2b - ab + r a a_r b + 2 r a^2 b_r - r a_v & = 8 \pi \, r^2 b \, T^{11},  
\label{eq:Einstein3} 
\\
2 a_r b^3 + r a_{rr} b^3 + 3r a_r b^2 b_r + 2 a b^2 b_r  + 2 r a b^2 b_{rr} - 2 r b_v b_r + 2 r b b_{vr} 
& = 16 \pi \, r^3 b^3 \, T^{22}.
\label{eq:Einstein4}
\end{align}
The remaining Einstein equations 
\be
\label{eq:musthold}
T^{02} = T^{03} = T^{12} = T^{13} = T^{23} = T^{22} - (\sin \theta)^2 \, T^{33} = 0, 
\ee
should be seen as compatibility condition that the matter model must satisfy and, indeed under our symmetry assumption, it will be
straightforward to check \eqref{eq:musthold} for the energy momentum tensor \eqref{eq:EM}.

Let us make some remarks about the structure of \eqref{eq:Einstein1}--\eqref{eq:Einstein4}. We have here equation for the derivatives 
$a_r$ and $b_r$, which can be integrated and provide $a$ and $b$ when the matter content is ``known'': 

\bei

\item[1.] On one hand, by integrating \eqref{eq:Einstein1}, 
since $b$ is positive, we find 
\bel{eq:300}
{b(v,r)^{-2}} = 1 - 8 \pi \int_0^r T^{00}(v,r') \, r' dr', 
\ee
{\it provided\/} the matter density $r T^{00}$ is locally integrable on $[0, +\infty)$. This formula implies $b(v,r) \leq 1$ for $r> 0$ 
and, moreover, one needs that $8 \pi \int_0^{+\infty} T^{00}(v,r') \, r' dr' \leq 1$ in order for $b$ to remain non-negative for all $r$.

\item[2.] On the other hand, by combining \eqref{eq:Einstein1} and \eqref{eq:Einstein2}, we obtain 
$$
 \del_r (ra-r) + (ra-r) 8 \pi rb^2 T^{00} = 8 \pi r^2 b (T^{01} - b T^{00}).
$$
The integrating factor $C(v,r) = 8 \pi \int_{r_0}^r T^{00}  r' b^2 \, dr'$ allows us to write  
$$
 \del_r \left( e^C (r a - r) \right) = e^C \, \big( 8 \pi r^2 b (T^{01} - b T^{00}) \big)
$$ 
and, therefore, 
\bel{eq:301}
 a(v,r) = 1 - \frac{1}{r} \int_0^r \left(8 \pi {r'}^2 b (T^{01} - b T^{00}) \right) \exp\Big( -8 \pi \int_{r'}^r  T^{00} \, r'' b^2 
 \, dr'' \Big) \, dr',
\ee
{\it provided\/} the integrand above is locally integrable on $[0, +\infty)$.

\item[3.] The above formulas for the coefficients $a,b$ use only two of the Einstein equations, the remaining ones can be thought of as
constraints, which can then be deduced from \eqref{eq:300}-\eqref{eq:301}.

\eei 

We will see shortly below that \eqref{eq:301} is correct, but that \eqref{eq:300} must be revisited and a different ``weight'' in
$r$ is required. 


\subsection{Euler equations in Eddington-Finkelstein coordinates}
\label{sec:22}

Under the assumption of spherical symmetry and when the metric is expressed in the Eddington-Finkelstein coordinates \eqref{eq:metric}, the
energy density and the velocity vector depend on the variables $(v,r)$, only, and we can write 
$\mu = \mu(v,r)$ and $u^\alpha = u^\alpha(v,r)$.
We now express the Euler equations $\nabla_\alpha T^{\alpha\beta} = 0$, 
and we are content with the components $\beta=0,1$, since the remaining two components $\beta=2,3$ will follow from the former. 
(Cf.~Section~\ref{sec:reducedsystem}, below.) 
 
\begin{lemma} 
\label{lem:21}
Under the assumption of spherical symmetry and in the generalized Eddington-Finkelstein coordinates \eqref{eq:metric}, the Euler
equations take the form 
$$ 
\aligned
 0 = \; & \del_v \left( \mu (1 + k^2) u^0 u^0 \right) + \del_r \left( \mu (1 + k^2) u^0 u^1 + k^2 \frac{\mu}{b} \right) \\
 & + \left( \frac{2 b_v}{b} + \frac{a_r b}{2} + a b_r \right) \mu (1 + k^2) u^0 u^0
 + \left( \frac{b_r}{b} + \frac{2}{r} \right) \left( \mu (1 + k^2) u^0 u^1 + k^2 \frac{\mu}{b} \right) - \frac{2 k^2}{rb} \mu, 
\endaligned
$$  
$$
\aligned
0  = \; & \del_v \left( \mu (1 + k^2) u^0 u^1 + k^2 \frac{\mu}{b} \right) + \del_r
\left( \mu (1 + k^2) u^1 u^1 + k^2 \mu a \right) \\
 & + \left( - \frac{a_v b}{2} + \frac{a a_r b^2}{2} + a^2 b b_r \right) \mu (1 + k^2) u^0 u^0
+ \left( \frac{b_v}{b} - a_r b
- 2 a b_r \right) \left( \mu (1 + k^2) u^0 u^1 + k^2 \frac{\mu}{b} \right) 
\\
 & + \left( \frac{2 b_r}{b} + \frac{2}{r} \right) \left( \mu (1 + k^2) u^1 u^1 + k^2 \mu a \right)
- \frac{2 k^2 a}{r} \mu. 
\endaligned
$$
\end{lemma}

\begin{proof}
The energy-momentum tensor defined in \eqref{eq:EM} reads 
$$ 
  \left( \begin{array}{cccc}
\hskip-.6cm \mu (1 + k^2) u^0 u^0 & \mu (1 + k^2) u^0 u^1 + \frac{k^2 \mu}{b} & \mu (1 + k^2) u^0 u^2
&  \hskip-.6cm \mu (1 + k^2) u^0 u^3 \\
 			     \mu (1 + k^2) u^0 u^1 +\frac{ k^2 \mu}{b} & \mu (1 + k^2) u^1 u^1 + k^2 \mu a
& \hskip-.6cm \mu (1 + k^2) u^1 u^2 & \mu (1 + k^2) u^1 u^3 \\
 			     \mu (1 + k^2) u^0 u^2 & \mu (1 + k^2) u^1 u^2 & \mu (1 + k^2) u^2 u^2 + \frac{k^2\mu}{r^2} 
& \hskip-.6cm \mu (1 + k^2) u^2 u^3 \\
 			     \mu (1 + k^2) u^0 u^3 & \mu (1 + k^2) u^1 u^3 & \mu (1 + k^2) u^2 u^3
&\hskip-.6cm  \mu (1 + k^2) u^3 u^3 + \frac{k^2 \mu}{r^2 \sin^2 \theta}
                             \end{array}
\right).
$$
In view of the Einstein equations \eqref{eq:Einstein} and the expressions of the components in Section~\ref{sec:21}
(cf.~the conditions \eqref{eq:musthold}), several components of the energy-momentum tensor vanish, that is, 
$$
 T^{02} = T^{03} = T^{12} = T^{13} = T^{23} = 0.
$$
The normalization $-1 = u^\alpha u_\alpha$ implies that $u^0 \neq 0$, and it thus follows from the condition $\mu>0$
that the last two components of the velocity vector vanish, i.e.~we have $u^2 = u^3 = 0$. Consequently, the energy--momentum
tensor has the form 
$$
 (T^{\alpha\beta}) = \left( \begin{array}{cccc}
                             \mu (1 + k^2) u^0 u^0 & \mu (1 + k^2) u^0 u^1 +  \frac{k^2\mu}{b} & 0 & 0 \\
			     \mu (1 + k^2) u^0 u^1 +  \frac{k^2\mu}{b} & \mu (1 + k^2) u^1 u^1 + k^2 \mu a & 0 & 0 \\
			     0 & 0 & \frac{k^2\mu}{r^2} & 0 \\
			     0 & 0 & 0 & \frac{k^2 \mu}{r^2 \sin^2 \theta}
                            \end{array}
\right).
$$

Using $\nabla_\delta T^{\alpha\beta} = \T{\alpha\beta}{,\delta} + \G{\alpha}{\gamma}{\delta} T^{\gamma\beta}
+ \G{\beta}{\gamma}{\delta} T^{\alpha\gamma}$ 
and the expressions \eqref{Christoffel} of the Christoffel symbols, we obtain ($\beta = 0$) 
\bel{eq:nablaT0} 
\aligned
\nabla_0 \T{00}{} = \; & \T{00}{,0} + 2 \G{0}{0}{0} T^{00} = \T{00}{,0} + \Big( \frac{2 b_v}{b} + a_r b + 2 a b_r  \Big) T^{00},
\\
\nabla_1 \T{10}{} = \; & 
\T{10}{,1} + \G{1}{0}{1} T^{00} + \G{1}{1}{1} T^{01} = \T{10}{,1} + \Big( - \frac{a_r b}{2} - a b_r\Big) T^{00}
 + \frac{b_r}{b} T^{01},
 \\
\nabla_2 \T{20}{} = \; & \G{2}{1}{2} T^{01} + \G{0}{2}{2} T^{22} = \frac{1}{r} T^{01} - \frac{r}{b} T^{22},
 \\ 
\nabla_3 \T{30}{} = \; & \G{3}{1}{3} T^{01} + \G{0}{3}{3} T^{33} = \frac{1}{r} T^{01} - \frac{r \sin^2 \theta}{b} T^{33},
\endaligned
\ee
and ($\beta = 1$) 
\bel{eq:nablaT1}
\aligned
\nabla_0 \T{01}{} = \; & \T{01}{,0} + \G{0}{0}{0} T^{01} + \G{1}{0}{0} T^{00} + \G{1}{0}{1} T^{01},
\\
                           = \; & \T{01}{,0} + \left( - \frac{a_v b}{2} + \frac{a a_r b^2}{2} + a^2 b b_r \right) T^{00}
                           + \frac{b_v}{b} T^{01},
\\
\nabla_1 \T{11}{} = \; 
& \T{11}{,1} + 2 ( \G{1}{0}{1} T^{01} + \G{1}{1}{1} T^{11} ) = \T{11}{,1}  - \left( a_r b + 2 a b_r \right) T^{01}
     + \frac{2 b_r}{b} T^{11},
 \\
\nabla_2 \T{21}{} = \; & \G{2}{1}{2} T^{11} + \G{1}{2}{2} T^{22} = \frac{1}{r} T^{11} - r a T^{22},
 \\
\nabla_3\T{31}{} = \; & \G{3}{1}{3} T^{11} + \G{1}{3}{3} T^{33} = \frac{1}{r} T^{11} - r a \sin^2 \theta T^{33}.
\endaligned
\ee
With $\beta= 2,3$, the corresponding components vanish identically and provide no further relations.  
 
Based on the above relations, we can now compute the first equation of the Euler system, obtained by setting $\beta = 0$ in
\eqref{euler}, that is,
\be 
\label{eqE1}
\aligned
 0 = \nabla_\alpha \T{\alpha 0}{} 
= \; & \del_v T^{00} + \del_r T^{10} + \Big( \frac{2 b_v}{b} + \frac{a_r b}{2} + a b_r \Big) \, T^{00}
+ \left( \frac{b_r}{b} + \frac{2}{r} \right) T^{01} - \frac{2r}{b} T^{22} 
\endaligned
\ee
and next,  $\beta = 1$, 
\be
\label{eqE2}
\aligned 
 0 = \nabla_\alpha \T{\alpha 1}{} 
= \; & \del_v T^{01} + \del_r T^{11} + \left( - \frac{a_v b}{2} + \frac{a a_r b^2}{2}
+ a^2 b b_r \right) T^{00} \\
 & + \left( \frac{b_v}{b} - a_r b - 2 a b_r \right) T^{01} + \left( \frac{2 b_r}{b} + \frac{2}{r} \right) T^{11} - 2 r a T^{22}. \qedhere
\endaligned
\ee 
\end{proof} 

\subsection{Formulation as a first-order system with source-terms}
\label{ssec:formulation}
 
Since we are interested in solutions with low regularity, it is necessary to put the principal parts of the Euler equations in a
divergence form. We are now going to check that all $v$-derivatives of the metric coefficients can be ``absorbed'' in the
principal part of the Euler equations, while all $r$-derivatives of the coefficients can be replaced by algebraic expressions
involving no derivatives.
To this end, we find it convenient to normalize the fluid variables in generalized Eddington-Finkelstein coordinates, and we introduce  
\bel{normal-fluid}
M := b^2 \mu \, u^0 u^0 \in (0, +\infty), 
\qquad
V := \frac{u^1}{b \, u^0} - \frac{a}{2} \in (-\infty, 0), 
\ee
which we refer to as the {\bf normalized fluid variables.} 
We also introduce the constant 
$$
K^2 := \frac{1-k^2}{1+k^2},
$$
which naturally arises in  the principal part of the Euler equations after multiplication by $1/(1+k^2)$. In terms of the variables $(M,V)$, 
the energy-momentum tensor read 
\bel{Tab}
\aligned
 T^{00} & = (1+k^2) \frac{M}{b^2}, &   T^{11} & = (1+k^2) M  \left( \frac{a^2}{4} + K^2 a V + V^2 \right), 
\\
 T^{01} & = (1+k^2) \frac{M}{b} \left( \frac{a}{2} + K^2 V \right), \qquad & T^{22} & = - \frac{2k^2}{r^2} M V.
\endaligned
\ee
Observe that $V$ is well-defined since $b$ and $u^0$ are, both, non-vanishing. Moreover, $-1 = g_{\alpha\beta} u^\alpha u^\beta$
implies $1 = b  u^0 (ab u^0 - 2 u^1)$, thus  
\bel{u1W} 
u^1 = \frac{1}{2} \left( ab u^0 - \frac{1}{b u^0} \right),
\qquad  
\ee
which was used to derive the sign of $V$ and will be useful later on. 

\begin{proposition}[Formulation of the Euler system in spherical symmetry]
\label{prop:23}
Under the assumption of spherical symmetry and in the  Eddington-Finkelstein coordinates \eqref{eq:metric}, the Euler equations 
\eqref{eq:EM}--\eqref{eq:pressure} for the normalized fluid variables $(M,V)$ defined in \eqref{normal-fluid}
can be expressed as a system of two coupled equations, i.e. 
\bel{eq:Euler}
\del_v U + \del_r F(U,a,b) = S(U,a,b),
\ee
with 
\bel{UF}
U := M \begin{pmatrix} 1 \\ \frac{a}{2} + K^2 V \end{pmatrix}, 
\qquad \qquad 
F(U, a, b) := bM \begin{pmatrix}  \frac{a}{2} + K^2 V \\ \frac{a^2}{4} +K^2 aV + V^2 \end{pmatrix}, 
\ee
and
\be 
\label{S1S2}
\aligned
S(U, a, b) &:= \begin{pmatrix} S_1(M,V,a,b) \\ S_2(M,V,a,b) \end{pmatrix}, 
\qquad \qquad   
S_1(M,V,a,b) := - \frac{1}{2r} b M \left( 1 + a + 4V \right),
\\
 S_2(M,V,a,b) & := - \frac{1}{2r} b M \Big( a^2 + 2 a V (2 + K^2) - 2 K^2 V + 4 V^2 \Big) - 16 \pi (1 - K^2) \, r b \, M^2 V^2,
\endaligned
\ee
in which the constant $K^2 := \frac{1-k^2}{1+k^2} \in (0,1)$ is determined from the sound speed. 
\end{proposition}

\begin{proof}
 We need to rewrite the equations \eqref{eqE1}-\eqref{eqE2} by eliminating the derivatives of the coefficients $a, b$. Combining the
 Einstein equations \eqref{eq:Einstein2} and \eqref{eq:Einstein3} yields an expression for $a_v$, which can thus be eliminated from
 the Euler equations, via 
$$
 \frac{a_v b}{2} = \frac{a_v}{rb} \frac{rb^2}{2} = 4 \pi rb^2 (ab T^{01} - T^{11} ).
$$
 The term $b_v$ can also be eliminated in the right-hand side of the Euler equations
\eqref{eq:Euler}, 
 by relying on the product rule, as follows: 
$$
 \frac{1}{b^2} \, \del_v (b^2 T^{00}) = \del_v T^{00} + \frac{2 b_v}{b} T^{00}, 
\qquad 
 \frac{1}{b} \, \del_v \left( b T^{01} \right) = \del_v T^{01} + \frac{b_v}{b} T^{01}.
$$
Indeed, let us multiply the Euler equation \eqref{eqE1} by $b^2$ and the second one \eqref{eqE2} by $b$: 
\bel{eq:Euler_b1}
\aligned
 \del_v (b^2 T^{00}) + b^2 \del_r T^{01} = \; & b^2 \Big(- \frac{1}{2} \left( a_rb + 2ab_r \right) T^{00} - \left(\frac{b_r}{b}
 + \frac{2}{r}\right)
T^{01} + \frac{2r}{b} T^{22} \Big), 
\\
 \del_v (b T^{01}) + b \del_r T^{11} = \; & b \Big(
\left( a_rb + 2 ab_r \right) (T^{01}- \frac{ab}{2} T^{00}) + 4 \pi rb^2 ( ab T^{01} - T^{11}) T^{00}
\\
& \quad - \left( \frac{2b_r}{b} + \frac{2}{r} \right) T^{11} + 2 r a T^{22} \Big).
\endaligned
\ee

Hence, in order to express the Euler equations in divergence form, we need to include the terms $b^2$ and $b$ in the spatial
derivatives of $T^{01}$ and $T^{11}$, respectively. Again, by the product rule we have  
\begin{align*}
 b^2 \, \del_r T^{01} &= \del_r (b^2 T^{01}) - 2bb_r T^{01}, 
\qquad \qquad 
 b \, \del_r T^{11} = \del_r (b T^{11}) - b_r T^{11},
\end{align*}
 and the system \eqref{eq:Euler_b1} now  reads
\bel{eq:Euler_b2}
\aligned
 \del_v (b^2 T^{00}) + \del_r (b^2 T^{01}) = \; & - \left( a_rb + 2ab_r \right) \frac{b^2}{2} T^{00} + \left(b_r - \frac{2b}{r} \right) b T^{01}
+ 2rb T^{22}, 
\\
 \del_v (b T^{01}) + \del_r (b T^{11}) = \; & \left( a_rb + 2 ab_r \right) b \, (T^{01}- \frac{ab}{2} T^{00})
 + 4 \pi rb^3 ( ab T^{01} - T^{11}) T^{00} 
\\
& - \left( b_r + \frac{2b}{r} \right) T^{11} + 2 rab T^{22}.
\endaligned
\ee
We can now eliminate $a_rb+2ab_r$ by using the second Einstein equation \eqref{eq:Einstein2}, since 
$$
 a_rb+2ab_r = \frac{1}{r} \left( 8 \pi r^2 b^2 T^{01} - (a-1)b \right).
$$
The radial derivative of $b$ is eliminated by using the first Einstein equation \eqref{eq:Einstein1}, that is, 
$b_r = 4 \pi r b^3 T^{00}$. Consequently, the right-hand side of \eqref{eq:Euler_b2} is free of derivatives, i.e. 
$$
\aligned
 \del_v (b^2 T^{00}) + \del_r (b^2 T^{01}) 
 = \; & \frac{(a-1)b^3}{r} T^{00} - \frac{2b^2}{r} T^{01} + 2rb T^{22},
\\
   \del_v (b T^{01}) + \del_r (b T^{11}) 
 = \; & \frac{1}{r} b \left( \frac{1}{2} ab^2 (a-1) T^{00} - (a-1) b T^{01} - 2 T^{11} \right) 
\\
& + 2 r b\left( 4 \pi b^2 (T^{01})^2 - 4 \pi b^2 T^{00} T^{11} + a T^{22} \right).
\endaligned
$$
Recalling the expression  \eqref{Tab} of the energy-momentum tensor, 
we arrive at the form $\del_v U + \del_r F(U) = S(U)$ stated in the proposition. 
\end{proof} 


\section{Einstein-Euler spacetimes with bounded variation} 
\label{sec:3}

\subsection{A notion of weak solutions}
\label{sec:24}

In \eqref{eq:Euler}--\eqref{S1S2}, the Euler equations are expressed as a first-order system of two partial differential equations
in the normalized variables $(M,V)$. On the other hand, in view of \eqref{eq:Einstein1}--\eqref{eq:Einstein4} and \eqref{Tab},
the Einstein equations are equivalent to the three ordinary differential equations  
\begin{align}
b_r & = 4 \pi  r b M \, (1 + k^2),
\label{eq:Einstein1-new} 
\\
a_r & = 4 \pi r M \, (1+k^2) \left( 2 K^2 V - a \right) + \frac{1-a}{r}, 
\label{eq:Einstein2-new}
\\
a_v & = 2 \pi r b M (1+k^2) \left( a^2 - 4 V^2 \right), 
\label{eq:Einstein3-new} 
\end{align}
and the partial differential equation
\begin{align}
 \left( \frac{b_r}{b} \right)_v +  \frac{1}{2} (a_r b)_r + (a b_r)_r 
& = - \frac{1}{r} (a b)_r - 16 \pi  b M\, k^2 V. 
\label{eq:Einstein4-new} 
\end{align}  
Observe that we have here reformulated \eqref{eq:Einstein4} so that its left-hand side has a meaning in the sense of distributions
(cf.~Definition~\ref{def:weaksolutions}, below).

Our reduction of the Einstein-Euler system which is closed related by our choice of normalized fluid variables now 
suggest a way to integrate out the equations satisfied by the metric coefficients $a,b$.
For the function $b$, taking into account our choice of fluid variables $M,V$, we have  
$$
 \frac{b_r}{b} = 4 \pi  r M \, (1 + k^2),
$$
which suggests us to recover the function $b$ from the fluid density $M$ via the integral formula
\bel{b-new}
b(v,r) = \exp \Big( 4\pi (1+k^2) \int_0^r M(v,r') \, r' dr' \Big). 
\ee 
This formula makes sense provided the function $r \, M$ is locally integrable on $[0, +\infty)$.
Interestingly, this formula differs from the one presented at the end of Section~\ref{sec:21} and relies on a physically more
consistent integrability assumption on $M$.   

Returning to the function $a$, we can rely on \eqref{eq:Einstein2-new} and obtain
$$
(r \, (a- 1) )_r + (r \, (a-1) ) \, 4 \pi r M (1+k^2) =  - 4 \pi r^2 M (1+k^2) \big( 2 K^2 \, |V| + 1 \big), 
$$
which after integration yields us  
\be
\label{a-new}
a(v,r) = 1 - \frac{4 \pi (1+k^2)}{r} \int_0^r \frac{b(v,r')}{b(v,r)} \, M(v,r') \big( 2 K^2 \, |V(v,r')| + 1 \big) \,  {r'}^2 dr'. 
\ee
By replacing the function $b$ in the above formula by its expression \eqref{b-new}, we conclude that the
spacetime geometry is determined once we know its matter content.

In the rest of this paper,  we regard the Euler equations \eqref{eq:Euler}--\eqref{S1S2}  as a first-order hyperbolic system with 
non-constant coefficients which depend on certain integral expressions of the unknowns $(M,V)$
given by \eqref{b-new} and \eqref{a-new}. 
We now formulate the initial value problem when data are prescribed on an outgoing light cone.
For this problem, we introduce a definition of solutions within
the space  $BV$ 
of function with bounded variation in $r$. 
We denote by $ L^\infty(BV)$ the space of functions depending also on $v$ whose total variation is bounded in $v$.
Motivated by the standard regularity properties of hyperbolic systems~\cite{Glimm,LeFloch-book},
we also assume that solutions are locally Lipschitz continuous in the time variable, specifically in $\Lip(L^1)$.
The low regularity imposed now will be fulfilled by the solutions to the initial value problem constructed in this paper. Observe
that no regularity is required on the first-order derivative $b_v$, which is consistent with the fact no such term arises in the
Einstein equations  
\eqref{eq:Einstein1-new}--\eqref{eq:Einstein3-new}. 

\begin{definition} 
\label{def:weaksolutions} 
A {\rm  spherically symmetric, Einstein-Euler spacetime with bounded variation} in generalized Eddington-Finkelstein coordinates
$$
g = -a b^2 \, dv^2 + 2 b \, dv dr + r^2 \, \big( d\theta^2 + \sin^2\theta d\varphi^2 \big) 
$$ 
 is determined by two metric coefficients $a,b$ 
and the two normalized fluid variables  
$$
M = b^2 \mu \, u^0 u^0 \in (0, +\infty), 
\qquad
V = \frac{u^1}{b \, u^0} - \frac{a}{2} \in (-\infty, 0), 
$$
all of these independent variables being defined for $v \in I:= [v_0, v_*]$ and $r \in J:= [0,r_0)$,  
and satisfying the regularity conditions 
$$
a_v, \, r a_r,  \, b_r,  \, M,  \, V \in L^\infty(I, BV(J)) \cap \Lip(I, L^1(J)),
$$
together with the following conditions: 
\bei

\item[1.] The (first three) Einstein equations \eqref{eq:Einstein1-new}--\eqref{eq:Einstein3-new}  are satisfied as
equalities between functions with bounded variation.  

\item[2.] The (fourth) Einstein equation \eqref{eq:Einstein4-new}  holds in the sense of distributions, as an equality between locally
bounded measures.  
 
\item[3.] The Euler equations \eqref{eq:Euler} (with the notation \eqref{UF}--\eqref{S1S2}) hold in the sense of distributions, as
equalities between locally bounded measures. 

\item[4.] The following regularity condition holds at the center: 
$$
\lim_{r \to 0} a(v,r) = \lim_{r \to 0} b(v,r) = 1, \qquad v \in I. 
$$
\eei
\end{definition} 

Under the integrability conditions in Definition~\ref{def:weaksolutions}, the formulas \eqref{b-new} and \eqref{a-new} make sense,
and determine both metric coefficients. 
Next, given $r_0 >0$, we formulate the initial value problem by imposing initial data for the (normalized) matter variables on the
hypersurface $v=v_0$, that is,  
\bel{eq:dataEuler}
M(v_0, r) = \Mb(r), \qquad V(v_0,r) = \Vb(r), \qquad r \in J, 
\ee
where $\Mb>0$ and $\Vb<0$ are functions with bounded variation.
The metric coefficients $\ab, \bb$ on the initial hypersurface are determined explicitly from $\Mb, \Vb$ by writing \eqref{b-new}
and \eqref{a-new} with $v=v_0$, and satisfy the regularity and decay conditions required in Definition~\ref{def:weaksolutions}:  
\bel{eq:initial-ab}
\aligned
& r \partial_r \ab, \, \partial_r \bb, \in BV(J), 
\\
& \lim_{r\to0} \ab(r) = \lim_{r\to0} \bb(r) = 1. 
\endaligned
\ee


\subsection{The reduced Einstein-Euler system}
\label{sec:reducedsystem}

It is convenient to analyze in this paper only a subset of the Einstein-Euler system, after observing that the remaining equations
are then automatically satisfied. We refer to \eqref{eq:Einstein}--\eqref{euler} as the {\it full system\/}, while the reduced system
consists of only four equations, obtained by keeping \eqref{eq:Einstein} with $(\alpha, \beta)=$ $(0,0) $ or $(1,0)$, together
with \eqref{euler} with $\alpha, \beta \in \big\{0,1\big\}$, only. 

\begin{definition} 
The first-order system \eqref{eq:Euler}--\eqref{S1S2} together with the metric expressions \eqref{b-new} and \eqref{a-new} 
is refered to as the {\rm reduced Einstein-Euler system.}
\end{definition}

The equations that are not taken into account in our main analysis can be recovered without further
initial data or regularity assumptions, as now stated.  

\begin{proposition}[From the reduced system to the full system] 
\label{prop:400}
Any solution $(M,V, a, b)$ to the reduced Einstein-Euler system is actually a solution to the full system of Einstein-Euler equations,
that is: 
if the two fluid equations \eqref{eq:Euler}--\eqref{S1S2} and the two metric equations
\eqref{eq:Einstein1-new}--\eqref{eq:Einstein2-new} hold true, then under the regularity and decay assumptions in
Definition~\ref{def:weaksolutions}, it then follows that the equations \eqref{eq:Einstein3-new} (satisfied as an equality between
BV functions) and \eqref{eq:Einstein4-new} (satisfied in the distributional sense)
also hold. 
\end{proposition}

The rest of this section is devoted to the proof of this result. However, the specific form of the energy momentum tensor is irrelevant for the present argument, and it is more convenient to treat general matter models.
We assume sufficient regularity first, so that all identities under consideration make sense between continuous functions, say, and
we postpone the discussion of the low regularity issue to the proof of Proposition~\ref{prop:400} below.  

Recall also that we impose spherical symmetry throughout this paper and the Eddington-Finkelstein coordinates \eqref{eq:metric} are
used. Some redundancies in the formulation of the full Einstein-Euler equations arise as a consequence of our assumption of spherical
symmetry. We use the following notation to simplify our calculations (recalling that $b>0$):
$$
B := \log b, \qquad X := a_r b + 2 a b_r = \frac{1}{b} (ab^2)_r.  
$$
From the discussion made in this section and in view of the expression of the Einstein tensor $G^{\alpha\beta}$ computed in
Section~\ref{sec:21} in Eddington-Finkelstein coordinates, we can state the Einstein equations as follows. Recall that $T^{\alpha\beta}$ is
always assumed to be symmetric. 

The Einstein equations $G^{\alpha\beta} = 8 \pi T^{\alpha\beta}$ 
are equivalent to the four (partial) differential equations 
\begin{subequations}
\begin{align}
 B_r & = 4 \pi \, r b^2 \,  T^{00}, 
\label{eq:T00} 
\\
 r X + b (a-1) & = 8 \pi \, r^2 b^2 \,  T^{01}, 
\label{eq:T01} 
\\
ab (a-1) + r (a X - a_v) & = 8 \pi \, r^2 b \, T^{11}, 
\label{eq:T11} 
\\
2 (ab)_r + r (X + 2 B_v)_r & = 16 \pi \, r^3 b \, T^{22}, 
\label{eq:T22}
\end{align}
\end{subequations}
supplemented with the following conditions 
\begin{subequations}
\label{eq:T22T33}
\begin{align}
T^{02} & = T^{03} = T^{12} = T^{13} = T^{23} = 0, 
\\
T^{22} & = (\sin \theta)^2 \, T^{33}, 
\end{align}
\end{subequations} 
which are regarded as restrictions on the energy momentum tensor. Observe that \eqref{eq:T11} may also be replaced (thanks to
\eqref{eq:T01}) by the simpler equation 
\be
\label{av}
 a_v = 8 \pi r b \, \big( ab T^{01} - T^{11} \big).
\ee

\begin{definition} An energy momentum tensor $ T^{\alpha\beta}$ is said to be {\it compatible with spherical symmetry\/} (with respect
to the generalized Eddington-Finkelstein coordinates \eqref{eq:metric}) if the conditions \eqref{eq:T22T33} hold. 
\end{definition}
 
For instance, the energy momentum tensor \eqref{eq:EM} of perfect fluids is compatible with spherical symmetry, {\it provided\/} the
velocity vector $u^\alpha$ is assumed to have vanishing components $\alpha=2,3$, that is, $u^2 = u^3 = 0$.  

\begin{lemma}\label{lem:E2} 
If the matter tensor is compatible with spherical symmetry, then the components $\beta=2,3$ of the matter equations, that is, 
$$
\nabla_\alpha T^{\alpha\beta} = 0, \qquad \quad \beta=2,3, 
$$
are also satisfied. 
\end{lemma}

\begin{proof}  In view of our assumption of radial symmetry, the partial derivatives in $\theta, \varphi$ are zero. From the
expressions of the Christoffel symbols~\eqref{Christoffel} and the conditions $T^{02}=T^{12}=0$, we obtain 
$$
\aligned
\nabla_0  \T{02}{} = \; & \T{02}{,0} + \G{0}{0}{0} T^{02} = 0,
\\
\nabla_1\T{12}{} = \; & \T{12}{,1} + \G{1}{0}{1} T^{02} + \G{1}{1}{1} T^{12} + \G{2}{1}{2} T^{12} = 0,
\\
\nabla_2\T{22}{} = \; & 2 \G{2}{1}{2} T^{12} = 0,
\\ 
\nabla_3\T{32}{} = \; & \G{3}{1}{3} T^{12} + \G{3}{2}{3} T^{22} + \G{2}{3}{3} T^{33} = \cot \theta \, T^{22} - \sin \theta \, \cos \theta \, T^{33}.
\endaligned
$$
 Thus, the $2$-component of the matter equations, that is $\nabla_\alpha T^{\alpha2} = 0$, holds when \eqref{eq:T22T33} hold.
 The assumptions $T^{03}=T^{13}=T^{23}$ imply that $\nabla_\alpha T^{\alpha3}=0$.
\end{proof}

It thus remains to be checked that solving the first two Einstein equations and the first two matter equations suffices to recover the
third and fourth Einstein equations.

\begin{lemma}
\label{lem:E0}
If the first two Einstein equations \eqref{eq:T00}--\eqref{eq:T01} hold 
and the matter tensor is compatible with spherical symmetry, then    
\eqref{eq:T22} holds. 
\end{lemma}

\begin{proof} Using \eqref{Christoffel} and the condition $T^{02}=0$, the $0$-component of the matter equations reads 
\begin{align*} 
 0 = \nabla_\alpha \T{\alpha 0}{} = \; 
& \T{00}{,0} + \T{10}{,1} + \left( 2 B_v + \frac{X}{2} \right) T^{00}
+ \left( B_r + \frac{2}{r} \right) T^{01} - \frac{2r}{b} T^{22},
\end{align*}
 with
\begin{align*}
 \T{00}{,0} &= \frac{1}{4 \pi r} \left( \frac{B_r}{b^2} \right)_v = \frac{B_{rv} - 2 B_r B_v}{4 \pi r b^2} = \frac{B_{rv}}{4 \pi r b^2} - 2 B_v T^{00},
 \\
 \T{01}{,1} &= \frac{1}{8 \pi} \left( \frac{rX + b (a-1)}{r^2 b^2} \right)_r 
 = - 2 \left( B_r + \frac{1}{r} \right) T^{01}  + \frac{X + rX_r + (ab)_r}{8 \pi r^2 b^2} - \frac{B_r}{8 \pi r^2b}, 
\end{align*}
 derived from the Einstein equations \eqref{eq:T00} and \eqref{eq:T01}.
 Therefore, again by the first two Einstein equations, 
\begin{align*}
 \frac{2r}{b} T^{22} &= \frac{B_{rv}}{4 \pi r b^2} - B_r T^{01} + \frac{X + rX_r + (ab)_r}{8 \pi r^2 b^2} - \frac{B_r}{8 \pi r^2b}
 + \frac{X}{2} T^{00} \\
 &= \frac{1}{8 \pi r^2 b^2} \left( r (2 B_{rv} + X_r) + X - ab B_r + (ab)_r \right) \\
 &= \frac{1}{8 \pi r^2 b^2} \left( r (2 B_v + X)_r + 2 (ab)_r \right). \qedhere
\end{align*}
\end{proof}

\begin{lemma}
\label{lem:E1} 
Suppose that the three Einstein equations \eqref{eq:T00}, \eqref{eq:T01} and \eqref{eq:T22} hold, and the matter tensor is compatible
with spherical symmetry. Then the Einstein equation \eqref{eq:T11} holds provided the regularity condition \eqref{eq:center} holds at
the center together with the additional condition  
\bel{T11condition}
\lim_{r \to 0} r^2 \big( G^{11}   - 8 \pi \, T^{11} \big) = 0. 
\ee   
\end{lemma}

\begin{proof} By \eqref{Christoffel} and the condition $T^{12}=0$, the $1$-component of the matter equations, i.e.\
$\nabla_\alpha T^{\alpha1} = 0$, reads
\begin{align*}
 0 = \T{01}{,0} + \T{11}{,1} + ( aX - a_v ) \frac{b}{2} T^{00} + ( B_v - X ) T^{01}  + 2 \left( B_r + \frac{1}{r} \right) T^{11}
 - 2 r a T^{22}. 
\end{align*}
 This is a linear ordinary differential equation (for the unknown $T^{11}$) in the variable $r$
\bel{eq:E1}
 \T{11}{,1} + 2 \left( B_r + \frac{1}{r} \right) T^{11} = R (a,b),
\ee
 where the right--hand side is explicitly given by the Einstein equations \eqref{eq:T00}, \eqref{eq:T01}, and \eqref{eq:T22}, that is, 
$$
 R(a,b) := - \T{01}{,0} + ( a_v - aX  ) \frac{b}{2} T^{00} + ( X - B_v ) T^{01} + 2 r a T^{22}.
$$ 

It remains to be checked that \eqref{eq:T11} is the only solution to \eqref{eq:E1}. Observe that the solutions to the
corresponding homogeneous equation of \eqref{eq:E1},
$$
\T{11}{,1} = - (\log (rb)^2)_r \, T^{11},
$$
are multiples of $\frac{1}{r^2 b^2}$. Moreover it is clear that $\frac{G^{11}}{8 \pi}$ is a particular solution to \eqref{eq:E1}: In
this situation all non--trivial Einstein equations $G^{\alpha\beta} = 8 \pi T^{\alpha\beta}$ are satisfied by assumption. Moreover, it
follows from the second Bianchi identity that $G^{\alpha\beta}$ is divergence-free, i.e.\ $\nabla_\alpha G^{\alpha\beta} = 0$. Thus,
in particular, $8 \pi \, \nabla_\alpha T^{\alpha1} = \nabla_\alpha G^{\alpha1} = 0$, which is just the equation \eqref{eq:E1} above
(after also using $T^{12} = 0$).

Thus, the general solutions to \eqref{eq:E1} are of the form
$$
T^{11} = \frac{1}{8 \pi r^2 b^2} \left( C + ab^2(a-1) + raa_rb^2 +  2ra^2 b b_r -r a_v b \right), \qquad C \in \RR.
$$ 
The limiting behavior \eqref{T11condition} and the regularity conditions \eqref{eq:center} at the center imply that
$$
0 = \lim_{r \to 0 +} r^2 \big( G^{11} - 8 \pi T^{11} \big) = \lim_{r \to 0 +} - \frac{C}{b^2} = - C.
$$
Therefore, the unique solution to \eqref{eq:E1} is indeed \eqref{eq:T11}.
\end{proof}

\begin{proof}[Proof of Proposition~\ref{prop:400}] Consider now a solution $M,V,a,b$ to the reduced Einstein-Euler system 
satisfying the regularity and decay conditions in Definition~\ref{def:weaksolutions}. Then, thanks to Lemma~\ref{lem:E0}, the Euler
equation $\nabla_\alpha T^{\alpha0} = 0$ together with the Einstein equations \eqref{eq:Einstein1} and \eqref{eq:Einstein2}, as well as 
the compatibility assumption \eqref{eq:musthold} imply \eqref{eq:Einstein4}. Those equations together with the additional
assumptions $\nabla_\alpha T^{\alpha1} = 0$ yield \eqref{eq:Einstein3} (by Lemma~\ref{lem:E1}),
since \eqref{T11condition} holds: From \eqref{Gab}, \eqref{Tab} and the first two Einstein equations \eqref{eq:Einstein1} and
\eqref{eq:Einstein2} we deduce that 
$$
 r^2 \, \big( G^{11} - 8 \pi T^{11} \big) = 2 \pi r^2 (1+k^2) M (a^2 - 4V^2) - \frac{r a_v}{b}. 
$$
Condition \eqref{T11condition} of Lemma~\ref{lem:E1} is thus satisfied due to the behavior of $a$ and $b$ at the center and the $BV$
regularity assumed for the functions involved, both specified in Definition~\ref{def:weaksolutions}. 

Finally, by Lemma \ref{lem:E2}, the assumption \eqref{eq:musthold} implies that $\nabla_\alpha T^{\alpha\beta} = 0$ for $\beta=2,3$.

It remains to discuss the regularity issue. We simply need to observe that all calculations in Lemmas \ref{lem:E2}, \ref{lem:E0} and
\ref{lem:E1} are valid even for solutions with low regularity, provided all equations under consideration are understood in the
distributional sense, along the lines of Definition~\ref{def:weaksolutions}. Most importantly, the divergence-like form of the Euler
equations is used without multiplication by an auxiliary factor with low regularity, which would not be allowed in the distribution
sense. 
\end{proof}


\section{The class of static Einstein-Euler spacetimes}
\label{sec:5}

\subsection{A reduced formulation} 

We now consider the Euler system \eqref{eq:Euler}--\eqref{S1S2} and focus on static solutions $(M,V)$ satisfying, by definition,
$\del_v M = \del_v V = 0$, thus 
\bel{eq:Euler-static}
 \del_r F(U, a, b) = S(U, a, b). 
\ee 
Using a different choice of coordinates, Rendall and Schmidt \cite[Theorem 2]{rendallschmidt} and Ramming and Rein \cite{RRein}
constructed radially symmetric static solutions by prescribing the mass density at the center of symmetry $r=0$.  We revisit here
their conclusions in our context of  Eddington-Finkelstein coordinates.

For spacetimes that need not be static, it is convenient to introduce the so-called {\it Hawking mass\/} $m=m(v,r)$, defined by 
$$
a = 1 - \frac{2m}{r},
$$
by analogy with the expression of the Schwarzschild metric.  From \eqref{eq:Einstein2-new}--\eqref{eq:Einstein3-new}, we obtain
\bel{eq:mr} 
\aligned
 m_r & =  2 \pi r^2 M (1+k^2) \left( 1 - \frac{2m}{r} + 2 K^2 \, |V| \right), 
\\
 m_v & = \pi r^2 b M (1+k^2) \left( 4 V^2 - \Big( 1 - \frac{2m}{r} \Big)^2 \right).  
\endaligned
\ee
Observe that the function $r \mapsto m(v,r)$ is increasing, provided $\frac{a}{2V} < K^2$; the latter condition does 
hold if, for instance, $a$ is positive but remains true even for negative values of $a$ (corresponding to the trapped region) at
least if the normalized velocity is sufficient large.
 On the other hand, the function $v \mapsto m(v,r)$ is decreasing provided the ratio $\frac{|a|}{2 |V|}$ is greater than $1$.

In view of \eqref{UF}, the condition $\del_v U = 0$ implies that $M_v=0$ and, thus, by \eqref{b-new} and \eqref{a-new} we obtain that
$b_v$ and $a_v$ (resp.\ $m_v$) vanish. By the ``third'' Einstein equation \eqref{eq:Einstein3-new}, we then have 
\be
a^2 = 4 V^2. 
\ee
Henceforth, the static equations may be simplified by keeping in mind that near the center, due to the regularity assumption
\eqref{eq:center}, $a$ should be positive while $V<0$. Hence, we find 
\bel{eq:555} 
 V = - \frac{a}{2} = \frac{m}{r} - \frac{1}{2}.
\ee
Returning to the definitions of $M,V$, we find that $u^1 = 0$ and $ab^2 (u^0)^2 = 1$, which implies
\bel{eq:556} 
a M = \left( 1 - \frac{2m}{r} \right) M = \mu.
\ee

The static Einstein-Euler equations can be expressed in terms of the local mass $m$ and the fluid density $\mu$, as follows. 

\begin{lemma} \label{static:formulation}
All solution to the static Euler equations~\eqref{eq:Euler-static} having $a > 0$ satisfy a system of first-order ordinary differential
equations in $m,\mu$ defined for $r \in (0,+\infty)$: 
\bel{mmu}
\aligned
m_r &= 4 \pi r^2 \mu, 
\\
\mu_r &= -  \frac{(1+k^2)\mu}{r-2m} \Big( 4 \pi r^2 \mu + \frac{m}{rk^2} \Big) \leq 0.
\endaligned
\ee
Moreover, the functions $V,M, a$ are recovered from  \eqref{eq:555}--\eqref{eq:556}, while the coefficient $b$ is given by 
\bel{b-new-3}
b(r)  = \exp \Big(4\pi (1+k^2) \int_0^r \frac{r'^2 \mu(r')}{r'-2m(r')} dr' \Big), 
\ee 
provided $\frac{r^2 \mu}{r-2m}$ is integrable at the center. 
\end{lemma}

\begin{proof} If $a\geq 0$, then $V$ is directly related to $m$ as above and \eqref{eq:mr} yields the first equation
 for $m_r$. The second equation in the system is derived from the first Euler equation in 
$\del_r F(U, a, b) = S(U, a, b)$, which reads (in terms of $m,M$)
$$
(1-K^2) \del_r \left( bM \left( 1 - \frac{2m}{r} \right) \right) = - \frac{2}{r^2} b m M.
$$
Using the first Einstein equation \eqref{eq:Einstein1-new} to replace $b_r$, the previous equation to replace $m_r$ and division by
$b>0$ yields  
$$
M_r = 4 \pi (1-k^2) r M^2 - \frac{1+3k^2}{k^2} \frac{mM}{r (r-2m)}
$$
and we only need to use \eqref{eq:556} and replace $M$ by $\mu$.  
\end{proof}

The following integral identity will also be useful. 

\begin{lemma} 
\label{remark:r-2m}
Given a solution to \eqref{mmu}, the function $z(r) := r - 2m(r)$ satisfies the differential equation
$
z_r = 1 - 8 \pi r z M 
$
with initial value $z(r_0)= r_0-2m(r_0)$, 
and is thus given by
$$
z(r) = z(r_0) e^{- 8 \pi \int_{r_0}^r t M(t) \, dt} +
\int_{r_0}^r e^{- 8 \pi \int_s^r {t M(t)} \, dt} \, ds.
$$
\end{lemma}

From this lemma, it follows that if $a(r_0) \geq 0$, then $z(r_0) = r_0 a(r_0) \geq 0$ and $z(r)$ therefore positive for all $r>r_0$.
This, in turn, implies that $a>0$ for $r>r_0$. This calculation thus also shows that an initial condition $a(0) \geq 0$ is sufficient
to satisfy the requirement of the function $a$ being positive on $(0,+\infty)$ needed for Lemma~\ref{static:formulation}.


\subsection{Existence of static solutions}

We prescribe initial conditions at the center $r = 0$, specifically 
$\mu_0 > 0$ and $m_0 = 0.$
The condition on the initial value on $m$ is consistent with $m$ being non-negative, 
however, it remains to be checked whether the second equation in \eqref{mmu} is well-defined.  

\begin{theorem}
\label{thm-static}
Fix any initial conditions $m_0 = 0$ and $\mu_0 > 0$ at the center. Then, there exists a unique global 
 solution $(m,\mu)$ to the static Einstein-Euler system \eqref{mmu} with prescribed values  
$$
 \lim_{r\to0} m(r) = 0, \qquad \lim_{r\to0} \mu(r) = \mu_0.
$$
 Moreover, the functions $m, \mu$ are smooth and positive on $(0,+\infty)$ and we have 
$\lim_{r\to +\infty} \mu (r) = 0$. These static solutions $(M,V,a,b)$ satisfy the low regularity conditions specified in
Definition~\ref{def:weaksolutions}, and the geometric coefficient $b$ can be recovered using \eqref{b-new-3} and satisfies
$\lim_{r\to0} b(r) = 1$.
\end{theorem}

Two remarks are in order at this juncture: 
\bei 
\item[1.] Observe that, since $\lim_{r\to0} \left(1-\frac{2m}{r}\right) = 1 - \lim_{r\to0} (8 \pi r^2 \mu)= 1$ by L'H\^{o}pital's rule,
the initial values $M_0 = \mu_0$ coincide, that is, the initial value for the fluid density $\mu$ is the same as for the fluid
variable $M$. In the proof, we switch between $M$ and $\mu$ and work with whatever is more convenient. 
\item[2.] It would be interesting to further determine the asymptotic behavior at infinity, especially 
whether $\lim_{r\to+\infty} m(r)$ exists or, equivalently, $r^2 \mu$ is globally integrable. Step 3 in the proof below implies that
$r^2 \mu$ is bounded by some constant and one would expect the stronger statement $\lim_{r\to+\infty} r^2 \mu =0$ which would also 
imply $\lim_{r\to+\infty} m_r(r) = 0$.
\eei 
 
\begin{proof} { \bf Step 1. Local existence near the center $r=0$.}  We introduce a new variable
$n= \frac{2m}{r-2m}$
with initial value 
$$
n_0 = \lim_{r\to0} \frac{2m}{r-2m} = \lim_{r\to0} \frac{2m_r}{1-2m_r} = 0
$$
and we rewrite the equations in terms of $\nh, \Mh$ by setting
$$
 n = r^\alpha \nh, \quad M = M_0 + r^\beta \Mh,
$$
where $M_0 = \mu_0$  
is the initial value of $M$ at the center, and $\alpha,\beta$ remain to be determined.
 System \eqref{mmu} then reads
\begin{align*}
 \nh_r = & 8 \pi r^{1+\beta-\alpha} \Mh + \left( 8\pi r M_0 - 1 - \frac{\alpha}{r} \right) \nh + 8 \pi r^{1+\beta} \nh \Mh
   - r^\alpha \nh^2 + 8 \pi r^{1-\alpha} M_0, \\
 \Mh_r = &\left( 8 \pi (1-k^2) r M_0 - \frac{\beta}{r} \right) \Mh + 4 \pi (1-k^2) r^{1+\beta} \Mh^2
   - \frac{1+3k^2}{2k^2} r^{\alpha -1} \nh \Mh \\
 &- \frac{1-3k^2}{2k^2} r^{\alpha-\beta-1} M_0 \nh + 4 \pi (1-k^2) r^{1-\beta} M_0^2.
\end{align*}
 For $0< \beta \leq \alpha \leq 2$ only, this system is of the form $$
 r \del_r f + N f = r G (r,f(r)) + g(r),
$$
 with $f=(\nh,\Mh)$, $N$ linear with positive eigenvalues and $g$ and $G$ are smooth on $[0,+\infty)$ and
 $[0,+\infty)\times \mathbb{R}^2$, respectively. In particular, for the "maximal" choice $\alpha = \beta = 2$, the system is of the
 form
\begin{align*}
& r \begin{pmatrix}
    \nh \\ \Mh
   \end{pmatrix}_r
  + \begin{pmatrix}
    2 & 0 \\ \frac{1+3k^2}{2k^2} M_0 & 2
   \end{pmatrix}
 \begin{pmatrix}
  \nh \\ \Mh
 \end{pmatrix}
 \\
& = r \begin{pmatrix}
8 \pi  r \Mh + (8 \pi r M_0 -1) \nh + 8 \pi r^3 \nh \Mh - r^2 \nh^2 \\
8 \pi (1-k^2) r M_0 + 4 \pi (1-k^2) r^3 \Mh^2 - \frac{1+3k^2}{2k^2} r \nh \Mh
    \end{pmatrix}
 + \begin{pmatrix}
   8 \pi M_0 \\ 4 \pi (1-k^2) M_0^2
  \end{pmatrix}.
\end{align*}
 Thus, by \cite[Theorem 1]{rendallschmidt} there exists an interval $[0,R)$ and a unique bounded $\mathcal{C}^1$ 
solution $f=(\nh,\Mh)$ on $(0,R)$ that extends to a $\mathcal{C}^\infty$ solution on $[0,R)$.
Thus, $\frac{2m}{r^2(r-2m)}$ and $\frac{M-M_0}{r^2} = \frac{\mu - \mu_0}{r(r-2m)}$ are bounded near $0$. The solution can be extended
in a unique way as long as it does not blow up or reach zero. It remains to be shown that global existence (and hence uniqueness) is
indeed given, that the fluid has infinite radius and that the decay is as desired.

\

\noindent{ \bf Step 2. Infinite extension of the solution.}
Whenever $\mu > 0$, the first equation in \eqref{mmu} implies that $m_r>0$ and hence $m(r)>m_0 = 0$ for $r>0$. The second equation in
\eqref{mmu} then gives $\mu_r<0$, thus $\mu$ is bounded above by the initial value $\mu_0$. This forces $m$ to be bounded by
\bel{mineq}
 +\infty > \frac{4 \pi}{3} r^3 \mu_0 \geq m(r) = 4 \pi \int_0^r s^2 \mu(s) ds \geq \frac{4\pi}{3} r^3 \mu(r) > 0.
\ee
Consequently, if we show that $\mu>0$ globally, global existence follows. The proof of $\mu>0$ globally is the content of this second
step.

Since $\mu_0>0$, it is clear that $\mu>0$ initially on some interval $[0,r_1)$. Suppose, contrary to our claim, that $\mu(r_1)=0$.
Together with \eqref{mineq}, the decay rate obtained for $\widehat n$ on $[0,R)$ in Step 1, implies that for some constant $C_1>0$,
\[
 \frac{\mu_r}{\mu} \geq - \frac{(1+k^2)(3k^2+1)}{k^2} \frac{m}{r(r-2m)} \geq - C_1 r 
\]
on $[0,R) \cap [0,r_1)$. If $R< r_1$ then this estimate can to be extended to $[0,r_1)$:
We have $r-2m \geq z(R)>0$ due to Lemma~\ref{remark:r-2m}, which together with $m(r) \leq \frac{4\pi}{3} r^3 \mu_0$ from \eqref{mineq}
and for $C_2 := \frac{(1+k^2)(3k^2+1)}{k^2} \frac{4 \pi r_1 \mu_0}{z(R)} \geq \frac{(1+k^2)(3k^2+1)}{k^2} \frac{m}{r^2 (r-2m)}$ yields
\[
 \frac{\mu_r}{\mu} \geq - \frac{(1+k^2)(3k^2+1)}{k^2} \frac{m}{r(r-2m)} \geq - C_2 r.
\]
With $C:= \max (C_1,C_2)>0$ it thus follows that
\[
 (\log \mu)_r \geq - C r, \qquad r \in [0,r_1],
\]
and hence $\mu(r) \geq \mu_0 e^{-Cr}$. This contradicts $\mu(r_1)=0$, and hence forces the solution to have an infinite radius.

\

\noindent{ \bf Step 3. Decay properties.}  We prove next that $\mu \to 0$ as $r \to +\infty$. By Step 2, $\mu$ is monotonically
decreasing (since $\mu_r<0$) and bounded from below by $0$. Therefore $\lim_{r\to +\infty} \mu = \mu_\infty \geq 0$ exists. By
\eqref{mineq} and the fact that $\frac{r}{r-2m} \geq 1$ one obtains for some constant $C_3>0$ that
$\mu_r \leq - (1+k^2) \frac{\mu}{r-2m} \left( 4 \pi r^2 \mu + \frac{4 \pi}{3k^2} r^2 \mu \right) \leq - C_3 r \mu^2$.
Hence
\[
 \left( \frac{1}{\mu} \right)_r \geq C_3 r,
\]
and integration yields $0 < \mu(r) \leq \frac{2 \mu_0}{C_3 r^2 \mu_0 + 2}$ which tends to $0$ for $r \to +\infty$.

\

\noindent{ \bf Step 4. Regularity properties.}
It is easy to check that the following expansions hold at the center as $r \to 0$:
\begin{align*}
 \mu (r) & = \mu_0 - 2 \pi \frac{(1+k^2)(1+3k^2)}{3k^2} \mu_0^2 r^2 + O(r^3), \\
 m(r) &= \frac{4\pi}{3} \mu_0 r^3 + O(r^4).
\end{align*}
Therefore, as $r \to 0$, we have 
\begin{align*}
 a(r) &= 1 - \frac{8\pi}{3} \mu r^2 + O(r^3), \\
 M(r) &= \frac{\mu}{a} = \mu_0 - \frac{2\pi}{3} \frac{3k^4 + k^2 + 1}{3k^2} \mu_0^2 r^2 + O(r^3), \\
 b(r) &= e^{4 \pi (1+k^2) \int_0^r M(s) s \, ds} = 1 + 2 \pi (1+k^2) \mu_0 r^2 + O(r^3),
\end{align*}
 which proves that $a, a_r, M, b, b_r$ and $V = - \frac{a}{2}$ are of bounded variation near the center and
 $\lim_{r\to0} a(r) = \lim_{r \to0} b(r) = 1$. According to Lemma~\ref{remark:r-2m}, $a$ is positive everywhere, hence $V<0$.
\end{proof}

 
\section{Euler system on a uniform Eddington-Finkelstein background}
\label{sec:4} 

\subsection{Algebraic properties}

In this section, we analyze the principal part of the Euler system \eqref{eq:Euler}--\eqref{S1S2}, which we now define by assuming
that the metric coefficients $a,b$ are prescribed functions and, in fact, are {\it constants\/}, and, in addition, by suppressing the
source--terms therein. In other words, in this section we consider the system of two equations 
\bel{eq:Eulerhom}
\aligned
& \del_v U + \del_r F(U) = 0, 
\\
& U  = M \begin{pmatrix} 1 \\ \frac{a}{2} + K^2 V \end{pmatrix}, 
\qquad \qquad
F(U)  = F(U, a, b) = bM \begin{pmatrix} \frac{a}{2} + K^2 V \\ \frac{a^2}{4} + K^2 a V + V^2 \end{pmatrix},
\endaligned
\ee
in which $M>0$ and $V<0$ are the unknown functions, while $a \in \RR$ and $b>0$ are constants 
and $K^2 = \frac{1-k^2}{1+k^2} \in (0,1)$ is given. We begin with some basic properties about the Jacobian matrix 
$D_U F(U)$.

\begin{proposition}[Algebraic structure of the fluid equations]
\label{prop:31}
The Euler system \eqref{eq:Eulerhom} on a uniform Eddington-Finkelstein background is a strictly hyperbolic system of conservation
laws, with eigenvalues
\bel{eigenvalues}
\lambda_1 := b \left( \frac{1+k}{1-k} V + \frac{a}{2} \right), 
\qquad \quad
\lambda_2 := b \left( \frac{1-k}{1+k} V + \frac{a}{2} \right),
\ee
and right--eigenvectors (which can be normalized to be) 
\bel{eigenvectors}
r_1 := - \begin{pmatrix} 1 \\ \frac{1+k}{1-k} V + \frac{a}{2} \end{pmatrix},
\qquad\quad
r_2 := \begin{pmatrix} 1 \\ \frac{1-k}{1+k} V + \frac{a}{2} \end{pmatrix}.
\ee
Moreover, each characteristic field associated with \eqref{eq:Eulerhom} is genuinely nonlinear, with 
$$
\nabla \lambda_1 \cdot r_1 > \nabla \lambda_2 \cdot r_2 > 0.
$$
\end{proposition}
 
Observe that the eigenvalues and eigenvectors are independent of $M$, and are {\it linear functions\/} in $V$. This property is not met
in the standard formulation of the Euler equations in Minkowski spacetime, and is a consequence of our choice of 
coordinates. Furthermore, from \eqref{eigenvalues}, we deduce that the sign of the eigenvalues $\lambda_1,\lambda_2$ is as follows
(recalling that $b$ is positive): 
\bel{eq:signlambda}
\aligned
& \lambda_1 < \lambda_2 < 0 && \text{ if and only if } && V < \min \Big( 0, - \frac{1+k}{1-k} \frac{a}{2} \Big),
\\
& \lambda_1 < 0 < \lambda_2 && \text{ if and only if } &&- \frac{1+k}{1-k} \frac{a}{2} < V < \min \Big (0, - \frac{1-k}{1+k} \frac{a}{2} \Big),
\\
& 0 < \lambda_1 < \lambda_2 && \text{ if and only if } &&- \frac{1-k}{1+k} \frac{a}{2} < V < 0.
\endaligned 
\ee
Since the velocity $V$ is always negative, we can also formulate these conditions in terms of $a$, as follows: 
\bel{eq:signlambda-with-a}
\aligned
& \lambda_1 < \lambda_2 < 0 && \text{ if and only if } && \frac{a}{2|V|} < \frac{1-k}{1+k}, 
\\
& \lambda_1 < 0 < \lambda_2 && \text{ if and only if } && \frac{a}{2|V|}      \in \Bigg(\frac{1-k}{1+k},\frac{1+k}{1-k} \Bigg),
\\
& 0 < \lambda_1 < \lambda_2 && \text{ if and only if } && \frac{a}{2|V|}      > \frac{1+k}{1-k}. 
\endaligned 
\ee
In particular, in a region where $a<0$, both eigenvalues $\lambda_1,\lambda_2$ are negative and the fluid flow toward the center. 
These conditions will play a role in Section~\ref{sec:7} when we will need to describe a class of initial data set of particular
interest. 

\begin{proof} 
1. In view of \eqref{eq:Eulerhom}, we can express $M$ and $V$ in terms of $U$: 
$$
 M = U_1, \qquad
 V = \frac{1}{K^2} \left( \frac{U_2}{U_1} - \frac{a}{2} \right),
$$
 and we thus obtain an explicit form of $F(U)$ in terms of $U$, i.e. 
$$
 F(U) = b \begin{pmatrix} U_2 \\ \frac{k^2}{(1-k^2)^2} \left( a^2 U_1 - 4 a U_2 \right) + \frac{1}{K^4} \frac{U_2^2}{U_1}  \end{pmatrix}.
$$
 The Jacobian matrix of $F$ is
$$
\aligned
 D_U F(U) 
& = b \begin{pmatrix}
               0 & 1 \\
               \frac{k^2}{(1-k^2)^2} a^2 - \frac{1}{K^4} \frac{U_2^2}{U_1^2} & \frac{2}{K^4} \frac{U_2}{U_1} - \frac{4k^2}{(1-k^2)^2} a
              \end{pmatrix}
\\
 & = b \begin{pmatrix}
                 0 & 1 \\
                 - \left( \frac{a^2}{4} + K^2 a V + V^2 \right) & 2 K^2 V + a
                \end{pmatrix},
\endaligned
$$
 with eigenvalues
\bel{ev1}
\aligned
\lambda_1 &= b \left( \frac{1+k^2}{(1-k)^2} \frac{U_2}{U_1} - \frac{k}{(1-k)^2} a \right)  
           = b \left( \frac{1+k}{1-k} V + \frac{a}{2} \right), \\
\lambda_2 &= b \left( \frac{1+k^2}{(1+k)^2} \frac{U_2}{U_1} + \frac{k}{(1+k)^2} a \right)  
           = b \left( \frac{1-k}{1+k} V + \frac{a}{2} \right),
\endaligned
\ee
 and right eigenvectors
$$
 r_1 = - \begin{pmatrix} 1 \\ \frac{\lambda_1}{b} \end{pmatrix}, \qquad r_2 = \begin{pmatrix} 1 \\ \frac{\lambda_2}{b} \end{pmatrix},
$$
 as stated in the proposition. Since $1-k < 1+k$ and $V<0$, it is clear that $\lambda_1 < \lambda_2$. 

\

\noindent 2. 
 The gradients of the eigenvalues $\lambda_1,\lambda_2$ are derived from \eqref{ev1}, i.e. 
\begin{align*}
 \nabla_U \lambda_1 &= \frac{1+k^2}{(1-k)^2} \frac{b}{U_1^2} \begin{pmatrix} - U_2 \\ U_1 \end{pmatrix}
                     = \frac{(1+k^2)}{(1-k)^2} \frac{b}{M} \begin{pmatrix} - \frac{a}{2} - \frac{1}{K^2} V \\ 1 \end{pmatrix}, \\
 \nabla_U \lambda_2 &= \frac{1+k^2}{(1+k)^2} \frac{b}{U_1^2} \begin{pmatrix} - U_2 \\ U_1 \end{pmatrix}
                     = \frac{(1+k^2)}{(1+k)^2} \frac{b}{M} \begin{pmatrix} - \frac{a}{2} - \frac{1}{K^2} V \\ 1 \end{pmatrix},
\end{align*}
and a straightforward computation yields
\begin{align*}
 \nabla \lambda_1 \cdot r_1 &= - \frac{(1+k^2)}{(1-k)^2} \frac{b}{M} \begin{pmatrix} - \frac{a}{2} - \frac{1}{K^2} V \\ 1 \end{pmatrix}
                                            \cdot \begin{pmatrix}
                                                   1 \\ \frac{1-k}{1+k} V + \frac{a}{2}
                                                  \end{pmatrix} = - \frac{2 k (1+k)}{(1-k)^3} \frac{b V}{M}
\end{align*}
 and
$$
 \nabla \lambda_2 \cdot r_2 = - \frac{2k(1-k)}{(1+k)^3} \frac{b V}{M}.
$$
 Since $k \in (0,1)$, the terms involving $k$ are positive. Moreover, since $b,M>0$ (by assumption) and $V<0$, the second factor is
 negative, and the statement in the proposition follows.
\end{proof}


\subsection{Shock curves and rarefaction curves}
\label{ssec:shockrarefaction}
 
We introduce the {\it Riemann invariants\/} $w,z$, associated with the hyperbolic system~\eqref{eq:Eulerhom}. By definition, the
functions $w,z$ are constant along the integral curves of the eigenvectors, i.e.~satisfy the differential equations
$$
D_U w(U) \cdot r_1 (U) = 0, \qquad\quad
D_U z(U) \cdot r_2 (U) = 0. 
$$
In the coordinates $(M,V)$, these equations are equivalent to
$$
\del_M w + \frac{2k}{(1-k)^2} \frac{V}{M} \, \del_V w = 0, \qquad\quad
\del_M z - \frac{2k}{(1+k)^2} \frac{V}{M} \, \del_V z = 0,
$$
respectively, so that  
\bel{Rinv}
w(M,V) := \log |V| - \frac{2k}{(1-k)^2} \, \log M,
\qquad\quad
z(M,V) := \log |V| + \frac{2k}{(1+k)^2} \, \log M.
\ee
Rarefaction waves are determined from integral curves of the vector fields $r_1,r_2$. As this is most convenient for the construction
of the solutions to the Riemann problem (in the following subsection), we consider here the ``forward''  $1$-curves and the
``backward'' $2$-curves. 

\begin{lemma}[Rarefaction waves]
\label{lemma:rarefaction}
The $1$-rarefaction curve $\Rr_1(U_L)$ and the $2$-rarefaction curve $\Rl_2(U_R)$ associated with the constant states $U_L= (M_L,V_L)$
and $U_R=(M_R,V_R)$, respectively,  are given by 
\bel{R12}
\aligned
\Rr_1 (U_L) :=& \Big\{M = M_L \left( \frac{V}{V_L} \right)^{\frac{(1-k)^2}{2k}}; 
\quad  V / V_L \in (0,1] \Big\}, 
\\
\Rl_2 (U_R) :=& \Big\{ M = M_R\left( \frac{V}{V_R} \right)^{-\frac{(1+k)^2}{2k}}; 
\quad
V / V_R \in [1,\infty) \Big\}.
\endaligned
\ee
Along $\Rr_1(U_L)$, the wave speed $\lambda_1(V)$ is increasing for $V$ increasing from $V_L$. Along $\Rl_2(U_R)$, the wave speed
$\lambda_2(V)$ is decreasing for $V$ decreasing from $V_R$. Moreover, $M$ is decreasing in both cases and the restriction of the
component $M$ to these curves satisfies 
\bel{eq:behave}
\lim_{V\to0} M|_{\Rr_1(U_L)}= \lim_{V \to - \infty} M|_{\Rl_2(U_R)} = 0.
\ee
\end{lemma}

Note also that, in Riemann invariant coordinates, the rarefaction curves read 
\bel{R12alt}
\aligned
\Rr_1 (U_L) & = \Big\{ (w,z) \, \big| \,  w(M,V) = w(M_L,V_L) ~\text{and}~ z(M,V) \leq z(M_L,V_L) \Big\}, 
\\
\Rl_2 (U_R) & = \Big\{ (w,z) \, \big| \, w(M,V) \geq w(M_R,V_R) ~\text{and}~ z(M,V) = z(M_R,V_R) \Big\}.
\endaligned
\ee

\begin{proof} Rarefaction waves for the system \eqref{eq:Eulerhom} are solutions of the form $U = U(\frac{r}{v})$, which
must therefore satisfy the ordinary differential equation 
$$ 
(D_U F(U) - \xi \, \mathbb{I} ) \, \del_\xi U = 0
$$
in the self-similar variable $\xi:=r/v$,  where $\mathbb{I}$ denotes the identity matrix. A characterization of the two rarefaction
curves passing through a given state $(M_0,V_0)$ in the phase space is provided by the Riemann invariants \eqref{Rinv}.
Specifically,  the $1$-rarefaction curve $\Rr_1(U_0)$ is determined implicitly
by the condition $w(U)=w(U_0)$, while the $2$-rarefaction curve $\Rl_2(U_0)$ is given by $z(U)=z(U_0)$. Hence, we arrive easily at the
expressions in \eqref{R12}. 

In view of \eqref{eigenvalues}, the speeds $\lambda_1(V)$ and $\lambda_2(V)$ increase when $V$ increases. Therefore, since $V<0$, the
speed $\lambda_1(V)$ increases along $\Rr_1(U_0)$ while $\lambda_2(V)$ decreases along $\Rl_2(U_0)$ (from the base point). The desired
monotonicity and limiting behavior for $M$ follows from \eqref{R12}. On the other hand, using \eqref{Rinv} and \eqref{R12}, along the
curve $\Rr_1(U_0)$ we obtain 
\begin{align*}
z(M,V) = \log |V| + \frac{2k}{(1+k)^2} \log M  
&= \log |V| + \frac{2k}{(1+k)^2} \log M_0 + \frac{(1+k)^2}{(1-k)^2} \log \frac{V}{V_0} 
\\
&\leq \log | V_0 | + \frac{2k}{(1+k)^2} \log M_0 = z(M_0,V_0). 
\end{align*}
Similarly, along the curve $\Rl_2(U_0)$, we obtain $w(M,V) \geq w(M_0,V_0)$.  
\end{proof}


Shock waves for the system \eqref{eq:Eulerhom} consist of two constant states $U_L$ and $U_R$ separated by a discontinuity which
propagates at the speed $s=s(U_L,U_R)$ determined by the so-called {\it Rankine-Hugoniot conditions\/}: 
\be
\label{eq:jump}
s \, [ U ] = [ F(U) ], 
\ee
with $[U] := U_R - U_L$ and $[F(U)] := F(U_R) - F(U_L)$.  Moreover, the {\it shock admissibility inequalities\/} 
\bel{eq:admis}
\lambda_i (M_R,V_R) < s_i < \lambda_i (M_L,V_L), \qquad i=1,2
\ee
are imposed in order to guarantee uniqueness of the Riemann solution, defined below. Before we state some properties of these shock
wave solutions, we introduce the functions  
\begin{align*}
\Phi_\pm (\beta) &:= \frac{1}{2 (1-K^4) \beta^2} \left( 1 - 2K^4 \beta + \beta^2 \pm ( 1 - \beta ) \sqrt{(1+\beta)^2 - 4 K^4 \beta} \right), 
\\
\Sigma_\mp (V_0,\beta) &:= b \left( \frac{a}{2} + V_0 \frac{1 + \beta \mp \sqrt{(1+\beta)^2-4 K^4 \beta}}{2 K^2} \right). 
\end{align*} 
The signs above are selected for convenience in the following statement.  

\begin{lemma}[Shock waves]
\label{lemma:shock}
 The $1$-shock curve issuing from a given state $U_L$ and the corresponding shock speed are given by 
$$
\aligned
& \Sr_1(U_L) = \Big\{ M = M_L \, \Phi_- \big( V/V_L \big); \quad 
V /V_L \in [1,\infty) \Big\}, 
\\
& s_1(U_L,U)  = \Sigma_+ \big(V_L, V /V_L \big), 
\endaligned
$$
while the $2$-shock curve  issuing from the state $U_R$ and the corresponding shock speed are given by 
$$
\aligned
& \Sl_2(U_R) = \Big\{ M = M_R \, \Phi_+ \big( V/V_R \big); \quad 
V /V_R \in (0,1] \Big\}, 
\\
& s_2(U,U_R)  = \Sigma_- \big(V_R, V /V_R \big).  
\endaligned
$$
Moreover, the $1$-shock speed $s_1$ is increasing for $V$ decreasing, while the $2$-shock speed $s_2$ is decreasing for $V$ increasing, 
and the shock admissibility inequalities \eqref{eq:admis} hold, together with  
\bel{eq:admis2}
 s_1 < \lambda_2(V_R), \qquad  \qquad \lambda_1(V_L) < s_2.
\ee
Furthermore, along the curve $\Sr_1(U_L) $, the mass density $M$ is increasing and reaches $\frac{M_L}{1-K^4}$ as $V \to -\infty$,
while along the curve $\Sl_2(U_R)$ it is increasing and blows up as $V \to 0$.
\end{lemma}

The geometry of the shock curves can also described in Riemann invariant coordinates $(w,z)$: namely, using the parameter
$\beta = \frac{V}{V_L} \in [1,\infty)$ for $\Sr_1(U_L)$ and $\beta=\frac{V}{V_R} \in (0,1]$ for $\Sl_2(U_R)$, we find 
\bel{shock-riemann}
\aligned
& \Sr_1 (U_L) \colon 
\begin{cases}
& 
 w - w _L = \log \beta - \frac{2k}{(1-k)^2} \log (\Phi_-(\beta)), 
\\
&  
z - z _L = \log \beta + \frac{2k}{(1+k)^2} \log (\Phi_-(\beta)),
\end{cases} 
\\
& \Sl_2 (U_R) \colon 
 \begin{cases}
& 
w - w_R = \log \beta - \frac{2k}{(1-k)^2} \log (\Phi_+(\beta)),
\\
& 
z - z_R = \log \beta + \frac{2k}{(1+k)^2} \log (\Phi_+(\beta)).
\end{cases}
\endaligned
\ee

\begin{proof}
1. In view of \eqref{eq:Eulerhom} and in terms of the conservative variable $U=(U_1, U_2)$, we obtain $F(U)_1 = b \, U_2$ and,
therefore, after eliminating the shock speed $s$ in the jump condition~\eqref{eq:jump}, we find 
\bel{eq:jump2}
b \, [U_2]^2 = [U_1] [F(U)_2].
\ee
Again in view of \eqref{eq:Eulerhom} and by using the notation $U_0, U$ rather than $U_L, U_R$, we have 
\begin{align*}
[U_2] & =  \frac{a}{2} (M_0 - M) + K^2 (M_0V_0 - M V), 
\\
\frac{1}{b}  \, [U_1] [F(U)_2] & =  \frac{a^2}{4}(M_0 - M)^2 + a K^2 (M_0 - M)(M_0 V_0 - M V)  + (M_0 - M) (M_0 V_0^2 - MV^2),
\end{align*}
hence \eqref{eq:jump2} simplifies and yields 
$$
(M_0 - M)(M_0V_0^2-MV^2) = K^4 (M_0V_0 - MV)^2.
$$
This relations can be written as a quadratic equation in terms of $\alpha = \frac{M}{M_0} > 0$ and $\beta = \frac{V}{V_0}>0$, which
admits two distinct and real solutions 
\bel{eq:shock1}
\alpha = \frac{1}{2 (1-K^4) \beta^2} \left( 1 - 2K^4 \beta + \beta^2 \pm ( 1 - \beta ) \sqrt{(1+\beta)^2 - 4 K^4 \beta} \right) 
= \Phi_\pm (\beta). 
\ee 
Thus, the shock curves $\Sr_1(U_0), \Sl_2(U_0)$ are given implicitly in terms of $\alpha, \beta$ in \eqref{eq:shock1}. Observe that
they do not depend on the geometric coefficients $a,b$, but only on the constant $K$ (and thus the sound speed $k$). Since $K^4 < 1$,
the term $(1+\beta)^2 - 4 K^4 \beta$ is positive.

Moreover, the ``first'' jump condition yields
\bel{eq:shock2}
s = \frac{[F(U)_1]}{[U_1]} = \frac{b [U_2]}{[U_1]} = b \left( \frac{a}{2} + K^2 V_0 \frac{1-\alpha\beta}{1-\alpha} \right)
= \Sigma_\mp(V_0, \beta),
\ee
in which the term $\frac{1-\alpha\beta}{1-\alpha}$ is expressed explicitly using the characterization $\alpha=\Phi_\pm(\beta)$ of the
shock curves, namely 
$$
\frac{1-\alpha\beta}{1-\alpha} 
= \frac{1}{2 K^4} \, \Big(1 + \beta \mp \sqrt{(1+\beta)^2-4 K^4 \beta} \Big).
$$
We emphasize that the negative sign (leading to $\Sigma_-$) corresponds to the function $\Phi_+$, while the positive sign
(leading to $\Sigma_+$) corresponds to $\Phi_-$.
By setting $V=V_0$, that is, $\beta = 1$ (and thus $\alpha=1$ in view of \eqref{eq:shock1}), we conclude that
$$
K^2 \frac{1-\alpha\beta}{1-\alpha} = \frac{1 \mp \sqrt{1-K^4}}{K^2} = \frac{(1 \mp k)^2}{1-k^2}. 
$$
Thus, \eqref{eq:shock2} is naturally associated with the eigenvalues $\lambda_2$ and $\lambda_1$, respectively. (Cf.~also
Proposition~\ref{prop:31}, above.) 

\hskip.15cm 

\noindent 2. It remains to be determined which half-curves are admissible with respect to the shock admissibility inequalities.
Consider for instance the $1$-shock curve $\Sr_1 (U_0)$, defined by the function $\Phi_-$ and the shock speed function $\Sigma_+$.
The shock inequalities are equivalent to saying  
$$
V \, \frac{1 + \sqrt{1-K^4}}{K^2} 
< \frac{1}{2 K^2} \, \Big( V + V_0 - \sqrt{(V+V_0)^2-4K^4 VV_0} \Big) 
< V_0 \, \frac{1 + \sqrt{1-K^4}}{K^2},
$$ 
which (since all values are negative) is equivalent to 
$$
 4 V_0^2 \left( 1 + \sqrt{1-K^4} \right)^2 < \left( V + V_0 - \sqrt{(V+V_0)^2-4K^4 VV_0} \right)^2 < 4 V^2 \left( 1 + \sqrt{1-K^4} \right)^2.
$$
 For $\beta > 1$, that is, $V<V_0<0$, the first inequality is obviously satisfied, since $(V+V_0)^2 - 4 K^4 V V_0 > (V+V_0)^2(1-K^4)$.
 The second inequality also holds, since   
\begin{align*}
& \Big( V + V_0  - \sqrt{(V+V_0)^2-4K^4 VV_0} \Big)^2 
\\
& = 2 \left( V^2+V_0^2 \right) + 4 (1-K^4) V V_0 - 2 (V + V_0) \sqrt{(V+V_0)^2-4K^4 VV_0} 
\\
& < 4 V^2 (2-K^4) + 8 V^2 \sqrt{1-K^4} = 4 V^2 \left( 1 + \sqrt{1-K^4} \right)^2.
\end{align*}
Adding a constant $\frac{a}{2}$ and multiplying by $b>0$ has no effect on the signs, hence we conclude that 
$\lambda_1(V) < s_1 < \lambda_1 (V_0)$.

We can similarly treat the $2$-shock curve $\Sl_2 (U_0)$, defined by $\Phi_+$ and $\Sigma_-$. For $\beta<1$, that is, $V_0<V<0$, we
find
\begin{align*}
\left( V_0 + V + \sqrt{(V_0+V)^2 - 4 K^4 V_0 V} \right)^2 & < \left( V_0 + V + |V-V_0| \sqrt{1-K^4} \right)^2 
\\
& < (V + V_0)^2 \left( 1 - \sqrt{1-K^4} \right)^2 < 4 V_0^2 \left( 1 - \sqrt{1-K^4} \right)^2,
\end{align*}
thus $\lambda_2(V_0) < s_2$. The second inequality $s_2 < \lambda_2(V)$ follows from 
\begin{align*}
& V_0 + V + \sqrt{(V_0+V)^2 - 4 K^4 V_0 V} 
\\
& = 2 V + (V_0-V) + \sqrt{4 V_0 V (1 - K^4) + (V_0 - V)^2} 
\\
& < 2 V + (V_0-V) + 2 |V| \sqrt{1-K^4} + |V_0 - V|
  = 2 V \left( 1 - \sqrt{1-K^4} \right),
\end{align*}
where we used 
\begin{align*}
 4V_0 V (1-K^4) + (V_0-V)^2 & = 4V^2(1-K^4) + 4V (V_0-V) (1-K^4) + (V_0-V)^2 
\\
& < 4V^2 (1-K^4) + 4 V (V_0-V) \sqrt{1-K^4} + (V_0-V)^2 
\\
& = \left( 2 |V| \sqrt{1-K^4} + |V_0 - V| \right)^2.
\end{align*}

\hskip.15cm 

\noindent 3. A straightforward calculation reveals \eqref{eq:admis2}, which we will check only for $s_1$. Since 
$\frac{1}{K^2} > \frac{1-k}{1+k}$ and $\sqrt{(V_0+V)^2 - 4 K^4 V_0 V} > |V_0-V|$, we obtain 
\begin{align*}
 s_1 (V_0,V) & = b \left( \frac{a}{2} + \frac{V_0 + V - \sqrt{(V_0+V)^2-4K^4V_0V}}{2K^2} \right) \\ 
 & < b \left( \frac{a}{2} + \frac{1-k}{1+k} \frac{V_0 + V - |V_0 -V|}{2}   \right)
   = b \left( \frac{a}{2} + \frac{1-k}{1+k} V \right) = \lambda_2(V)
\end{align*}
and, moreover, 
$$
\aligned
s_1'(V_0,V) 
& = \frac{1}{2K^2} \left( 1 - \frac{V_0+V-2K^4V_0}{\sqrt{(V_0+V)^2-4K^4VV_0}} \right)
\\
& > \frac{1}{2K^2} \left( 1 - \frac{2 V_0 (1-K^4)}{\sqrt{(V_0+V)^2-4K^4VV_0}} \right) > 0,
\endaligned
$$
so that the shock speed of $\Sr_1(U_0)$ is monotone increasing in $V$.

\hskip.15cm 

4. To study the behavior of $M$ with respect to $V$, we set $\widetilde \Phi_\pm (V) := \Phi_\pm \left( \frac{V}{V_0} \right)$ and
observe that 
$$
\widetilde \Phi_\pm'(V) = \pm \frac{V_0 \, \phi_\pm(V)}{(1-K^4)V^3 \sqrt{(V+V_0)^2 - 4 K^4 V V_0}}
$$
 with the auxiliary function
\begin{align*}
 \phi_\pm (V) := \; & K^4 V \left( V - 3 V_0 \pm \sqrt{(V+V_0)^2 - 4 K^4 V V_0} \right)
+ V_0 \left( V + V_0 \mp \sqrt{(V+V_0)^2 - 4 K^4 V V_0} \right) \\
 =\; & K^4 (V^2 + V_0 V + V_0^2) - 3 K^4 VV_0 \pm (K^4 V - V_0) \sqrt{(V+V_0)^2 - 4 K^4 V V_0}.
\end{align*}
 Since $V^2 + V_0^2 > 2 V V_0$, this implies that
\begin{align*}
 \phi_\pm(V) > \pm (K^4 V - V_0) \sqrt{(V+V_0)^2 - 4 K^4 V V_0},
\end{align*}
 and hence $\phi_+$ is positive as $V_0<V<0$ and $K<1$. In the case of $\phi_-$ we distinguish between two cases, as follows. 
If $K^4 V - V_0 \leq 0$, then $\phi_-$ is positive by the same inequality. On the other hand, if $K^4 V - V_0 > 0$, then the sign of
$\phi_-$ is derived separately by using $V<V_0<0$: 
\begin{align*}
 \phi_-(V) &= K^4 V \left( V - 3 V_0 - \sqrt{(V+V_0)^2 - 4 K^4 V V_0} \right)
+ V_0 \left( V + V_0 + \sqrt{(V+V_0)^2 - 4 K^4 V V_0} \right) \\
&= (V_0 - K^4 V) \left( V + V_0 + \sqrt{(V+V_0)^2 - 4 K^4 V V_0} \right) + 2 K^4 V (V-V_0) \\
&> (V_0 - K^4 V) \left( V + V_0 + \sqrt{(V+V_0)^2 - 4 K^4 V V_0} \right) > 0.
\end{align*}
 Hence $\phi_\pm > 0$ and $\widetilde \Phi_+' >0$ and $\widetilde \Phi_-' <0$. Thus, on both shock curves, $M$ is increasing when $V$
moves away from $V_0$. The limiting behavior $V \to - \infty$ on $\Sr_1(U_0)$ and $V \to 0$ on $\Sl_2(U_0)$ is clear from the
expressions of $\Phi_\pm$.
\end{proof}


\subsection{The Riemann problem} 
\label{ssec:Riemann}

We observe that the geometry of the wave curves is independent of the geometry of the spacetime and solely depends on the fluid
variables $M$ and $V$, while the wave speeds also depend on the geometry variables $a$ and $b$. This provides an important advantage
for our analysis in this paper, which strongly relies on the properties of these wave curves and wave speeds.  
We begin by solving the Riemann problem for the homogeneous model \eqref{eq:Eulerhom}  of interest in this section, that is, we solve
the initial value problem with data prescribed on $v=0$ with a single jump located at some point $r_1 \in (0,\infty)$: 
\be
\label{Riemann} 
U(0,r) = \begin{cases}
U_L,   & r < r_1,
\\
U_R,   & r > r_1,
\end{cases}
\ee
where $U_L$ (determined by $M_L, V_L$) and $U_R$  (determined by $M_R, V_R$) are constants satisfying the physical constraints 
$$
M_L, M_R >0, \qquad \quad V_L, V_R < 0.
$$
Obviously, since the coefficients of the system \eqref{eq:Eulerhom} are independent of $r$, we can consider that the solutions are
defined for all $r$ (even negative values) and, due to the invariance of the Riemann problem by self-similar scaling, we search for a
solution depending upon the variable $r/t$ only. Recall also that all variables $(M,V)$ under consideration satisfy the conditions
$M>0$ and $V<0$.

\begin{proposition}[Riemann problem on an Eddington-Finkelstein background]  
\label{prop-Riemn}
The Riemann problem associated with the homogeneous version \eqref{eq:Eulerhom} of the Euler system on a uniform
Eddington-Finkel\-stein background and with the initial condition \eqref{Riemann} with arbitrary initial data $U_L, U_R$, 
admits a unique self-similar solution $U=U(r/t)$ made of two waves, each being a rarefaction wave or a shock wave satisfying the shock
admissibility inequalities. Moreover, the regions (with $\rho>0$) 
$$
\Omega_\rho:=  \big\{(w,z) \, | \, - \rho \leq w, z\leq \rho \big\} 
$$
are invariant domains for the Riemann problem, that is, if the data $U_L, U_R$ belong to $\Omega_\rho$ for some $\rho>0$, 
then so does the solution for all times $v \geq 0$.  
\end{proposition}

\begin{proof} By Proposition~\ref{prop:31} the system \eqref{eq:Eulerhom} is strictly hyperbolic and genuinely nonlinear as long as $V$ is
nonzero and $M$ is bounded. Thus, for sufficiently small jumps $|U_R - U_L|$, the claim is standard (cf.~, for instance, \cite{LeFloch-book}).
In order to extend the Riemann solution to arbitrarily large initial data, we rely on the explicit formulas derived earlier in this section.
The Riemann solution is constructed in the phase space by piecing together constant states, shock curves, and rarefaction curves
(defined in Section~\ref{ssec:shockrarefaction}) and, specifically, we introduce the $1$-wave curve issuing from the data $U_L$, 
$$
\Wavr_1(U_L) : = \Rr_1(U_L) \cup \Sr_1 (U_L),  
$$
which, according to our earlier notation, is naturally parametrized by a variable $\beta$ describing the interval $(0,1]$ (within the
rarefaction part $\Rr_1(U_L)$) and the interval $[1, +\infty)$ (within the shock part $\Sr_1(U_L)$). 
The wave curve $\Wavr_2(U_R)$ is defined similarly and the Riemann problem is solved if these two curves intersect at a unique point
$U_* \in \Wavr_1(U_L) \cap \Wl_2(U_R)$ so that the Riemann solution can be defined as a $1$-wave pattern connected to a $2$-wave pattern.

In order to establish the validity of this construction, we argue as follows. Thanks to Lemmas~\ref{lemma:rarefaction} and
\ref{lemma:shock}, the wave speeds arising in the Riemann solution do increase from left to right in the proposed construction. From
Lemmas~\ref{lemma:rarefaction} and \ref{lemma:shock}, it follows that $V$ decreases from $0$ toward $-\infty$, while $M$ increases
from $0$ toward $\frac{M_L}{1-K^4}$ along the curve $\Wavr_1(U_L)$. On the other hand, along $\Wavr_2(U_R)$, the velocity $V$ decreases
from $0$ toward $-\infty$, while the mass density $M$ decreases from $+\infty$ toward $0$. Therefore, in view of these global
monotonicity properties, the intersection point $U_* \in \Wavr_1(U_L) \cap \Wl_2(U_R)$ exists and is unique (for any given initial
states $U_L,U_R$ satisfying $M_L, M_R >0$ and $V_L, V_R < 0$).

We next claim that any domain $\Omega_\rho$ is an invariant region for the Riemann problem. We write $w_L$ for $w(U_L)$, etc. and,
for definiteness, we suppose that $U_* \in \Rl_1(U_L) \cap \Rr_2(U_R)$. Then, by Lemma~\ref{lemma:rarefaction}, we have $w=w_L$ and
$z \leq z_L$ for all states between $U_L$ and $U_*$, while $w \geq w_R$ and $z=z_R$ for all states between $U_*$ and $U_R$. Thus, we
obtain 
$$
w_R \leq w = w_L, \quad \quad z_R = z \leq z_L
$$
along the solution of the Riemann problem, and, in particular $w,z \in [-\rho,\rho]$ if $w_L,w_R,z_L,z_R \in [-\rho,\rho]$.

We are going to prove that both shock curves $\Sr_1(U_L)$ and $\Sl_2(U_R)$ remain within an
upper-left triangle in the $(w,z)$-plane so that, if intersected with each other or with $\Rl_2(U_R)$ and $\Rr_1(U_L)$, respectively,
the corresponding Riemann solution belongs to the region $\Omega_\rho$.
Namely, the tangent to the shock curve $\Sr_1$ in the $(w,z)$-plane satisfies
$$
\aligned
 \frac{d w}{d z} 
 = \frac{d (w - w_L)}{d (z - z_L)} 
&= \frac{d (w-w_L)}{d \beta} \left( \frac{d (z-z_L)}{d \beta} \right)^{-1} 
\\
& = \frac{(1+k)^2 \left( 1+k^2- \frac{2k (1+\beta)}{\sqrt{(1+\beta)^2-4K^4\beta}}\right)}{(1-k)^2 \left( 1+k^2+
\frac{2 k (1+\beta)}{\sqrt{(1+\beta)^2-4K^4\beta}} \right)},
\endaligned
$$
which is less than $1$ (since $k \in (0,1)$). 
Moreover, $\Sr_1$ is convex, since $\beta \geq 1$ and
$$
 \frac{d}{d\beta} \frac{dw}{dz} = \frac{8 k(1+k)^2K^2(-1 + \beta)}{\sqrt{(1+\beta)^2-4K^4\beta}\left(\sqrt{(1+\beta)^2-4K^4\beta}
+k\left(2+2\beta+k\sqrt{(1+\beta)^2-4K^4\beta}\right)\right)}
$$
is non-negative. 
Since $\frac{2k}{\sqrt{1-K^4}} = 1 + k^2$, we have 
$$ 
 \lim_{\beta \to 1+} \frac{dw}{dz} = \frac{(1+k)^2 \left(1+k^2-\frac{2k}{\sqrt{1-K^4}}\right)}{(1-k)^2 \left(1+k^2+\frac{2k}{\sqrt{1-K^4}}\right)} = 0,
$$
and the second--order derivative being positive, we conclude that $\frac{dw}{dz} \in [0,1]$. It is checked similarly that the shock curve $\Sl_2(U_R)$
satisfies
$$
 \frac{d z}{dw} = \frac{(1-k)^2 \left( 1+k^2- \frac{2k (1+\beta)}{\sqrt{(1+\beta)^2-4K^4\beta}}\right)}{(1+k)^2 \left( 1+k^2+ 
\frac{2 k (1+\beta)}{\sqrt{(1+\beta)^2-4K^4\beta}} \right)} \in \left[0,\frac{(1-k)^4}{(1+k)^4}\right] \subset[0,1)
$$
and, since $\beta \in (0,1]$ and 
$$
 \frac{d}{d\beta}\frac{dz}{dw} = \frac{8 k (1-k)^2 (1-k^2) K^2 (-1 + \beta)}{(1+k)^2 \sqrt{(1+\beta)^2-4K^4\beta}\left(\sqrt{(1+\beta)^2-4K^4\beta}
+k\left(2+2\beta+k\sqrt{(1+\beta)^2-4K^4\beta}\right)\right)}
$$
is non-positive. 
In other words, the curve $\Sl_2(U_R)$ is concave in the $(w,z)$-plane. 
\end{proof}


\subsection{Wave interactions} 
 
To conclude this section we derive some estimates concerning a pair of Riemann solutions associated with the system  
\eqref{eq:Eulerhom}. We now assume that the initial data consists of three constant states, denoted by $U_L,U_M,U_R$ and, specifically, for some $0<r_1<r_2<+\infty$, we prescribe at $v=0$ the data 
\begin{align}
 \label{Riemann3}
 U(0,r) = \begin{cases}
           U_L, & r<r_1, \\
	   U_M, & r_1<r<r_2, \\
           U_R, & r>r_2.
          \end{cases}
\end{align}
Again we can consider that $r$ describes the real line.  For sufficiently small times $v$, it is clear that the solution can be
constructed by combining the Riemann problems associated with the initial data $U_L,U_M$ and $U_M,U_R$, respectively. In general
these waves interact and generate a complex wave pattern. Yet, for sufficiently large times $v$ after all waves have interacted, the
solution is expected to approach the solution of the Riemann problem with initial data $U_L,U_R$; more precisely, this is true for
the wave strength (defined below) and wave speeds, while the location of the wave depends upon the past interactions. 

By definition, the \textsl{wave strength} $\str(U_L,U_R)$ of a Riemann problem $(U_L,U_R)$ measures the magnitude of the waves in the
solution and, in Riemann invariant coordinates, reads
$$
\str(U_L,U_R) := 
 \left| \log M_R - \log M_* \right| + \left| \log M_* - \log M_L \right|, 
$$
where $M_*$ denotes the intermediate state characterized by the condition  $U_* \in \Wl_1(U_L) \cap \Wavr_2(U_R)$. 
The following property will be essential in order to derive a bound on the
total variation of the solutions to the general Cauchy problem. 

\begin{lemma} 
\label{lem94}
Given arbitrary states $U_L, U_M, U_R$, the wave strengths associated with the Riemann problems $(U_L,U_M)$, $(U_M,U_R)$, and
$(U_L,U_R)$ satisfy the inequality 
\bel{eq:inter}
 \str(U_L,U_R) \leq \str(U_L,U_M) + \str(U_M,U_R).
\ee
\end{lemma}

\begin{proof} We consider the wave curves in the plane of the Riemann invariants. Recall that, in this plane, rarefaction curves are straightlines, while shock curves are described by the expressions \eqref{shock-riemann}. The shock curves have the same geometric shape independently of the base point $U_L$ or $U_R$ and are essentially described by the functions $\Phi_\pm$. Moreover, by observing the 
remarkable algebraic property 
$$
\aligned
\Phi_-(\beta) \Phi_+(\beta)
& = \Big( 4 (1 - K^4)^2 \beta^4\Big)^{-1}  \Big( (1 - 2 K^4 \beta + \beta^2 \big)^2 
- (1 - \beta)^2 \big( (1 + \beta)^2 - 4 K^4 \beta \big) \Big) 
\\
& = \Big( 4 (1 - K^4)^2 \beta^4\Big)^{-1}  \Big( (1 -\beta)^2  + 2 \beta (1 - K^4) \big)^2 
- (1 - \beta)^2 \big( (1 - \beta)^2 - 4 \beta (1 - K^4 \beta) \big)\Big) 
\\
& = 1, 
\endaligned 
$$
it follows that $\log (\Phi_-(\beta)) = -\log (\Phi_+(\beta))$ and 
the expressions in \eqref{shock-riemann} coincide up to a change of the role of the variables $w$ and $z$. Therefore, 
the shock curves are symmetric with respect to the $w=z$ axis. Finally, since the wave strengths, by definition, are measured along this $w=z$ axis, these symmetry properties are sufficient to imply that the wave strengths are non-decreasing at each interaction, as stated in \eqref{eq:inter}. 
\end{proof}


\section{The dynamical formation of trapped surfaces}
\label{sec:6} 

\subsection{Random choice method}
\label{sec:random}

We now state our main result about the existence of solutions $U = U(M,V,a)$
to the Einstein-Euler system \eqref{eq:Euler}--\eqref{S1S2}, 
supplemented with the 
 equations \eqref{eq:Einstein1-new}--\eqref{eq:Einstein2-new} for the geometry coefficients $a,b$ given by the integral expressions
 \eqref{b-new} and \eqref{a-new}. We will also use the notation $Z:=(M,V,a,b)$. 

We consider initial data which are compactly-supported perturbations of a given static solution, denoted by $Z^{(0)}= (M^{(0)}, V^{(0)}, a^{(0)}, b^{(0)})$. The perturbation is assumed to be initially localized on an interval $[r_*- \delta, r_* + \delta]$ with for some $r_* > \delta >0$ (with a ``sufficiently small''  $\delta$)
and we construct a spacetime which remains static in a neighborhood of the center of symmetry. Due to the property of finite speed of
propagation, the support of the initial perturbation remains finite and bounded away from the center (for all times), but may increase
in space as the time evolves.  

For solutions defined for times $v \in [v_0, v_*]$, we expect that 
$$
\supp (U - U^{(0)})(v, \cdot) \subset J(v) := [R_*^-(v), R_*^+(v)], \qquad v \in [v_0, v_*], 
$$
for some functions 
$$
R_*^-(v)= r_* -  \delta - C_* (v- v_0), \qquad 
R_*^+(v)= r_* +\delta + C_* (v- v_0), \qquad 
 v \in [v_0, v_*].  
$$
These functions involve a constant $C_*$, which should be an upper bound of all wave speeds of the Euler equations. Choosing $C_*$ is
done from the explicit expressions of the wave speeds computed earlier, 
once we have a uniform bound on the sup-norm of $Z$ in the spacetime slab under consideration. 
All our analysis will take place in the region 
$$
\Omega_* := \big\{ (v, r) \, | \,  v \in [v_0, v_*], r \in [R_*^-(v), R_*^+(v)] \big\}. 
$$
The solutions will be defined in a time slab $[v_0, v_*]$ and $v_*-v_0$ will be estimated below from the prescribed initial data. 

Our main unknowns are the fluid variables $M,V$ which must satisfy the Euler system. The geometry coefficients $a,b$ arise in an
undifferentiated form in the conservative and flux variables $U, F(U)$, as well as in the source term $S(U)$.  If these coefficients
were prescribed
functions, we would simply have a non--homogeneous hyperbolic system of first-order. However, the functions $a,b$ are not a priori
prescribed and must be recovered from the fluid variables thanks to \eqref{b-new}-\eqref{a-new}. 

To study the initial value problem with data prescribed on $v=v_0$, we rely on the random choice method, which is based on the Riemann
problem and takes the source of the Euler equations into account, as follows. Consider the Riemann problem for the Euler system with
constant geometric background coefficients $a,b$ and an initial jump at time $v'$ centered at some point $r'$:   
\be
U(v', \cdot) = \begin{cases}
U_L, \quad &  r \leq r', 
\\
U_R, \quad & r > r'.    
\end{cases} 
\ee
A generalized solution $\Rcal_G(v,r; U_L,U_R, a, b)$ (the dependence in $v',r'$ being kept implicit) is constructed from the solution
$\Uh$ of the 
Riemann problem $\Rcal(U_L,U_R; a, b)$ (constructed earlier in Section~\ref{ssec:Riemann}) by evolving it with the system of ordinary
differential equations associated with the source-terms and the geometry of the Euler system. More precisely, we set 
\be
\label{grsol}
\Rcal_G(v,r; U_L,U_R, a, b) : = \Uh (v,r; U_L,U_R, a,b) + \int_{v'}^v \widetilde S(v'', W(v'', r), a,b) \, dv'',
\ee
where $W(v'', r) := \Pcal_{v''} \Uh (v',r; U_L,U_R, a, b)$ and 
$\Pcal$ denotes the solution operator for the 
ODE system 
\be
 \label{UH}
\aligned
\frac{d}{dv} W &= \widetilde S (W,a,b), 
\\
W(v', r) &= \Uh (v', r; U_L,U_R, a,b),
\endaligned
\ee   
where $S$ is the source term of the Euler system (cf.~Proposition~\ref{prop:23}) and 
$\widetilde S$ takes also the variation of the geometry into account: 
\begin{align}\label{deftildeS}
\widetilde S := ~& S - a_r \partial_a F - b_r \partial_b F \nonumber \\
= ~& - \frac{bM}{2r} \begin{pmatrix}
                      2 + 4 V \\
                      a + 4 aV + 4 V^2
                     \end{pmatrix}
     - \pi (1+k^2) r b M^2 \begin{pmatrix}
		    8 K^2 V \\
                    - a^2 + 4 K^2 a V + 12 V^2
                   \end{pmatrix}.
\end{align}
The generalized random choice method for the class of initial data of interest ``supported'' in the domain $\Omega_*$ is now introduced. 
We denote by $\Delta v, \Delta r >0$ the time and space mesh--lengths, respectively, and by
$(v_i, r_j)$ (for $i \in \mathbb{N} \cup \{ 0 \}$, $j \in \mathbb{Z}$) 
the mesh points of the grid, that is, 
$$
v_i := v_0 + i \Delta v,\qquad r_j:= r_* + j \Delta r. 
$$ 
We also fix an equidistributed sequence $(\omega_i)$ in the interval $(-1,1)$ and set 
$$
r_{i,j}:= r_* + (\omega_i + j) \Delta r. 
$$
We will let $\Delta v, \Delta r$ tend to zero, while keeping the ratio $\Delta v/\Delta r$ constant. 
We can now define the approximate solutions $Z_\D=Z_\D (v,r)$ to the Cauchy problem 
for the Einstein-Euler system associated with the (fluid) initial data 
$$
U_0(r):=U(v_0, r), \qquad r \in J(v_0) = [r_* - \delta, r_*+\delta]. 
$$
Also, throughout the evolution and for the fluid variables, we impose the boundary values determined by the prescribed static solution,
i.e.
$$
(M,V)(v,R_*^-(v)) = (M,V)^{(0)} (R_*^-(v)), 
\qquad 
(M,V)(v,R_*^+(v)) = (M,V)^{(0)} (R_*^+(v)). 
$$

The approximate solutions are defined inductively. First of all, the initial data are approximated by piecewise constant functions by
setting for all even $j$:
\bel{initial}
\aligned
 U_\D(v_0,r) &:= U_0 (r_{j+1}), & r \in [r_j, r_{j+2}), \\
 a_\D(v_0,r) &:= a_0 (r_j), &  r \in [r_{j-1},r_{j+1}), \\
 b_\D(v_0,r) &:= b_0(r_j), & r \in [r_{j-1},r_{j+1}).
\endaligned
\ee
Then, we evolve $U_\sharp$, $a_\sharp$, and $b_\sharp$ successively: 
\begin{itemize}

 \item[1.] If $U_\D$ is known for all $v < v_i$, we define $U_\D$ at the level $v = v_i$ as
$$
 U_\D (v_i+,r) := U_\D (v_i-,r_{i,j+1}), \qquad r \in [r_j,r_{j+2}), \quad i + j ~\text{even}. 
$$

 \item[2.] Similarly, we randomly pick a value for $a_\D$ and $b_\D$ between $r_{j-1}$ and $r_{j+1}$ using the equidistributed 
sequence:
\begin{align*}
 a_\D (v_i+,r) &:= a_\D (v_i-,r_{i,j}), \qquad r \in [r_{j-1},r_{j+1}), \quad i + j ~\text{even}, \\
 b_\D (v_i+,r) &:= b_\D (v_i-,r_{i,j}), \qquad r \in [r_{j-1},r_{j+1}), \quad i + j ~\text{even}.
\end{align*}

 \item[3.] The approximation $U_\D$ is defined in each slab 
$$
\Omega_{i,j} := \big\{ 
v_i < v < v_{i+1}, \quad  {r_{j-1}\leq r< r_{j+1}}, \quad i + j \, \mathrm{even}
\big\} 
$$
 from the Riemann problem and we set 
$$
 \qquad \qquad 
 U_\D (v,r) := \Rcal_G \big( v,r ; \, U_\D (v_i+, {r_{j-1}}), \; U_\D (v_i+, {r_{j+1}}); \, v_i,  {r_{j}}, 
  a_\D (v_i +,  {r_{j}}),  b_\D (v_i+,  {r_{j}})) \big),
$$
 as introduced in \eqref{grsol}.

\item[3.] Next, we update the metric coefficient $b$ using the integral formula \eqref{b-new}, that is
$$
 b_\D(v,r) = \exp \Big( 4\pi (1+k^2) \int_0^r M_\D (v,r') \, r' dr' \Big), \quad v \in (v_i,v_{i+1}),
$$
with $M_\D = U_{\D,1}$ being the first component of $U_\D$ for $r \in (R_*^-(v),R_*^+(v))$, by relying on 
the static solution $M^{(0)}$ outside $\Omega_*$.
 
\item[4.] Similarly, we update the metric coefficient $a_\D$ using the integral formula \eqref{a-new},  
that is for $v \in (v_i,v_{i+1})$: 
\begin{align*}
 \aaa(v,r) = 1 - \frac{4\pi(1+k^2)}{r} \int_0^r \frac{\ba(v,s)}{\ba(v,r)} \Ma(v,s) \left( 1 - 2 K^2 \Va(v,s) \right) s^2 ds. 
\end{align*}  
\end{itemize}


\subsection{The class of initial data of interest}
\label{sec:initialdata}

In order to establish the dynamical formation of trapped surfaces, we focus on a class of initial data for which we can prove an
existence result on a sufficiently long time interval $[v_0,v_*]$ so that
trapped surfaces form within this time interval, while the initial data are chosen to be untrapped.
Here we derive suitable conditions on the (untrapped) initial data so that trapped surfaces do form in the future.
The evolution takes place within the cone $[R_*^-(v),R_*^+(v)]$, defined earlier so that the support of the solution expands in time.
The accumulation of mass in a short amount of time is controlled by the behavior of the derivative $a_v$, as we now explain. 

The class of initial data under consideration here consists of a localized perturbation of a static solution for which $a_v = 0$. 
Generally, the derivative $a_v$ is essentially determined by $a^2-4V^2$, which we choose to be initially large and negative within
an interval $[r_*-\delta,r_*+\delta]$.  The sup--norm of $V, a$ being controlled during the evolution, we can guarantee that it varies
``slowly'' in time so that, at later times $v$, we have $a_v(v,r) < \frac{1}{h^2}$ within a {\it smaller\/} spatial interval in $r$,
determined by the property of propagation with finite speed. 

Heuristically, we expect to choose $-V>0$ to be sufficiently large, and that a ``large'' mass is concentrated on a sufficiently
``small'' interval $[r_*-\delta, r_*+\delta]$. 
To complete the argument, we need to carefully quantify all the relevant ``effects'' in the problem.

\

We identify a set of initial data $(M,V,a,b)$ at an initial hypersurface at time $v_0$ that satisfy $a>0$ everywhere and $a_v\ll0$ in
a small region. With the notation
\be
\label{initialdata1}
\aligned
 M &= \Ms + \Mp, & V &= \Vs + \Vp, \\
 a &= \as + \ap, & b &= \bs + \bp,
\endaligned
 \ee
we denote solutions that consist of a \textsl{static solution} $(\Ms,\Vs,\as,\bs)$ as derived in Theorem~\ref{thm-static} and of a
certain \textsl{perturbation} $(\Mp,\Vp,\ap,\bp)$. By adding a suitable perturbation,  the initial data $\Vp_0$ has small support in
the radial direction but large absolute value. In order to control the positive sign of $a_0$ we have to ensure that the $L^1$-norm
of $\Vp_0$ is small. On the other hand, $\Vp_0$ must be sufficiently large (pointwise) to ensure that $a_v$ is large and negative,
which will lead to the formation of a trapped surface in a short amount of time.

The initial data at time $v_0$ are specified as follows. We choose a radius $r_*>0$, a region of perturbation
$[r_*-\delta, r_*+\delta]$ given by $\delta > 0$ small and a step function
\begin{align}\label{initialdata2}
 \Vp_0 (r) := \begin{cases}
             0,            & r < r_* - \delta, \\
             \frac{\Vs(r)}{h},           & r \in [r_* - \delta, r_* + \delta], \\
             0,       & r > r_* + \delta,
            \end{cases}
\end{align}
determined by a constant scaling factor $h = h(r_*,\delta)$. There is no perturbation assumed for the fluid density $M$,  hence
\begin{align}\label{initialdata3}
 \Mp_0 = 0, \qquad \bp_0 = 0.
\end{align}
The perturbed geometric coefficient $\ap_0$, resp.\ the initial value $a_0$, is given by the integral formula~\eqref{a-new} and the
fact that $\Vs = - \frac{\as}{2}$:
\begin{align}\label{initialdata4}
  a_0 (r) =  1 - \frac{4\pi(1+k^2)}{r} \int_{0}^{r} \frac{\bs(s)}{\bs(r)} \Ms(s) \left(1 + K^2 \left(1
  + \frac{1}{h}\chi_{[r_*-\delta,r_*+\delta]}\right) \as(s) \right) s^2 \, ds,
 \end{align}
where $\chi_{[r_*-\delta,r_*+\delta]}$ denotes the characteristic function on $[r_*-\delta,r_*+\delta]$.

\begin{proposition}[The class of initial data of interest]
\label{prop:initialdata}
Given $r_* > \varDelta>0$, there exist constants 
 $C_1,C_2,C_3>0$ depending on $r_*$ and $\varDelta$
 such that for all $\delta,h>0$ with 
 $ \frac{\delta}{h} \leq \frac{1}{C_1}$ the following holds:
\begin{align*}
 & 0 < a_0(r) \leq \as(r), \quad  r \in [0,r_*+\varDelta], \\
 & a_v(v_0,r) \begin{cases}
               = 0, & r \in [0,r_*-\delta), \\
               \leq - C_2 \frac{\delta}{h^3}, & r \in [r_*-\delta,r_*+\delta], \\
               \leq - C_3 \frac{\delta}{h}, & r \in (r_*+\delta,r_*+\varDelta].
              \end{cases}
\end{align*}
\end{proposition}

Geometrically speaking, the conclusion of Proposition~\ref{prop:initialdata} is that we have an untrapped initial data set
\eqref{initialdata2}--\eqref{initialdata4} from which, since $a_v$ is large and negative, the coefficient 
 $a$ should change sign in a small region around $r_*$ within a short amount of time, and trapped surfaces are expected to form.   

\begin{proof}
\noindent{ \bf Step 1. Positivity of $a_0$.}
The following calculations are true for all  {$\delta \leq \varDelta$},  
only the ratio of $\delta$ and $h$ is relevant. 
 Since $\as$ is positive, so is $a_0$ for $r<r_*-\delta$. Since $\Ms,\as,\bs>0$ it is immediate
 from \eqref{initialdata4} that 
 \begin{align}
  a_0(r) &= \as(r) + \ap_0(r) 
\\
& = \as(r) - \frac{4\pi(1-k^2)}{rh} \int_{r_*-\delta}^{\min(r,r_*+\delta)} \frac{\bs(s)}{\bs(r)} \Ms(s)
  \as(s) s^2 \, ds \nonumber \\ 
  &\geq \as(r) - \frac{4\pi(1-k^2)}{rh} \int_{r_*-\delta}^{r_*+\delta} \frac{\bs(s)}{\bs(r)} \Ms(s) \as(s) s^2 \, ds. \label{a0_1_n}
 \end{align}
 Recall that by Theorem~\ref{thm-static} static solutions are smooth. Choose $\eps = \eps(r_*,\varDelta)>0$ sufficiently small so that
 for all $r \in [r_*-\varDelta,r_*+\varDelta]$
 \begin{align}
  \as(r) &> \as(r_*-\varDelta) \eps.
 \label{eps1}
 \end{align}
Since $\bs$ is increasing and $\mus = \as \Ms$ is monotonically decreasing and positive by Theorem~\ref{thm-static}, we conclude from
\eqref{a0_1_n} that for $r\geq r_*-\delta$,
 \begin{align*}
  0 < -\ap_0(r) &\leq \frac{4\pi (1-k^2)}{r h} \as(r_*-\delta) \Ms(r_*-\delta) \left[ \frac{r^3}{3} \right]_{r_*-\delta}^{r_*+\delta} \\
  &\leq \frac{4 \pi (1-k^2)}{3} \frac{\delta}{h} \frac{6r_*^2 + 2\varDelta^2}{r_*-\varDelta} \as(r_*-\varDelta) \Ms(r_*-\varDelta).
 \end{align*}
 Thus, if we set $\delta, h >0 $ such that $C_1(r_*,\varDelta):= \frac{4 \pi (1-k^2)}{3} \frac{6r_*^2+2\varDelta^2}{r_*-\varDelta}
 \Ms(r_*-\varDelta) \frac{1}{\eps} \leq \frac{h}{\delta}$, then by the above assumptions
 \[
  a_0(r) \geq \as(r) - \as(r_*-\varDelta) \eps > 0, \qquad r \in [0,r_*+\varDelta].
 \]
 
\

\noindent{ \bf Step 2. Negativity of $a_v$.} By Theorem~\ref{thm-static} the static solutions are smooth and satisfy
$0 < \as = - 2 \Vs \leq 1$. By the choice of initial data \eqref{initialdata1}--\eqref{initialdata2}, $a_0$ is just the static
solution $\as$ on the central interval $[0,r_*-\delta]$. In particular, $a_v = 0$ there.
More generally, from \eqref{eq:Einstein3-new} we obtain
\begin{align}
 a_v(v_0,r) &= 2 \pi r \bs(r) \Ms(r) \left( a_0(r)^2 - 4 V_0(r)^2 \right) \nonumber \\
 &= 2 \pi r \bs(r) \Ms(r) \left( \big( \as(r) + \ap(r) \big)^2 - \as(r) \left( 1 + \chi_{[r_*-\delta,r_*+\delta]}(r) \right)^2 \right) \nonumber \\
 &= 2 \pi r \bs(r) \Ms(r) \Big( \ap_0(r) \big( a_0(r) + \as(r) \big) - \chi_{[r_*-\delta,r_*+\delta]}(r) \as(r)^2 \frac{2h + 1}{h^2} \Big), \label{eq:av0}
\end{align}
where $\chi_{[r_*-\delta,r_*+\delta]}$ again denotes the characteristic function on the interval of perturbation
$[r_*-\delta,r_*+\delta]$. Step 1 made it obvious that $a_0$ is positive on $[r_*-\varDelta,r_*+\varDelta]$ and $\ap_0$ is negative
on $(r_*-\delta,+\infty)$. Hence, we find 
\[
 a_v(v_0,r) < 0, \qquad r \in (r_*-\delta,r_*+\varDelta],
\]
independently of the size of $h$.

\

\noindent{ \bf Step 3. Bound for $a_v$ on $( r_*+ \delta,r_*+\varDelta]$.}
 To obtain finer estimates for the behavior of $a_v(v_0,r)$ we first need to estimate $\ap_0$ from above. This follows by the same
 method as in Step 1. For $r \geq r_*-\delta$, since $\mus$ is decreasing and $\bs \geq 1$ is increasing by Theorem~\ref{thm-static},
\begin{align}
 \ap_0(r) &= - \frac{4 \pi (1-k^2)}{rh} \int_{r_*-\delta}^{\min(r,r_*+\delta)} \frac{\bs(s)}{\bs(r)} \Ms(s) \as(s) s^2 \, ds \nonumber \\
 &\leq - \frac{4 \pi (1-k^2)}{r h} \frac{\mus(r_* + \delta)}{\bs(r_*+\delta)} \left[ \frac{s^3}{3} \right]_{r_*-\delta}^{\min(r,r_*+\delta)}.
 \label{eq:av01}
 \end{align}
For $r\geq r_*+\delta$, in particular,
\begin{align*}\label{eq:a0_bound}
 \ap_0(r) &\leq - \frac{4 \pi (1-k^2)}{rh} \frac{\mus(r_*+\delta)}{\bs(r_*+\delta)} \frac{\delta(6r_*^2+2\delta^2)}{3}
 \leq - 4 \pi (1-k^2) \frac{\delta}{h} \frac{(6r_*^2+2\delta^2)}{3r} \frac{\mus(r_*+\delta)}{\bs(r_*+\delta)}.
\end{align*}
Finally, we return to the explicit formula~\eqref{eq:av0} for $a_v$ at time $v_0$ and again make use of the monotonicity properties of
the static parts $\bs$ and $\mus$ as well as the conclusion of Step 1 that $0<a_0(r)+\as(r) \leq 2 \as(r)$ for all
$r \leq r_*+\varDelta$. It is then clear that $a_v$ in the interval $(r_*+\delta,r_*+\varDelta]$ is bounded above by a negative
constant times the scaling factor $\frac{\delta}{h}$ of the perturbation:
\begin{align*}
 a_v(v_0,r) &\leq 2 \pi r \bs(r) \Ms(r) \left[ - 4 \pi (1-k^2) \frac{\delta}{h} \frac{(6r_*^2+2\delta^2)}{3r}
 \frac{\mus(r_*+\delta)}{\bs(r_*+\delta)} \right] 2 \as(r) \\
 &\leq - 32 \pi^2 (1-k^2) r_*^2 \big(\mus(r_*+\varDelta)\big)^2 \frac{\delta}{h} \\
 &=: - C_3(r_*,\varDelta) \frac{\delta}{h}, \hskip3.cm  r \in (r_*+\delta,r_*+\varDelta].
\end{align*}
 
\

\noindent{ \bf Step 4. Bound for $a_v$ on $[r_*-\delta,r_*+\delta]$.} On the interval of perturbation the contribution of the first
term in the bracket of \eqref{eq:av0} is negative and tends to $0$ for $r\to r_*-\delta$ by \eqref{eq:av01} and hence is negligible.
Consequently, by making use of $\eps(r_*,\varDelta) \geq \frac{4\pi(1-k^2)}{3} \frac{6 r_*^2 +2\varDelta^2}{r_*-\varDelta}
\Ms(r_*-\varDelta) \frac{\delta}{h}$ chosen in Step 1 together with \eqref{eps1}, as well as the monotonicity properties of the static
solution,
\begin{align*}
 a_v(v_0,r) &\leq - 2 \pi \frac{2h+1}{h^2} r \bs(r) \Ms(r) \as(r)^2 \\
 &\leq - 2 \pi \bs(r_*-\varDelta) \mus(r_*+\varDelta) \as(r_*-\varDelta) \eps \frac{2h+1}{h^2}  \\
 &\leq - \frac{8 \pi^2 (1-k^2)}{3} \frac{6r_*^2+2\varDelta^2}{r_*-\varDelta} \bs(r_*-\varDelta) \mus(r_*-\varDelta) \mus(r_*+\varDelta) \frac{2h+1}{h^2}  \frac{\delta}{h} \\
 &=: - C_2(r_*,\varDelta) \frac{\delta}{h^3}, \hskip4.cm r \in [r_*-\delta,r_*+\delta]. \qedhere
\end{align*} 
\end{proof}


\subsection{Statement of the perturbation property}
\label{sec:postulates}

We are interested here in the evolution of initial data consisting of a radially symmetric, static solution which is (sufficiently
strongly) perturbed in a (sufficiently small) shell $[r_*-\delta,r_*+\delta]$. In Proposition~\ref{prop:initialdata}, we have shown
that such initial data exist which are untrapped in a central region $[0,r_*+\varDelta]$. The speed of propagation influences the
domains of dependence and let the constant $C_*>0$ be an upper bound for the (modulus) of all wave speeds $\lambda_i$ of the Euler
equations. We then define
\begin{align*}
 R_*^-(v) &= r_* - \delta - C_* (v-v_0), & R_*^+(v)&= r_* + \delta + C_* (v-v_0), \\
 \Xi_*^-(v) &= r_*-\delta + C_* (v-v_0), & \Xi_*^+(v) &= r_*+\delta-C_* (v-v_0).
\end{align*}
We assume that $[R_*^-(v),R_*^+(v)] \subseteq [r_*-\varDelta,r_*+\varDelta]$ for all times $v \in [v_0,v_*]$.

As introduced in Section~\ref{sec:random}, the effect of the perturbation $\Mp,\Vp,\ap,\bp$ then takes place in the region
\[
 \Omega_* = \{ (v,r) | v \in [v_0,v_*], r \in [R_*^-(v),R_*^+(v)] \}.
\]
Outside of this cone-like region, the fluid variables $M,V$ coincide with the unperturbed static components $\Ms,\Vs$.
A trapped surface will form during evolution when the relevant terms in $a_v$ are preserved in the ``big data'' region 
\[
 \Xi_* = \{ (v,r) | v \in [v_0,v_*], r \in [\Xi_*^-(v),\Xi_*^+(v)] \}
\] 
and the speed of propagation is sufficiently slow for the dynamical formation to take place before the region $\Xi_*$
``closes up'', that is, before we reach $v>v_0$ such that $\Xi_*^-(v) = \Xi_*^+(v)$. 
The ``mixed'' region $\triangledown_* = \Omega_* \setminus \Xi_*$, i.e.
\[
 \triangledown_* = \{ (v,r) | v \in [v_0,v_*], r \in [R_*^-(v),R_*^+(v)], r \notin (\Xi_*^-(v),\Xi_*^+(v)) \},
\]
is influenced by the outer static solution as well as the perturbation. As such, it is more difficult to control its evolution.
However, this is only relevant for the geometry variables $a,b$, which are defined as integrated quantities (from the center) using
\eqref{b-new}--\eqref{a-new}. They behave exactly like the static solutions $\as,\bs$ in the central region $[0,R_*^-(v)]$, but also
remain (slightly) perturbed on the right side of $\Omega_*$.


We will be able to establish a growth behavior of the solution $M,V,a,b$ in the domain of dependence $\Omega_*$ and the ``big data''
region $\Xi_*$ of the perturbation. By Proposition~\ref{prop:initialdata} these properties are satisfied at the initial
time $v_0$. We will check that the same behavior remains valid at each step of the random choice method introduced in
Section~\ref{sec:random}.

Before we state the desired control of the solution, let us fix an important constant that depend on the sound speed $k$, namely
 \begin{align}\label{kk}
  \kappa_0 := \frac{4k}{(1-k^2)K^2}.
 \end{align}
 
\begin{remark}\rm 
\label{rem:kappa}
It is clear that $\kappa_0 > 0$. Note that $\kappa_0 < \frac{1}{2}$ for $k$ sufficiently small.
More precisely, we need that $8k < (1-k^2) K^2 = \frac{(1-k^2)^2}{1+k^2}$, which is equivalent to
$-1 + 8k + 2k^2 +  8k^3 - k^4 < 0.$
For $k = 0$, this condition is satisfied, hence also holds in a neighborhood of $0$. Numerically, 
the smallest positive root is $k_0 \approx 0.1197$. 
We will be able to prove the formation of trapped surfaces for $k$ smaller than $k_0$.
\end{remark}

\begin{definition}
\label{def:pproperty}
An approximate solution $\Ma,\Va,\aaa,\ba$ to the Euler--Einstein system \eqref{eq:Euler}--\eqref{S1S2} is said to preserve the
\emph{perturbation property\/} if there exist constants $C_0,C, C_b,\Lambda>0$ depending only on the static solution $\Ms,\Vs,\as,\bs$,
 {a constant $\kappa > 1$ depending on $k$,}  
so that for all $v \in [v_0,v_*]$, with $v_* := v_0 + \tau h^{\kappa}$
(being the time of existence in Theorem~\ref{apptheo}, below), 
 \begin{itemize}

  \item one has in the domain of dependence $\Omega_*$:
  \begin{align*}
     \frac{1}{C_0} e^{-C \frac{v-v_0}{h^\kappa}} \left( 1 + \frac{1}{h} \right)^{-\kappa_0} \leq \Ma(v,r) &\leq C_0 e^{C\frac{v-v_0}{h^\kappa}} \left( 1 + \frac{1}{h} \right)^{\kappa_0}, \\
   \frac{1}{C_0} e^{-C\frac{v-v_0}{h^\kappa}}  \leq - \Va(v,r) &\leq C_0 e^{C\frac{v-v_0}{h^\kappa}} \left( 1 + \frac{1}{h} \right), \\ 
 {- \frac{1}{h}} \leq \aaa(v,r) \leq 1, \qquad 1 \leq \ba(v,r) &\leq C_b, \qquad - \frac{\Lambda}{h} \leq \lambda_i(v,r) \leq \frac{C_b}{2}, 
  \end{align*}
  
\item and in the ``big data'' region $\Xi_*$, the estimates for the fluid can be improved as follows: 
   \begin{align*}
       \frac{1}{C_0} e^{-C \frac{v-v_0}{h^\kappa}} \leq \Ma(v,r) &\leq C_0 e^{C \frac{v-v_0}{h^\kappa}}, \\
    \frac{1}{C_0} e^{-C \frac{v-v_0}{h^\kappa}} \left( 1 + \frac{1}{h} \right) \leq - \Va(v,r) &\leq C_0 e^{C \frac{v-v_0}{h^\kappa}} \left( 1 + \frac{1}{h} \right).
   \end{align*} 
 \end{itemize}
\end{definition}

The set of approximate solutions satisfying the perturbation property will be shown to be non-empty. First of all, the initial data
chosen in Proposition~\ref{prop:initialdata} satisfy the above bounds for
 \begin{align*}
  \log C_0 & := \tilde k \xi = \max \left( 1, 2 \kappa_0 \right) \bigg( \max_{r \in [r_*-\varDelta,r_*+\varDelta]} \left| \log \left( - \Vs(r) \right) \right| + \frac{(1+k)^2 K^4}{8k} \max_{r \in [r_*-\varDelta,r_*+\varDelta]} \left| \log \Ms(r) \right| \bigg), \\
  \log C & := \tilde k \rho = \max \left( 1, 2 \kappa_0 \right) \left(e^{1 + \kappa \xi} D_1 + D_2 \right),
\end{align*}
 with constants $D_1,D_2$ depending on $k,r_*,\varDelta$ and $C_b := \bs(r_*+\varDelta)$ (see Section~\ref{sec:7} for details). Note
 that $e^{C \frac{v-v_0}{h^\kappa}} \leq e^{C \tau} = e^{\frac{\tilde k}{\kappa}} \leq e$ only depends on $k$ and is bounded as long
 as $v \in [v_0,v_*]$ (cf.\ Remark~\ref{rem:tau}).
 
 To show that the above estimates are satisfied at each time step $[v_i,v_{i+1}]$ of the approximate solution defined through the
 random choice method, we rewrite the estimates in terms of Riemann invariant coordinates and proceed by induction in $i$. In each
 inductive step, we have to make sure that the perturbation property is preserved in the Riemann step, the ODE step and when updating
 the integral quantities $\aaa,\ba$. The details of this induction, including both the perturbation property and the BV bound stated
 below, will be carried out in Section~\ref{sec:7}. It is important to note that $\kappa$ has 
to satisfy certain
 constraints to obtain an existence result and observe the formation of trapped surfaces.
 

\subsection{Statement of the main result} 

We can now state our existence result for the Einstein--Euler equations with initial data satisfying the perturbation property. The
solutions are constructed from the initial data described in Section~\ref{sec:initialdata} which are known to satisfy the perturbation
property and be untrapped.

\begin{theorem}[A class of spherically-symmetric Einstein-Euler spacetimes with bounded variation]
\label{apptheo}
Fix $k \in (0,1)$ and $\kappa \geq 1 + 2 \kappa_0$. Given any $\mu_0>0$, let $\Ms,\Vs$ be the static solution whose density equals
$\mu_0$ at the center (cf.\ Theorem~\ref{thm-static}).
Fix any $r_*>\varDelta>0$ together with perturbation parameters $h, \delta>0$ satisfying $\delta \leq \frac{h}{C_1}$ with $C_1$ as in
Proposition~\ref{prop:initialdata}. By $Z_0 = (M_0,V_0,a_0,b_0)$, denote the initial data set
\eqref{initialdata2}--\eqref{initialdata4} consisting of a compactly supported perturbation of this static solution. 
Then there exists a constant $\tau>0$ depending upon the given static solution in the interval $[r_*-\varDelta,r_*+\varDelta]$ only, 
so that the approximate solutions $Z_\sharp$ constructed by the random choice method are well-defined on the time interval
$[v_0, v_*]$ with $v_* = v_0 + \tau h^\kappa$ and satisfy the perturbation property 
(stated in Definition \ref{def:pproperty}) within the domain $\Omega_*$ 
and, moreover, satisfy the uniform BV property (for some constant $C_2>0$) 
\be
\label{eq:20}
\sup_{v \in [v_0, v_*]} TV\big( Z_\D(v, \cdot) - Z^{(0)} \big)
\leq C_2 \,TV( Z_0 - Z^{(0)} ),
\ee
and the Lipschitz continuity property ($v, v' \in [v_0, v_*]$)
\be
\label{eq:21}
\int_{R_*^-(v)}^{R_*^+(v)} | U_\D(v,r) - U_\D(v', r)| \, dr \leq C_2 \,TV(U_0 - U^{(0)}) \, (|v - v'|+\Delta v). 
\ee
Consequently, the sequence $Z_\D$ (or a subsequence of it, at least) converges pointwise   
toward a limit $Z= (M,V, a, b)$ which is a bounded variation solution to the Euler--Einstein system in spherical symmetry and satisfies
the initial condition and the perturbation property. 
\end{theorem}

\begin{proof}[Sketch of the proof] 
The details of the proof are presented in Section~\ref{sec:7} and we only outline here the argument. The proof is based on an inductive argument along time steps following the random choice method in Section~\ref{sec:random}.

\hskip.15cm  

{\noindent \bf Step 1. \it Regular initial data.}
We use initial data that consist of a compact perturbation of a static solution. Due to Theorem~\ref{thm-static}, the static solution
is defined by $\mu_0>0$ only. A  perturbation of order $\frac{1}{h}$ is added to the fluid variable $V$ in a region
$[r_*-\delta,r_*+\delta]$. By Proposition~\ref{prop:initialdata} the initial data can be chosen to be  {untrapped} 
 in the bigger interval
$[0,r_*+\varDelta]$ if $\delta \leq \frac{h}{C_1}$ with $C_1$ being a positive constant.

\hskip.15cm  

{\noindent \bf Step 2. \it Approximate solutions satisfy the perturbation property.}
By Lemmas~\ref{lem:initial1} and \ref{lem:initial2} and Propositions~\ref{prop:omega} and \ref{prop:delta} we show that initial data
as specified in Step 1 satisfy the perturbation property at time $v_0$. We proceed by induction in time steps following the random
choice method in Section~\ref{sec:random}. Suppose the perturbation property is satisfied up to time $v_i$. We first suppose that the
geometric variables $\aaa,\ba$ are constant over time, hence their bounds are preserved. By Theorem~\ref{prop-Riemn}, the approximate
Riemann invariants $\wa,\za$ do not change their size in the Riemann problem step. In the ODE step, the Riemann invariants $\wa,\za$
only increase by a factor $C \frac{\Delta v}{h^\kappa}$ as derived in Theorems~\ref{thm:ode-omega} and \ref{thm:ode-delta}, and still
satisfy the desired bounds. Using Propositions~\ref{prop:omega} and \ref{prop:delta} the result can be converted to the fluid
quantities $\Ma,\Va$. The waves speeds $\lambda_i$ are controlled by $\frac{\Lambda}{h}$ according to 
Lemma~\ref{lem:speeds}. Finally, the geometric variables $\aaa,\ba$ are updated using the integral
formulas~\eqref{b-new}--\eqref{a-new}. By Proposition~\ref{prop:ab}, their bounds are preserved as well. Thus, the perturbation
property of Definition~\ref{def:pproperty} holds up to time $v_{i+1} \leq v_0 + \tau h^{\kappa}$.

\hskip.3cm  

{\noindent \bf Step 3. \it BV estimate and convergence.} Since our method relies on the Riemann problem associated with the Euler system
described in Section 4, and since the Riemann solutions enjoy uniform sup-norm and total variation bounds, the approximate solutions
constructed by the random choice method also enjoy such bounds \cite{BLSS,Glimm,GroahTemple}.  \qedhere 
\end{proof}


Based on this result, we now prove that solutions satisfying the perturbation property yield the dynamic formation of trapped surfaces
out of untrapped initial data.
The proof strongly relies on a careful analysis of the order of the time of existence, wave propagation and initial region of
perturbation in terms of $h$ and $\frac{1}{h}$. In a first step we investigate the behavior of $a_v$ over time, and in particular show
that it remains bounded from above by a negative constant times $\frac{1}{h^2}$. Finally, the control of the speed of propagation of
order $\frac{1}{h}$ and the time of existence of order $h^{\kappa}$ ensures the formation of a trapped surface before time $v_*$.
To observe the formation of trapped surfaces it is crucial to have $\kappa < 2$. If we choose $\kappa := 1 + 2\kappa_0$ optimal, this
is possible for small $k$ (cf.\ Remark~\ref{rem:kappa}).

\begin{corollary}[Formation of trapped surfaces]
\label{thm:trapped}
 Fix $k \in (0,k_0)$ with $k_0$ as in Remark~\ref{rem:kappa} and $\kappa < 2$. Let $\mu_0 >0$ and $r_*>\varDelta>0$ be given so that
 for the constant $C_1$ from Proposition~\ref{prop:initialdata} and for $C_0,\Lambda$ from Definition~\ref{def:pproperty},
 \begin{align}\label{initialrestriction}
  8 \pi r_* > e^3 \Lambda C_0^3 C_1.
 \end{align}
 Let $(M,V,a,b)$ be the solution associated with this initial data set (cf.~Theorem~\ref{apptheo}) under the assumption 
that $\delta = \frac{h}{C_1}$.
 Then, if $h$ is chosen to be sufficiently small, a trapped surface forms before the time $v_*$, i.e.~there exists
 $(v_\bullet,r_\bullet) \in \Xi_*, v \in (v_0,v_*)$ such that $a(v_\bullet,r_\bullet) < 0$.
\end{corollary}

\begin{proof}
 By Theorem~\ref{apptheo}, the solution to the initial value problem exists on a time interval $[v_0,v_*]$ with
 $v_* = v_0 + \tau h^\kappa$ and preserves the initial perturbation in the way stated in Definition~\ref{def:pproperty}.
 Suppose, contrary to our claim, that $a$ remains positive as long as the solution exists. Thus, in particular,
we have 
 \begin{align}\label{abound}
  0 \leq a(v,r) \leq 1, \qquad (v,r) \in \Xi_*. 
 \end{align}
 We will put together all bounds known about the solution to control $a_v$. More precisely, we will establish an upper (negative)
 bound for $a_v$ in the region $\Xi_*$ in terms of the initial data and the perturbation factor $h$.
 
 By Proposition~\ref{prop:initialdata}, $a_v$ is negative initially, at least on the interval $[r_*-\delta,r_*+\varDelta]$, and its
 size can be controlled by the perturbation constants $h$ and $\delta$. To show that $a_v$ remains negative, it is essential to
 control the term $a^2 - 4 V^2$ in \eqref{eq:Einstein3-new}. The solution satisfies the perturbation property and
 $e^{C \frac{v-v_0}{h^\kappa}} \leq e^{C \tau} \leq e$, hence for all $(v,r) \in \Xi_*, v \in [v_0,v_*]$ 
 \[
  4 V(v,r)^2 \geq \frac{4}{C_0^2} e^{- 2C \frac{v-v_0}{h^\kappa}} \left( 1 + \frac{1}{h} \right)^2 \geq \frac{4}{e^2 C_0^2}
  \left( 1 + \frac{1}{h} \right)^2 \geq \frac{8}{e^2 C_0^2 h} + \frac{4}{e^2 C_0^2 h^2}.
 \]
 Choose $h$ sufficiently small to have $\frac{8}{e^2 C_0^2} \geq h$. Then, by \eqref{abound},
 \[
  a(v,r)^2 - 4V(v,r)^2 \leq 1 - 1 - \frac{4}{e^2 C_0^2 h^2} \leq - \frac{4}{e^2 C_0^2 h^2},
 \]
 which yields
 \begin{align*}
  a_v(v,r) &= 2\pi r b(v,r) M(v,r) \left( a(v,r)^2 - 4 V(v,r)^2 \right) \\
  &\leq - 2 \pi r \frac{1}{C_0} e^{- C \tau} \frac{4}{e^2 C_0^2 h^2} \leq -  \frac{8\pi r}{e^3 C_0^3 h^2} .
 \end{align*}
 Integrating $a$ in time implies
 \begin{align*}
  a(v,r) = a_0(r) + \int_{v_0}^v a_v(w,r) dw \leq 1 - \frac{8\pi r}{e^3 C_0^3 h^2} (v-v_0).
 \end{align*}
 We need to make sure that $v$ is sufficiently small to still have $\Xi_*^-(v) < \Xi_*^+(v)$. By
 Lemma~\ref{lem:eps}, the assumption on $k$ and \eqref{initialrestriction}, there exists $\varepsilon>0$ so that
 \[
  \frac{e^3 C_0^3}{8 \pi r_*} < \varepsilon \leq \frac{1}{\Lambda C_1} < \tau h^{-2+\kappa}.
 \]
 Thus, for $v_\bullet := v_0 + \varepsilon h^2 < v_0 + \tau h^\kappa = v_*$
 and $r_\bullet := r_* \in  {[\Xi_*^-(v_\bullet),\Xi_*^+(v_\bullet)]}$, we find 
 \[
  a(v_\bullet,r_\bullet) \leq 1 - \frac{8\pi r_*}{e^3 C_0^3} \frac{\varepsilon h^2}{h^2} < 0. \qedhere
 \] 
\end{proof}

\begin{lemma}[Size of $\Xi_*$]\label{lem:eps}
 Let $h$ be the perturbation parameter in the region $[r_*-\delta,r_*+\delta]$, $\delta = \frac{h}{C_1}$, with constant $C_1$ from
 Proposition~\ref{prop:initialdata}. Let $M,V,a,b$ be a solution to \eqref{eq:Euler}--\eqref{S1S2} as derived in Theorem~\ref{apptheo}.
 Then $\Xi_*^-(v) \leq \Xi_*^+(v)$ at least as long as
 \begin{align}\label{eq:delta}
  \delta \geq C_* (v-v_0) = \Lambda \frac{v-v_0}{h},
 \end{align}
 with the constant $\Lambda>0$ as defined in Lemma~\ref{lem:speeds}.
  In particular, the above estimate holds for all $v \leq v_0 + \varepsilon h^2$ with $\varepsilon := \frac{1}{\Lambda C_1}
  < \tau h^{-2+\kappa}$.
\end{lemma}

\begin{proof}
In Lemma~\ref{lem:speeds} we derive a bound for the wave speeds in the bigger region $\Omega_*$.
For $h$ small, the lower bound is much larger than the upper bound and $C_* = \frac{\Lambda}{h}$
 is an upper bound for the modulus of all wave speeds of the (homogeneous) Euler equations. 
 Hence the initial region $[r_*-\delta,r_*+\delta]$ does not close up as long as \eqref{eq:delta} holds. In terms of $v$, this yields
 the proposed bound in the statement since
 \begin{align*}
  C_* (v-v_0) &\leq \Lambda h^{-1} \varepsilon h^2 \leq \Lambda \varepsilon h = \frac{h}{C_1} \leq \delta. \qedhere
 \end{align*}
\end{proof}

It remains to prove the existence of the initial conditions $\mu_0,r_*,\varDelta,k$ satisfying \eqref{initialrestriction} in
Corollary~\ref{thm:trapped}.

\begin{proposition}[Existence of initial data]
\label{prop67}
 There exist initial data sets that satisfy the requirements of Corollary~\ref{thm:trapped}.  
\end{proposition}

\begin{proof}
 Note that for $k \leq k_0$ as in Remark~\ref{rem:kappa}, $\frac{1+k}{1-k} \leq 2$ and the relevant constants of the perturbation
 property are
 \begin{align*}
  \log C_0 &= \max_{r \in [r_*-\varDelta,r_*+\varDelta]} \left| \log \left( - \Vs(r) \right) \right| + \frac{(1+k)^2 K^4}{8k}
  \max_{r \in [r_*-\varDelta,r_*+\varDelta]} \left| \log \Ms(r) \right|, \\
  C_b &= \bs(r_*+\varDelta).
 \end{align*}
 For $k$ small, the term $\frac{(1+k)^2 K^4}{8k}$ is very large and diverging as $k \to 0$. For $k_0 \approx 0.1197$, the value is
 $\frac{(1+k)^2}{2 (1+k^2)} \approx  1.236$. Hence we may assume that $\frac{(1+k)^2 K^4}{8k} \leq 2$ for $k$ greater than some
 $k_1<k_0$ and only consider such values. Fix the initial value for the static solution with $\mu_0=\frac{1}{\sqrt{2}}$. Then, by
 Theorem~\ref{thm-static}, $\Ms(0) = \mu_0 = \frac{1}{\sqrt{2}}$ and $\Vs(0) = - \frac{1}{2}$ and both $\Ms$ and $\Vs$ are smooth and
 do not change their sign. Thus, there exists a radius $r_1>0$ such that for $r \in [0,r_1]$,
 \[
 - 1  {\leq}  -\Vs(r) \leq - \frac{1}{4}, \qquad \frac{1}{2} \leq \Ms(r) \leq 1.
 \]
 Then, we have 
\begin{align*}
 \log \widetilde C_0 &\! := \max_{r \in [0,r_1]} \left| \log \left( - \Vs(r) \right) \right| + \frac{(1+k)^2 K^4}{8k}
 \max_{r \in [0,r_1]} \left| \log \Ms(r) \right| \\
 &\leq - 2 \log 2 - 2 \log 2 = - 4 \log 2,
\end{align*}
 and therefore $\widetilde C_0 \leq e^{-4 \log 2} = \frac{1}{16}$ is an upper bound for $C_0$ if $r_*+\varDelta < r_1$. Similarly we
 estimate $C_b$ by using the integral formula \eqref{b-new}. If we assume, without loss generality, that $r_1 \leq \frac{1}{2}$, then
 \begin{align*}
  \widetilde C_b & \! := \bs(r_1) = e^{4 \pi (1+k^2) \int_0^{r_1} \Ms(s) s ds} \leq e^{2 \pi \left( 1 + \frac{1}{64} \right) r_1^2}
  \leq e^2 < 8.
 \end{align*}
 We turn to the constant $C_1$, which appeared in Step 1 of the proof of Proposition~\ref{prop:initialdata} to control the ratio
 between $\delta$ and $h$. Recall that the static solution satisfies $\as = - 2 \Vs$. Thus by the above, for $r \in [0,r_1]$,
 \[
  \frac{1}{2} \leq \as(r) \leq 2
 \]
 and $\as(r) \geq \as(r_*-\varDelta) \eps$ is satisfied for $\eps = \frac{1}{4}$ if $r_*+\varDelta \leq r_1$. The latter condition
 is, for example, satisfied for $r_* := \frac{(n-1)r_1}{n}$ and $\varDelta := \frac{r_1}{n}$, where $n$ is a large natural number that will
 be specified later. For such a choice of $r_*,\varDelta$ we can estimate the constant $C_1$ from above by a very small number (for
 $n$ large)
 \begin{align*}
  \widetilde C_1 &\! := \frac{4 \pi}{3} \frac{6 r_*^2 + 2 \varDelta^2}{r_*-\varDelta} 4 = \frac{16 \pi}{3} \frac{6n-4}{n(n-3)} r_1.
 \end{align*}
 Thus, we can estimate \eqref{initialrestriction} by
 \begin{align*}
  e^3 C_1 C_0^3 \Lambda 
&\leq e^3 \widetilde C_1 \widetilde C_0^3 \widetilde \Lambda   
\leq e^3 \widetilde C_1 \widetilde C_0^3 \widetilde C_b \left( 12 \widetilde C_0 + 1\right)
\\
&   \leq e^3 \frac{16 \pi}{3} \frac{6n - 4}{n(n-3)} r_1 \frac{1}{16^3} 8 \left( 1 + 1 \right) = \frac{e^3}{48} \frac{6n - 4}{n(n-3)}
\pi r_1.
 \end{align*}
 For $n=4$, we have 
 \[
  e^3 C_1 C_0^3 \Lambda \leq \frac{5 e^3}{48} \pi r_1 = \frac{20 e^3}{144} \pi \frac{3 r_1}{4} \leq 3 \pi r_*. \qedhere
 \]
\end{proof}

\begin{remark} \rm 
 All statements above assume that the radial component of $\Omega_*$ is contained in $[0,r_*+\varDelta]$. It is
 easy to see that $\Omega_*$ does not go beyond $r_*+\varDelta$  {during the relevant time interval $[v_0,v_*]$}. 
We only need to show that
 \begin{align}\label{domain}
  C_* (v-v_0) \leq (r_* + \varDelta) - (r_*+\delta) = \varDelta - \delta, \qquad v \in [v_0,v_*],
 \end{align}
 with $v_*=v_0+\tau h^\kappa$.
 Due to Proposition~\ref{prop:initialdata}, $\delta$ was chosen to be less than or equal to $\frac{h}{C_1}$. Moreover, by
 Lemma~\ref{lem:speeds}, $C_* = \frac{\Lambda}{h}$ for some positive constant $\Lambda$ that only depends on the static solution
 within the interval $[0,r_*+\varDelta]$ (see Section~\ref{sec:7}).
 Note that $\tau = \frac{1}{\kappa \rho}$ also only depends on the initial data in the region $[0,r_*+\varDelta]$ (again, see
 Section~\ref{sec:7} for the details). Hence we may, without loss of generality,
 assume that $h\leq 1$ is sufficiently small to have
 \[
  h^{\min(\kappa-1,1)} \leq \frac{C_1 \varDelta}{1+\Lambda \tau C_1}.
 \]
 Then, \eqref{domain} holds for all $v \leq v_* = v_0 + \tau h^\kappa$:
 \begin{align*}
   C_* (v-v_0) = \frac{\Lambda}{h} \tau h^{\kappa} \leq h^{\kappa-1} \Lambda \tau
   \leq \frac{\varDelta \Lambda \tau C_1}{1 + \Lambda \tau C_1}
   \leq \varDelta - \frac{\varDelta}{1+\Lambda \tau C_1} \leq \varDelta - \frac{h^{\min(\kappa-1,1)}}{C_1} \leq \varDelta - \delta.
 \end{align*}
\end{remark}


\section{Completion of the proof of the main result} 
\label{sec:7} 

\subsection{A formulation in terms of the Riemann invariants}

The fluid variables $M,V$ may be expressed in terms of the Riemann invariants $w,z$ \eqref{Rinv} as follows:
\begin{align}\label{MV-wz}
 \log M = \frac{4k}{(1-k^2)K^2} (z-w), \qquad \log (-V) = \frac{K^2}{2} \left( \frac{1-k}{1+k} w + \frac{1+k}{1-k} z \right).
\end{align}
Initially, the sum and difference of the Riemann invariants can therefore be controlled by the static solution $\Ms,\Vs$ and the
perturbation term $1 + \frac{1}{h}$. From this, and the ODE step in the following section, the derive the bounds for $w,z$ by induction
in time (steps). Here we only check that these bounds are satisfied initially and imply the perturbation property stated for $M,V$ in
Definition~\ref{def:pproperty}.

\begin{lemma}[Initial condition in $\Omega_*$] \label{lem:initial1}
 Let $\xi$ be a positive constant depending on the static solution in the interval $[r_*-\varDelta,r_*+\varDelta]$,
 \begin{align}\label{xi_1}
 \xi = \xi^{(0)}_\varDelta := \max_{r \in [r_*-\varDelta,r_*+\varDelta]} \left| \log \left( - \Vs(r) \right) \right|
 + \frac{(1+k)^2K^4}{8k} \max_{r \in [r_*-\varDelta,r_*+\varDelta]} \left| \log \Ms(r) \right|.
\end{align}
Then, initially, the (approximate) Riemann invariants satisfy
\[
 \wa(v_0,r),\za(v_0,r) \in \left[ - \xi , \log \left( 1+ \frac{1}{h} \right) + \xi \right].
\]
\end{lemma}

\begin{proof}
 The region $\Omega_*$ is influenced by the static solution as well as the perturbation. By \eqref{MV-wz}, the Riemann invariants
 $w,z$ can initially at time $v_0$ be controlled by
\begin{align*}
 \frac{(1-k^2)K^2}{4k} \log \left( \min_{r\in[r_*-\varDelta,r_*+\varDelta]} \Ms(r) \right) & \leq \za - \wa \leq \frac{(1-k^2)K^2}{4k}
 \log \left( \max_{r\in[r_*-\varDelta,r_*+\varDelta]} \Ms(r) \right),
\end{align*}
and
\begin{align*}
 \frac{2}{K^2} \log \left( - \max_{r\in[r_*-\varDelta,r_*+\varDelta]} \Vs(r) \right) & \leq \frac{1-k}{1+k} \wa + \frac{1+k}{1-k} \za \\
 & \leq  \frac{2}{K^2} \log \left( - \min_{r\in[r_*-\varDelta,r_*+\varDelta]} \Vs(r) \right) + \frac{2}{K^2} \log
 \left( 1 + \frac{1}{h} \right).
\end{align*}
Adding up both inequalities implies bounds for $\wa$ and $\za$. For $\za$ the upper bound reads
\begin{align*}
 \za &= \frac{K^2}{2} \left[ \frac{1-k}{1+k} (\za-\wa) + \left( \frac{1-k}{1+k} \wa + \frac{1+k}{1-k} \za \right) \right] \\
 &\leq \frac{(1-k)^2K^4}{8k} \log \left( \max_{r\in[r_*-\varDelta,r_*+\varDelta]} \Ms(r) \right)
 + \log \left( - \min_{r\in[r_*-\varDelta,r_*+\varDelta]} \Vs(r) \right) + \log \left( 1 + \frac{1}{h} \right),
\end{align*}
and the lower bound is
\begin{align*}
 \za &\geq \frac{(1-k)^2 K^4}{8k} \log \left( \min_{r\in[r_*-\varDelta,r_*+\varDelta]} \Ms(r) \right) 
 +  \log \left( - \max_{r\in[r_*-\varDelta,r_*+\varDelta]} \Vs(r) \right).
\end{align*}
The bounds for $\wa$ are
\begin{align*}
 \wa &= \frac{K^2}{2} \left[ - \frac{1+k}{1-k} (\za-\wa) + \left( \frac{1-k}{1+k} \wa + \frac{1+k}{1-k} \za \right) \right] \\
 &\leq - \frac{(1+k)^2K^4}{8k} \log \left( \min_{r\in[r_*-\varDelta,r_*+\varDelta]} \Ms(r) \right)
 + \log \left( - \min_{r\in[r_*-\varDelta,r_*+\varDelta]} \Vs(r) \right) + \log \left( 1 + \frac{1}{h} \right)
\end{align*}
 and
\begin{align*}
 \wa &\geq - \frac{(1+k)^2K^4}{8k} \log \left( \max_{r\in[r_*-\varDelta,r_*+\varDelta]} \Ms(r) \right)
 + \log \left( - \max_{r\in[r_*-\varDelta,r_*+\varDelta]} \Vs(r) \right).
\end{align*}
Thus in the first Riemann problem step, the Riemann invariants are contained in the large region
\[
 \wa,\za \in \left[ - \xi, \log \left(1 + \frac{1}{h} \right) + \xi \right],
\]
with $\xi$ as defined in \eqref{xi_1}. Because of Theorem~\ref{thm-static}, $\xi$ is well-defined and positive.
\end{proof}

During each Riemann problem step, by Proposition~\ref{prop-Riemn}, the Riemann invariants remain unchanged. Later we show that in each
ODE step $\wa,\za$ may only increase by a factor $\rho \frac{\Delta v}{h^\kappa}$, where $\rho$ is a positive constant that can be
uniformly chosen for all time steps. As such we expect for all times $v \in [v_0,v_*]$ that the perturbation properties
\eqref{wz_pp_omega} stated below remains true in each step of the random choice method. It thus remains to be shown that
\eqref{wz_pp_omega} determines the perturbation property in $\Omega_*$ as stated in Definition~\ref{def:pproperty}.

\begin{proposition}[Conversion in $\Omega_*$]\label{prop:omega}
 Suppose the (approximate) Riemann invariants satisfy 
 \begin{align}\label{wz_pp_omega}
  \wa,\za \in \left[ - \xi - \rho \frac{v-v_0}{h^\kappa}, \log \left( 1+ \frac{1}{h} \right) + \xi + \rho \frac{v-v_0}{h^\kappa} \right],
 \end{align}
 with $\xi>0$ as in \eqref{xi_1} and $\rho>0$, $\kappa \geq 1$ as specified in Theorem~\ref{thm:ode-omega}. Then the corresponding
 (approximate) solution $\Ma,\Va$ has the perturbation property of Definition~\ref{def:pproperty} in the region $\Omega_*$ with
 $\tilde k = \max \left( 1, \frac{8k}{(1-k^2)K^2} \right)$, $C_0:= e^{\tilde k \xi}$ and $C := \tilde k \rho$.
\end{proposition}

\begin{proof}
 By \eqref{MV-wz} and \eqref{wz_pp_omega},
$$
\aligned
   \left( 1 + \frac{1}{h} \right)^{-\frac{4k}{(1-k^2)K^2}} 
 e^{-\frac{8k}{(1-k^2)K^2} \xi} e^{-\frac{8k}{(1-k^2)K^2} \rho
   \frac{v-v_0}{h^\kappa}}  
& = e^{\frac{4k}{(1-k^2)K^2} \left( - \log \left( 1 + \frac{1}{h} \right) -2\xi
   - 2 \rho \frac{v-v_0}{h^\kappa} \right)} \\
  &\leq \Ma(v,r) = e^{\frac{4k}{(1-k^2)K^2} (\za-\wa)} \\
  &\leq e^{\frac{4k}{(1-k^2)K^2} \left( \log \left( 1 + \frac{1}{h} \right) + 2\xi + 2 \rho \frac{v-v_0}{h^\kappa} \right)}  = \left( 1 + \frac{1}{h} \right)^\frac{4k}{(1-k^2)K^2} e^{\frac{8k}{(1-k^2)K^2} \xi} e^{\frac{8k}{(1-k^2)K^2} \rho
  \frac{v-v_0}{h^\kappa}}
 \endaligned
$$
 and
 \begin{align*}
 e^{-\xi} e^{- \rho \frac{v-v_0}{h^\kappa}}
 \leq - \Va(v,r) = e^{\frac{K^2}{2} \left( \frac{1-k}{1+k} w + \frac{1+k}{1-k} z \right)}
 \leq \left( 1 + \frac{1}{h} \right) e^\xi e^{\rho \frac{v-v_0}{h^\kappa}}.
 \end{align*}
 With $\tilde k, \kappa, C_0,C$ as specified in the statement, $\Ma,\Va$ satisfy the perturbation property in $\Omega_*$.
\end{proof}

We now turn to the estimates in the region $\Xi_*$ which are obtained in a similar fashion.

\begin{lemma}[Initial condition in $\Xi_*$]\label{lem:initial2}
 Let $\xi$ be a positive constant depending on the static solution in the interval $[r_*-\delta,r_*+\delta]$,
\begin{align}\label{xi_2}
 \xi = \xi^{(0)}_\delta := \max_{r \in [r_*-\delta,r_*+\delta]} \left| \log \left( - \Vs(r) \right) \right|
 + \frac{(1+k)^2K^4}{8k} \max_{r \in [r_*-\delta,r_*+\delta]} \left| \log \Ms(r) \right|.
\end{align}
Then, initially, the (approximate) Riemann invariants satisfy
\[
 \wa(v_0,r),\za(v_0,r) \in \log \left( 1+ \frac{1}{h} \right) + \left[ - \xi ,  \xi \right].
\]
\end{lemma}

Note that $\xi^{(0)}_\delta \leq \xi^{(0)}_\varDelta$. We may thus use $\xi := \xi^{(0)}_\varDelta>0$ throughout for the definition
of the constants in Definition~\ref{def:pproperty}.

\begin{proof}
 The ``big data'' region $\Xi_*$ is solely influenced by the perturbation, and the relevant terms are
\begin{align*}
 \frac{(1-k^2)K^2}{4k} \log \left( \min_{r\in[r_*-\delta,r_*+\delta]} \Ms(r) \right)
 & \leq z - w \leq \frac{(1-k^2)K^2}{4k} \log \left( \max_{r\in[r_*-\delta,r_*+\delta]} \Ms(r) \right)
\end{align*}
and 
\begin{align*}
 \frac{2}{K^2} \log \left( - \max_{r\in[r_*-\delta,r_*+\delta]} \Vs(r) \right)& + \frac{2}{K^2} \log \left( 1 + \frac{1}{h} \right) \nonumber \\
\leq \frac{1-k}{1+k} w + \frac{1+k}{1-k} z \leq & \, \frac{2}{K^2} \log \left( - \min_{r\in[r_*-\delta,r_*+\delta]} \Vs(r) \right)
+ \frac{2}{K^2} \log \left( 1 + \frac{1}{h} \right).
\end{align*}
Therefore, the Riemann invariants in the first Riemann problem step are bounded by 
(with a constant $\xi$ as in \eqref{xi_2}) so that 
$w,z \in \log \left( 1 + \frac{1}{h} \right) + [- \xi,\xi].$
\end{proof}

\begin{proposition}[Conversion in $\Xi_*$]\label{prop:delta}
Suppose the (approximate) Riemann invariants satisfy
 \begin{align}\label{wz_pp}
  \wa,\za \in \log \left( 1+ \frac{1}{h} \right) + \left[ - \xi - \rho \frac{v-v_0}{h^\kappa}, \xi + \rho \frac{v-v_0}{h^\kappa} \right]
 \end{align}
with $\xi>0$ as in \eqref{xi_2} and $\rho>0$, $\kappa \geq 1$ as specified in Theorem~\ref{thm:ode-delta}. Then the corresponding
(approximate) solution $\Ma,\Va$ has the perturbation property of Definition~\ref{def:pproperty} in the region $\Omega_*$ with
$\tilde k = \max \left( 1, \frac{8k}{(1-k^2)K^2} \right)$, $C_0:= e^{\tilde k \xi}$ and $C := \tilde k \rho$.
\end{proposition}

\begin{proof}
We use \eqref{MV-wz} to translate the property \eqref{wz_pp} back to $\Ma,\Va$. Since the bounds are symmetric it is sufficient to
consider the upper bounds,
\begin{align*}
 \Ma(v,r) &\leq e^{\frac{4k}{(1-k^2)K^2} 2 \left( \xi + \rho \frac{v-v_0}{h^\kappa} \right)} = e^{\frac{8k}{(1-k^2)K^2} \xi}
 e^{\frac{8k}{(1-k^2)K^2} \rho \frac{v-v_0}{h^\kappa} }, \\
 -\Va(v,r) &\leq e^{\frac{K^2}{2} \frac{2}{K^2} \left( \log \left(1 + \frac{1}{h} \right) + \xi + \rho \frac{v-v_0}{h^\kappa}  \right)} = \left(1 + \frac{1}{h} \right) e^{\xi} e^{\rho \frac{v-v_0}{h^\kappa} }.
\end{align*}
Let $\tilde k = \max \left( 1, \frac{8k}{(1-k^2)K^2} \right)$. By defining $C_0 := e^{\tilde k \xi}$ and $C := \tilde k \rho$, it is clear
that Definition~\ref{def:pproperty} is true for $\Ma,\Va$.
\end{proof}

  It remains to be shown that $\wa,\za$ satisfy  {\eqref{wz_pp_omega} and}
\eqref{wz_pp} during the evolution and that $\aaa,\ba$ also satisfy the perturbation
  property. 
 

\subsection{The Riemann invariant bounds 
in each ODE step}

In each Riemann problem step, the Riemann invariants $\wa,\za$ are non-increasing and \eqref{wz_pp_omega} and \eqref{wz_pp} are
preserved. In each ODE step, the sup-norm of $\wa,\za$ may only increase by a factor $\rho \frac{\Delta v}{h^\kappa}$. By iterating
our estimates within the time interval $[v_0,v_*]$, we obtain the desired uniform bounds.

We consider the nonlinear system of ordinary differential equations in the $v$-variable, that is, 
\[
 \partial_v U = \widetilde S(U,a,b), 
\]
with conservative variable $U$ and right-hand side $\widetilde S (U,a,b)$ as derived in \eqref{deftildeS}. Here, the geometry terms
$a = a(v)$ and $b = b(v)$ are assumed to be (regular) functions of $v$, only. In particular, $a_v$ satisfies \eqref{eq:Einstein3-new}.
We will show that the solutions to these equations satisfy the perturbation property, thus in particular the physical bounds
$M > 0$ and $V < 0$ hold and they cannot blow up in finite time. We will work with the variables $w,z$ and prove that they remain bounded on every bounded time
 interval. First we establish that the sup-norm of the approximate solutions $\wa,\za$ remains uniformly bounded.

Let us write the spatially independent solutions $\partial_v U = \widetilde S(U)$ in terms of the fluid variables $M,V$. By
Proposition~\ref{prop:23}, \eqref{eq:Einstein3-new} and \eqref{deftildeS}, we have 
\begin{align*}
 M_v = ~& \widetilde S_1 = - b M \left[ \frac{1}{r} ( 1+ 2V) + 8 \pi (1-k^2) r M V \right], \\
 V_v = ~& \frac{1}{K^2 M} \left[ \widetilde S_2 - \widetilde S_1 \left( \frac{a}{2} + K^2 V \right) - M \frac{a_v}{2} \right] \\
 = ~& - \frac{b}{K^2} \left[  \frac{1}{r} \left( (a-K^2) V  + 2 (1-K^2) V^2 \right) + 16 \pi k^2 r M V^2 \right]
\end{align*}

 Alternatively we may write the ODE system in terms of the Riemann invariant coordinates $w,z$. The first equation,
 $\del_v U_1 = \widetilde S(U)_1$, implies 
\begin{align}\label{eq:wv-zv}
 w_v - z_v & = - \frac{(1-k^2) K^2}{4k} ( \log M)_v = \frac{(1-k^2)K^2}{4k} b \left[ \frac{1}{r} + \frac{2}{r} V + 8 \pi (1-k^2) r M V \right], 
\end{align}
 and the second equation, $\del_v U_2 = \widetilde S(U)_2$, together with \eqref{eq:Einstein3-new} implies
\begin{align}\label{eq:wv+zv}
 \frac{1-k}{1+k} w_v + \frac{1+k}{1-k} z_v &= \frac{2}{K^2} \left( \log (-V) \right)_v 
 = - \frac{2b}{K^4} \left[  \frac{1}{r} (a-K^2) + \frac{2 (1-K^2)}{r} V + 16 \pi k^2 r M V \right]
 \end{align}
 Adding up the two equations \eqref{eq:wv-zv}--\eqref{eq:wv+zv}, we thus obtain a system of two nonlinear ordinary differential equations for for $w$ and $z$,
 which is used to prove that the $w,z$ remain under control by the initial data
 in each ODE step. We then estimate the equations for $w,z$ by a Riccati type equation to gain the desired bounds for $w_v,z_v$.

\begin{lemma}[Estimates for $w,z$]
 Let $k \in (0,1)$ and $\kappa_0 := \frac{4k}{(1-k^2)K^2}$. Suppose that $1 \leq b \leq C_b$ as well as
  {$-\frac{1}{h} \leq a \leq 1$ hold}.  
If $w,z$ satisfy the nonlinear ordinary differential system
 \eqref{eq:wv-zv}--\eqref{eq:wv+zv}, then  for $h \leq 1$
\begin{align}\label{eq:bounds}
 \pm w_v,z_v \leq A_1 + \frac{A_2}{ {h}}  
+ A_3 e^{\max(w,z)} + A_4 e^{(1+\kappa_0) \max(w,z)} e^{-\kappa_0 \min(w,z)}, 
\end{align}
 with some expressions $A_i>0$ that only depend upon $k,r_*,\varDelta,C_b$.
\end{lemma}

\begin{proof}
 Adding and subtracting \eqref{eq:wv-zv}--\eqref{eq:wv+zv} in a suitable way yields equations for $w_v,z_v$, i.e.
\begin{align*}
 w_v &= \frac{K^2}{2} \left( \frac{1-k}{1+k} w_v + \frac{1+k}{1-k} z_v \right) + \frac{(1+k)^2}{2(1+k^2)} (w_v - z_v), \\
 z_v &= \frac{K^2}{2} \left( \frac{1-k}{1+k} w_v + \frac{1+k}{1-k} z_v \right) - \frac{(1-k)^2}{2(1+k^2)} (w_v - z_v).
\end{align*}
 Both equations exhibit a very similar structure, and to obtain upper and lower bounds for $w_v,z_v$ it thus remains to estimate the right-hand sides
$\frac{K^2}{2}$\eqref{eq:wv+zv}, $\frac{(1+k)^2}{2(1+k^2)}$\eqref{eq:wv-zv} and $\frac{(1-k)^2}{2(1+k^2)}$\eqref{eq:wv-zv}.
 
 By assumption and for $h$ sufficiently small,  {$-\frac{1}{h} \leq - 1 \leq a \leq 1$} and $1 \leq b \leq C_b$. We use \eqref{MV-wz} to replace
 $M,V$ by expressions in $w,z$. Thus, we have 
 \begin{align*}
\frac{K^2}{2} \left| \frac{1-k}{1+k} w_v + \frac{1+k}{1-k} z_v \right| 
 &\leq \frac{C_b}{K^2} \left[  \frac{1}{r} \left( \frac{1}{ {h}}+K^2 \right) + \frac{2 (1-K^2)}{r} e^{\frac{K^2}{2}
 \left(\frac{1-k}{1+k} w + \frac{1+k}{1-k} z \right)}  + 16 \pi k^2 r e^{\kappa_0 (z-w)} e^{\frac{K^2}{2} \left(\frac{1-k}{1+k} w + \frac{1+k}{1-k} z \right)} \right] \\
  & \leq \frac{C_b}{r} + \frac{C_b}{K^2 r} \frac{1}{ {h}} + \frac{2 (1-K^2) C_b}{K^2 r} e^{\max(w,z)} + \frac{16 \pi k^2 r C_b}{K^2} e^{(1+\kappa_0) \max(w,z)} e^{-\kappa_0 \min(w,z)} 
 \end{align*}
 and, similarly,
 \begin{align*}
  &\frac{(1-k)^2}{2(1+k^2)} \left| w_v - z_v \right| \leq \frac{(1+k)^2}{2(1+k^2)} \left| w_v - z_v \right| \\
  & \leq \frac{(1+k)^2}{2(1+k^2)} \frac{(1-k^2)K^2}{4k} C_b \left[ \frac{1}{r} + \frac{2}{r} e^{\frac{K^2}{2}
 \left(\frac{1-k}{1+k} w + \frac{1+k}{1-k} z \right)} + 8 \pi (1-k^2) r e^{\kappa_0 (z-w)} e^{\frac{K^2}{2}
 \left(\frac{1-k}{1+k} w + \frac{1+k}{1-k} z \right)} \right] \\
  & \leq \frac{(1+k)^2 K^4 C_b}{8kr} + \frac{(1+k)^2 K^4 C_b}{4kr} e^{\max(w,z)} 
+ \frac{\pi (1+k)^2 (1-k^2) K^4 rC_b}{k} e^{(1+\kappa_0) \max(w,z)} e^{- \kappa_0 \min(w,z)}.
 \end{align*}
 Therefore,
\eqref{eq:bounds} holds with constants
 \begin{align*}
  A_1 &:= \frac{C_b}{r_*-\varDelta} + \frac{(1+k)^2 K^4 C_b}{8k(r_*-\varDelta)},
  & A_3 &:= \frac{2 (1-K^2) C_b}{K^2 (r_*-\varDelta)} + \frac{(1+k)^2 K^4 C_b}{4k(r_*-\varDelta)}, \\
  A_2 &:= \frac{C_b}{K^2 (r_*-\varDelta)},
  & A_4 &:= \frac{16 \pi k^2 (r_*+\varDelta) C_b}{K^2} + \frac{\pi (1+k)^2 (1-k^2) K^4 (r_*+\varDelta) C_b}{k}. \qedhere
 \end{align*}
\end{proof} 

\begin{theorem}[Bounds for the ODE step in $\Omega_*$]\label{thm:ode-omega}
 Fix $k \in (0,1)$ and $\kappa \geq 1 + 2\kappa_0$. Suppose $1 \leq b \leq C_b$ and $-  {\frac{1}{h}} \leq a \leq 1 $.  
  Then there exists $\rho>0$ so that the (approximate) Riemann invariants $\wa,\za$ obeying the differential system
  \eqref{eq:wv-zv}--\eqref{eq:wv+zv} with initial values as derived in Lemma~\ref{lem:initial1} satisfy
  \begin{align}\label{newpp2}
   \wa,\za \in \left[ - \xi - \rho \frac{v-v_0}{h^\kappa},  \log \left( 1 + \frac{1}{h} \right) + \xi + \rho \frac{v-v_0}{h^\kappa} \right]
  \end{align}
 for all $v \in [v_0,v_*]$ with $v_* := v_0+\tau h^\kappa$, $\tau := \frac{1}{\kappa \rho}$, and $\xi$ as defined in \eqref{xi_1}.
 \end{theorem}
 
  \begin{remark}\rm
\label{rem:tau}
  The parameter $\tau$ for the time of existence can be estimated by $\frac{1}{\kappa}$ if we assume that $\rho$ is always chosen
  greater than $1$. Moreover, by definition of $C := \tilde k \rho$, we see that $e^{C \tau}$ is always independent of $\rho$.
  More precisely, $e^{C \frac{v-v_0}{h^\kappa}} \leq e^{C \tau} = e^{\frac{\tilde k}{\kappa}} \leq e$ holds for all
  $v \in [v_0,v_*]$, since $\tilde k = \max \left( 1, \frac{8k}{(1-k^2)K^2} \right) \leq 1 +  2 \kappa_0 \leq \kappa$.
 \end{remark}

\begin{proof}
{\noindent \bf Step 1. Linearizing the nonlinear ODE system.} 
 Let us assume that $w,z$ are bounded by some function $\gamma(v)$ so that
 \begin{align}\label{eq:gamma}
  w(v), z(v) \in \left[ -\gamma(v) , \log \left( 1 + \frac{1}{h} \right) + \gamma(v) \right].
 \end{align}
 We therefore get by \eqref{eq:bounds},
 \begin{align*}
  \pm w_v,z_v \leq ~& A_1 + \frac{A_2}{ {h}} + \frac{2 A_3 (1+h)}{h} e^{\gamma}
  + \frac{2 A_4 (1+h)^{1+\kappa_0}}{h^{1+\kappa_0}} e^{(1+2\kappa_0) \gamma} ,
 \end{align*}
 Without loss of generality we assume that $h$ is sufficiently small to satisfy $(1+h)^{1+\kappa_0} \leq 2$. This condition only depends
 on $k$ and does not disturb the inductive argument.
 For $\kappa \geq 1 + 2\kappa_0$, the inequalities are still satisfied and we get that $\gamma$ must satisfy the differential equation
 \[
  \gamma_v = \frac{D_1}{ {h}} + \frac{D_2}{h^{1+\kappa_0}} e^{\kappa \gamma}
 \]
 for some expressions $D_1,D_2$ depending on $k,r_*,\varDelta,C_b$. Introducing $g = e^{\kappa \gamma}$ yields a Riccati type differential equation
 \[
 g_v = \kappa g g_v = \kappa \frac{D_1}{ {h}} g + \kappa \frac{D_2}{h^{1+\kappa_0}} g^2, 
 \]
 which can be solved by standard methods, namely by rewriting it as a linear differential equation with $G = \frac{1}{g}$, i.e.
 \begin{align}\label{eq:G}
  G_v = - \frac{g_v}{g^2} = - \kappa \frac{D_1}{ {h}} G - \kappa \frac{D_2}{h^{1+\kappa_0}}.
 \end{align}
 
 \
 
 {\noindent \bf Step 2. Solution and estimates for the linear ODE.} 
 We proceed by induction in time steps. Suppose \eqref{newpp2} is true up to some time $v_i \geq v_0$. According to Theorem~\ref{prop-Riemn},
 the Riemann invariants $w,z$ do not change their size during the Riemann problem step. It remains to be shown that they are also preserved
 in the ODE step. The differential equation~\ref{eq:G} is considered with initial value at time $v_i$ given by
 \[
  G_i := G(v_i) = e^{-\kappa \gamma_i} = e^{-\kappa \left( \xi + \rho \frac{v-v_0}{h^\kappa} \right)},
 \]
for some function $\rho$. It remains to be shown that for all $v \in [v_0,v_{i+1}]$, $v_{i+1} - v_0 \leq \tau h^{\kappa}$, also
\begin{align}\label{eq:G_i+1}
 G(v) \geq e^{-\kappa \left( \xi + \rho \frac{v-v_0}{h^\kappa} \right)} \geq G_{i+1}.
\end{align}
We show that we can choose $\rho$ and $\tau$ so that \eqref{eq:G_i+1} holds (independent of $i$ and $h$). The solution to the initial
value problem \eqref{eq:G} is
\begin{align*}
 G(v) &= e^{-\kappa \frac{D_1}{ {h}} (v-v_i)} \left( G_i - \kappa \frac{D_2}{h^{1+\kappa_0}} \int_{v_i}^v
 e^{\kappa \frac{D_1}{ {h}} (t-v_i)} dt \right) \\
 &= \left( G_i + \frac{D_2}{D_1}  {h^{-\kappa_0}} \right) e^{-\kappa \frac{D_1}{ {h}}(v-v_i)} - \frac{D_2}{D_1}
 h^{-\kappa_0} \\
 &= e^{- \kappa \left( \xi + \rho \frac{v-v_0}{h^\kappa} \right)} e^{\kappa \left( \rho - \frac{D_1 h^{\kappa}}{ {h}} \right)
 \frac{v-v_i}{h^\kappa}} + \frac{D_2}{D_1}   {h^{-\kappa_0}} \left(  e^{-\kappa \frac{D_1}{ {h}} (v-v_i)} - 1 \right).
\end{align*}
The term  {$e^{- \kappa \frac{D_1}{h} (v-v_i)}$}
is small but negative, hence we have to use  {$e^{\kappa \left( \rho - \frac{D_1 h^{\kappa}}{h} \right) \frac{v-v_i}{h^\kappa}}
> 1$} to compensate for it. Estimates of the exponential map by the first two terms of the Taylor expansion imply
\begin{align*}
 G(v) &\geq e^{- \kappa \left( \xi + \rho \frac{v-v_0}{h^\kappa} \right)} \left( 1 + \kappa \left( \rho
 - \frac{D_1 h^{\kappa}}{ {h}} \right) \frac{v-v_i}{h^\kappa} \right) + \frac{D_2}{D_1}   {h^{-\kappa_0}}
 \left(1 - \kappa \frac{D_1}{ {h}} (v-v_i) - 1 \right) \\
 &= e^{- \kappa \left( \xi + \rho \frac{v-v_0}{h^\kappa} \right)} 
 + \kappa \frac{v-v_i}{h^{1 + \kappa_0}} \left[ \left( \rho - \frac{D_1 h^{\kappa}}{ {h}} \right) h^{1+\kappa_0-\kappa}
 e^{- \kappa \xi} e^{ -\kappa \rho \frac{v-v_0}{h^\kappa}} - D_2 \right],
\end{align*}
 and it remains to be shown that the term in the square bracket is not negative. To achieve this we only have to make sure that
 $\rho$ and $\tau$ are defined in a way that $\rho - D_1 h^{\kappa-1} \geq e^{1+\kappa \xi} D_2$ and
 $\kappa \rho \frac{v-v_0}{h^\kappa} \leq \kappa \rho \frac{v_* -v_0}{h^\kappa} = \kappa \rho \tau = 1$ hold. Since
 $1 + \kappa_0-\kappa \leq 0$,
 \begin{align*}
  \left( \rho -  {D_1 h^{\kappa-1}} \right) & h^{1+\kappa_0-\kappa} e^{- \kappa \xi} e^{ - \rho \frac{v-v_0}{h^\kappa}}
  - D_2 \geq e e^{\kappa \xi} D_2 h^{1+\kappa_0-\kappa} e^{-\kappa \xi} e^{-1} - D_2 \geq 0.
 \end{align*}
 This completes the inductive argument and shows that \eqref{eq:gamma} is true for
 \[
  \gamma(v) = \xi + \rho \frac{v-v_0}{h^\kappa}, \qquad v \in [v_0,v_{i+1}],
 \]
with $\rho := e^{1+\kappa\xi} D_1 + D_2$ and $\frac{v-v_0}{h^\kappa} \leq \tau := \frac{1}{\kappa \rho}$.
\end{proof}

The proof of the analogous statement in the region $\Xi_*$ is now straightforward. Due to the different boundaries, we can
get rid of some more $\frac{1}{h}$ terms and would obtain slightly better constants. In the following, however we assume that $\rho$
is the same constant in both regions $\Omega_*$ and $\Xi_*$.
 
 \begin{theorem}[Bounds for the ODE step in $\Xi_*$] \label{thm:ode-delta}
  Fix $k \in (0,1)$ and $\kappa \geq 1 + 2\kappa_0$. Suppose $1 \leq b \leq C_b$ and
$-  {\frac{1}{h}} \leq a \leq 1 $.  
  Then there exists $\rho >0$ so that the (approximate) Riemann invariants $\wa,\za$ obeying the differential system
  \eqref{eq:wv-zv}--\eqref{eq:wv+zv} with initial values as derived in Lemma~\ref{lem:initial2} satisfy
  \begin{align}\label{newpp}
   \wa,\za \in \log \left( 1 + \frac{1}{h} \right) + \left[ - \xi - \rho \frac{v-v_0}{h^\kappa}, \xi
   + \rho \frac{v-v_0}{h^\kappa} \right],
  \end{align}
 for all $v \in [v_0,v_*]$ with $v_* := v_0+\tau h^\kappa$, with $\tau := \frac{1}{\kappa \rho}$, $\xi$ as in \eqref{xi_2} and $\rho$ a
 positive constant.
 \end{theorem}

\begin{proof} 
 We follow the proof of Theorem~\ref{thm:ode-omega} but assume that
 \begin{align}\label{gamma2}
  w(v),z(v) \in \log \left( 1 + \frac{1}{h} \right) + \left[ -\gamma(v),\gamma(v) \right].
 \end{align}
By \eqref{eq:bounds}, we may estimate
\[
 \pm w_v, z_v \leq A_1 + \frac{A_2}{ {h}} + \frac{A_3(1+h)}{h} e^\gamma + \frac{A_4 (1+h)}{h} e^{(1+2\kappa_0) \gamma}.
\]
 For $h$ small (e.g., $h\leq 1$) and $\kappa \geq 1+ 2\kappa_0$ we thus get and ordinary differential equation of the form
 \[
  \gamma_v = \frac{D_1}{ {h}} + \frac{D_2}{h}e^{\kappa \gamma}.
 \]
Solving the corresponding linearized ODE, we derive as in the proof of Theorem~\ref{thm:ode-omega} that we must choose
$\rho := e^{1+\kappa \xi} D_2 + D_1 \geq e^{1+\kappa \xi} D_2 + D_1  {h^{\kappa-1}}$. 
\end{proof}

 The above results are true for any $k \in (0,1)$ and can be generalized to general existence results for the ODE system
 $\partial_v U = \widetilde S(U,a,b)$ or $\partial_v U = S(U,a,b)$ (assuming that $M,-V>0$ should be preserved). To obtain a good
 control on the time of existence in view of the dynamical formation of trapped surfaces we rely on the above control of the random
 choice method and, moreover, we would like to have that $\kappa < 2$ which we saw in Remark~\ref{rem:kappa} is possible for small
 sound speeds $k$.
 

\subsection{Estimates of the wave speeds and geometric terms}

After the fluid variables $\Ma,\Va$ have been computed using the random choice method of Section~\ref{sec:random}, the (approximate)
geometric variables $\aaa,\ba$ are updated using the integral equations $\eqref{b-new}$ and $\eqref{a-new}$. It remains to be shown
that $\aaa,\ba$ satisfy the bounds stated in Definition~\ref{def:pproperty}. To control the integrals, it is necessary to control the
``size'' of the regions $\Omega_*$ and $\Xi_*$. The boundaries of both regions are defined using an upper bound for the
modulus of all wave speeds of the (homogeneous) Euler system, denoted by $C_*$.

\begin{lemma}[Wave speeds in $\Omega_*$]\label{lem:speeds}
 Fix $k \in (0,1)$ and $\kappa \geq 1+ 2\kappa_0$. Suppose $V,a,b$ satisfy the perturbation property of Definition~\ref{def:pproperty} 
 up to some time $v_i < v_0 + \tau h^{\kappa}$. Then, for $v \in [v_0,v_{i+1}]$, the wave speeds in the
 region $\Omega_*$ are controlled by $C_* := \frac{\Lambda}{h}$ for $h$ sufficiently small, with positive constant $\Lambda$, defined
 by
 \[
 \Lambda :=  C_b \left( 2e \frac{1+k}{1-k} C_0 + 1 \right) . 
 \]
\end{lemma}

\begin{proof}
 By assumption, $V,a,b$ satisfy the bounds of Definition~\ref{def:pproperty} up to time $v_i$. At time $v_i$ the Riemann problem and
 the ODE step are solved to compute $V$ up to time $v_{i+1}$. The geometric variables $a,b$ remain constant in both steps. The wave
 speeds $\lambda_1,\lambda_2$ of Proposition~\ref{prop:31} are
 \[
  b \left( \frac{1+k}{1-k} V + \frac{a}{2} \right) \leq \lambda_i \leq b \left( \frac{1-k}{1+k} V + \frac{a}{2} \right).
 \]
Plugging in the estimates for $V,a,b$ from the perturbation property yields upper and lower bounds for $\lambda_i$ independent of
$(v,r) \in \Omega_*$. In particular, since $V$ is negative and $a \leq 1$, for $h$ small,
 \begin{align*}
  \lambda_i \leq \frac{1}{2} a b \leq \frac{C_b}{2} \leq C_*.
 \end{align*}
 On the other hand, for $v \in [v_0,v_{i+1}]$ and $h$ sufficiently small,
 \begin{align*}
  \lambda_i 
\geq b \left( \frac{1+k}{1-k} b V + a \right) 
& \geq - C_b \left( \frac{1+k}{1-k} C_0 e^{C \frac{v-v_0}{h^\kappa}}
  \left( 1 + \frac{1}{h} \right) -  {\frac{1}{h}} \right) \\
  &\geq - C_b \left( 2 \frac{1+k}{1-k} e^{C \tau} \frac{C_0}{h} - \frac{1}{h} \right) \geq - \frac{\Lambda}{h},
 \end{align*}
 The last inequality is due to Remark~\ref{rem:tau} which states that $e^{C \tau} \leq e$.
\end{proof}

\begin{corollary}\label{cor:speeds}
 Suppose $V,a,b$ satisfy the perturbation property of Definition~\ref{def:pproperty} and $a \geq 0$. Then
  \begin{align*}
     C_* = \frac{1+k}{1-k} C_b C_0 e^{C \tau} \left( 1 + \frac{1}{h} \right).
  \end{align*}
\end{corollary}

We are now in a position to estimate the ``updated'' integral quantities $a$ and $b$ in the random choice method.

\begin{proposition}[Estimates for $a$ and $b$ in $\Omega_*$]\label{prop:ab}
 Fix $k\in (0,1)$  {and} $\kappa \geq 1 + 2\kappa_0$. 
 Suppose that at time $v_0$ the fluid variables $M_0,V_0$ satisfy the initial conditions stated in Proposition~ \ref{prop:initialdata}
 with $\delta \leq \frac{h}{C_1}$.
 Then there exists a positive constant $C_b$ so that for $(v,r) \in \Omega_*$ with $v \in [v_0,v_0+ \tau h^{\kappa}]$, and $h$
 sufficiently small,
 \begin{align} \label{eq:ab}
  1 \leq b(v,r) \leq C_b, \qquad -  {\frac{1}{h}} \leq a(v,r) \leq 1,
 \end{align}
\end{proposition}

\begin{proof}
{\noindent \bf Step 1. Initial time $v_0$.}
Again, we proceed by induction in time steps. Since $b$ is independent of $a$, we consider it first. The initial step at $v_0$ is
true since $b$ is equal to the static solution,
\begin{align*}
 1 \leq b_0(r) = \bs(r) &= e^{4 \pi (1+k^2) \int_0^r \Ms(s) s ds} 
 \leq \bs(r_*+\varDelta) =: C_b.
\end{align*}
 Similarly, for $a$, we have seen in Proposition~\ref{prop:initialdata}, that for an appropriate choice of compact perturbation of the
 static solution,
 \[
 0 < a_0 (r) \leq 1, \qquad r \in [0,r_*+\varDelta].
 \]
Suppose the inequalities \eqref{eq:ab} hold up to time $v_i$. In view of Section~\ref{sec:random}, this is sufficient to compute the
approximate solutions $M,V$ up to time $v_{i+1}$.
 Theorems~\ref{thm:ode-omega} and \ref{thm:ode-delta} (or its equivalent formulation in Definition~\ref{def:pproperty} in terms of
 $M,V$) moreover state certain bounds for $M$ and $V$ valid up to time $v_{i+1}$. This allows us to compute the maximal wave speeds
 by Lemma~\ref{lem:speeds}. We will use those bounds to show that $a,b$ satisfy the above bounds up to time $v_{i+1}$, too.
 
 \
 
{\noindent \bf Step 2. Inductive step for $b$.}
 To estimate $b$ we use the integral formula \eqref{b-new}, and that $M$ is positive,
 \begin{align*}
  b(v,r) &= e^{4 \pi (1+k^2) \int_0^r M(v,s) s \, ds} 
  \leq e^{4 \pi (1+k^2) \left( \int_0^{R_*^-(v)} + \int_{R_*^-(v)}^{\Xi_*^-(v)} 
  + \int_{\Xi_*^-(v)}^{\Xi_*^+(v)} + \int_{\Xi_*^+(v)}^{R_*^+(v)}  \right)} \\
  &= \bs(R_*^-(v)) e^{4 \pi (1+k^2) \left( \int_{R_*^-(v)}^{\Xi_*^-(v)} 
  + \int_{\Xi_*^-(v)}^{\Xi_*^+(v)} + \int_{\Xi_*^+(v)}^{R_*^+(v)} \right)}.
 \end{align*} 
 By Lemma~\ref{lem:speeds} and the assumed bound on $M$, for $h$ small and up to time $v_{i+1} < v_0 + \tau h^\kappa$,
 \begin{align*}
  \int_{R_*^-(v)}^{\Xi_*^-(v)} & \! := \int_{R_*^-(v)}^{\Xi_*^-(v)} M(v,s)s \, ds 
  \leq C_0 e^{C \frac{v-v_0}{h^\kappa}} \left( 1 + \frac{1}{h} \right)^{\kappa_0}
  \left[ \frac{s^2}{2} \right]_{r_*-\delta-C_*(v-v_0)}^{r_*-\delta+C_*(v-v_0)}\\
 &\leq 2 C_0 e^{C \tau} \left( 1 + \frac{1}{h} \right)^{\kappa_0} (r_*-\delta) C_* (v-v_0)
 \leq 2^{\kappa_0+1} e C_0 r_* \tau \Lambda h^{\kappa-1-\kappa_0}
 =: B_1 h^{\sigma},
 \end{align*}
 with $\sigma := \kappa - 1 - \kappa_0 > 0$. Similarly, for $(v,r) \in \Xi_*, v \leq v_{i+1}$, by
 $\delta \leq \frac{h}{C_1}$ of Proposition~\ref{prop:initialdata},
 \begin{align*}
 \int_{\Xi_*^-(v)}^{\Xi_*^+(v)} & \! := \int_{\Xi_*^-(v)}^{\Xi_*^+(v)} M(v,s) s \, ds
 \leq C_0 e^{C \tau} \left[ \frac{s^2}{2} \right]_{r_*-\delta+C_*(v-v_0)}^{r_*+\delta-C_*(v-v_0)} \\
 &\leq 2 C_0 e^{C \tau} r_* (\delta - C_* (v-v_0)) \leq 2e r_* C_0 \frac{h}{C_1} =: B_2 h,
 \end{align*}
 as well as
 \begin{align*}
  \int_{\Xi_*^+(v)}^{R_*^+(v)} & \! := \int_{\Xi_*^+(v)}^{R_*^+(v)} M(v,s) s \, ds 
  \leq C_0 e^{C \tau} \left( 1 + \frac{1}{h} \right)^{\kappa_0}
  \left[ \frac{s^2}{2} \right]_{r_*+\delta+C_*(v-v_0)}^{r_*+\delta-C_*(v-v_0)} \\
  &\leq 2^{\kappa_0+1}e C_0 (r_*+\varDelta) \tau \Lambda h^{\kappa-1-\kappa_0}
 =: B_3 h^{\sigma}.
 \end{align*}
 Therefore we may estimate $b$ up to time $v_{i+1}$ by
 \begin{align*}
  b(v,r) \leq \bs(r_*-\delta) e^{4 \pi (1+k^2) \left( (B_1+B_3) h^{\sigma} + B_2 h \right)} \leq \bs(r_*)
  e^{4 \pi (1+k^2) \left( B_1+B_3 + B_2\right) h^{\min(\sigma,1)}},
 \end{align*}
  which, for $h$ sufficiently small (independent of $v$) is bounded by 
  \[
   b(v,r) \leq \bs(r_*) \frac{\bs(r_*+\varDelta)}{\bs(r_*)} = C_b.
  \]

 The estimate for the lower bound of $b$ follows immediately from \eqref{eq:Einstein1-new} itself, since $M$ is positive everywhere
 (thus so is $b_r$) and $b$ is equal to the static solution $\bs$ at the center, i.e.\ $b(v,r) \geq \bs(0) = 1$ for all $r>0$.
 
 \
 
  {\noindent \bf Step 3. Inductive step for $a$.}
 The geometric term $a$ can now be estimated using the integral formula \eqref{a-new}. We already know that $M,-V$ as well $b$ are
 positive and satisfy certain bounds up to time $v_{i+1}$. Thus for $(v,r) \in \Xi_*, v\in [ v_i,v_{i+1}]$ it is immediate
 that
 \begin{align*}
 a(v,r) &= 1 - \frac{4 \pi (1+k^2)}{r} \int_0^r \frac{b(v,s)}{b(v,r)} M(v,s) \left( 1-2K^2 V(v,s) \right) s^2 \, ds \leq 1.
 \end{align*}
 To estimate $a$ from below, we carefully check all parts involved. For $r \in [R_*^-(v),R_*^+(v)]$, due to the monotonicity of and
 bounds for $b$ (cf.\ Step 2) and the positivity of the static solution $\as$, we have 
 \begin{align*}
  a(v,r) &= 1 - \frac{4 \pi (1+k^2)}{r} \int_0^r \frac{b(v,s)}{b(v,r)} M(v,s) \left( 1-2K^2 V(v,s) \right) s^2 \, ds \\
  &\geq 1 - \frac{R_*^-(v) b(v,R_*^-(v))}{r b(v,r)} \left( 1 - \as(R_*^-(v))\right) - \frac{4\pi(1+k^2)}{rb(v,r)}
 \left[ \int_{R_*^-(v)}^{\Xi_*^-(v)} + \int_{\Xi_*^-(v)}^{\Xi_*^+(v)}
 + \int_{\Xi_*^+(v)}^{R_*^+(v)} \right] \\
  &\geq \frac{R_*^-(v) b(v,R_*^-(v))}{r b(v,r)} \as(R_*^-(v)) - \frac{4\pi(1+k^2)}{rb(v,r)}
 \left[ \int_{R_*^-(v)}^{\Xi_*^-(v)} + \int_{\Xi_*^-(v)}^{\Xi_*^+(v)} + \int_{\Xi_*^+(v)}^{R_*^+(v)} \right] \\
  &\geq - \frac{4\pi(1+k^2)}{r_*-\varDelta} \left[ \int_{R_*^-(v)}^{\Xi_*^-(v)}
  + \int_{\Xi_*^-(v)}^{\Xi_*^+(v)} + \int_{\Xi_*^+(v)}^{R_*^+(v)} \right],
 \end{align*}
 with integral terms $\int_{R_*^-(v)}^{\Xi_*^-(v)}, \int_{\Xi_*^-(v)}^{\Xi_*^+(v)},
 \int_{\Xi_*^+(v)}^{R_*^+(v)}$ as follows.
 The first term may be estimated by
\begin{align*}
 0 \leq \int_{R_*^-(v)}^{\Xi_*^-(v)} & \! := \int_{R_*^-(v)}^{\Xi_*^-(v)} b(v,s) M(v,s) (1 - 2 K^2 V(v,s)) s^2 ds \\
 &\leq C_b C_0 e^{C \tau} \left( 1 + \frac{1}{h} \right)^{\kappa_0} \left( 1 + 2 K^2 C_0 e^{C \tau} \left(1 + \frac{1}{h}\right) \right)
 \left[ \frac{s^3}{3} \right]_{r_*-\delta-C_*(v-v_0)}^{r_*-\delta+C_*(v-v_0)} \\
 &\leq 2^{\kappa_0+2} e^2 K^2 C_b C_0^2 e^{2 C\tau} h^{- 1 - {\kappa_0}}
 \left[ \frac{s^3}{3} \right]_{r_*-\delta-C_*(v-v_0)}^{r_*-\delta+C_*(v-v_0)},
\end{align*}
 where, for $h$ small,
\begin{align*}
 \left[ \frac{s^3}{3} \right]_{r_*-\delta-C_*(v-v_0)}^{r_*-\delta+C_*(v-v_0)} &= \frac{1}{3} C_* (v-v_0) \left[ 3 (r_*-\delta)^2
 + C_*^2 (v-v_0)^2 \right] \\
 &\leq \frac{1}{3} \Lambda h^{-1} \tau h^{\kappa} \left[  3 (r_*-\delta)^2 +  \Lambda^2 h^{-2} e^{2C \tau} \tau^2 h^{2 \kappa} \right]
 \leq 2 \Lambda \tau (r_*-\delta)^2 h^{\kappa - 1} ,
\end{align*}
 since $C_* = \Lambda h^{-1}$ by Lemma~\ref{lem:speeds}.
 Therefore, for some constant $I_1>0$  
 and  {$\sigma = \kappa - 1 - \kappa_0 >0$},  
 \begin{align*}
  \int_{R_*^-(v)}^{\Xi_*^-(v)} &\leq 2^{\kappa_0+3} e^2 K^2 C_b C_0^2 \Lambda \tau (r_*-\delta)^2 h^{\kappa - 2 - \kappa_0}
  \leq 2^{\kappa_0+3} e^2 K^2 C_b C_0^2 \Lambda \tau r_*^2 h^{\kappa -2 - \kappa_0}=: I_1  {h^{-1+\sigma}}.
 \end{align*}
In a similar fashion we derive, by $\delta \leq \frac{h}{C_1}$,
\begin{align*}
 \int_{\Xi_*^-(v)}^{\Xi_*^+(v)} &\leq C_b C_0 e^{C \tau} \left( 1 + 2 K^2 C_0 e^{C \tau} \left(1 + \frac{1}{h}
 \right) \right) \left[ s^3\right]_{r_*-\delta+C_*(v-v_0)}^{r_*+\delta-C_*(v-v_0)} \\
 &\leq 4 K^2 C_b C_0^2 e^{2 C \tau} h^{-1} \frac{\delta}{3} \left( 3 r_*^2 + (\delta + C_*(v-v_0))^2 \right) \\
 &\leq 2^{\kappa_0+2} e^2 K^2 C_b C_0^2 h^{-1} \delta r_*^2 \leq 128 K^2 \frac{C_b C_0^2}{C_1} r_*^2 =: I_2,
 \end{align*}
and
\begin{align*}
 0 \leq \int_{\Xi_*^+(v)}^{R_*^+(v)} &\leq 2^{\kappa_0+2} e^2 K^2 C_b C_0^2 e^{2 C \tau} h^{-1 - \kappa_0}
 \left[ \frac{s^3}{3} \right]_{r_*+\delta-C_*(v-v_0)}^{r_*+\delta+C_*(v-v_0)} \\
 &= 2^{\kappa_0+2} e^2 K^2 C_b C_0^2 e^{2 C \tau} h^{-1 - \kappa_0} \frac{1}{3} C_* (v-v_0)
 \left[ 3 (r_*+\delta)^2 + C_*^2 (v-v_0)^2 \right] \\
 &\leq 2^{\kappa_0+3} e^2 K^2 C_b C_0^2 \Lambda \tau (r_* + \varDelta)^2 h^{\kappa-2 -\kappa_0} =: I_3  {h^{-1+ \sigma}}.
 \end{align*}
 Summing up all contributions we finally derive a (negative) lower bound for $a$. For $h$ sufficiently small, since
 {$\sigma > 0$},   
 \[
  a(v,r) \geq - \frac{4 \pi (1+k^2)}{r_*-\varDelta} \left[ I_2  {h} + (I_1 + I_3)  {h^{\sigma}} \right]  {h^{-1} \geq
  - \frac{1}{h}}. \qedhere
 \]
\end{proof}


\subsection{The total variation estimate} \label{sec:bv}

As mentioned earlier in Section 6, the total variation bound and the consistency are standard and we only provide a sketch. We refer
to \cite{Glimm,GroahTemple} for further details and focus on the derivation of the total variation bound on the approximate solutions.
From this bound, it is a standard matter to deduce that a subsequence converges and we can also check that the limit is a solution of
the Euler system.
   
To this end, denote by $U_{i, j+1}$ the value achieved by the approximate solution $U_\D$ at the point $(v_i, r_{j+1})$, so 
$$
U_{i,j+1} := U_\D(v_i, r_{j+1}), \quad i+j \,\mathrm{even}. 
$$
Let $\Uh _{i,j}$ be the solution to the classical Riemann problem $\Rcal(U_{i-1, j}, U_{i-1,j+2}; v_{i-1}, r_{j+1})$, in which 
$$
U_{i,j+1} := \Uh _{i,j+1} + \int_{v_{i-1}}^{v_j} \widetilde S(v', \Pcal_{v'}\Uh _{i,j+1}) \, dv'.
$$
We divide the $(v,r)$-plane into diamonds $\Diamond_{i,j}$ ($i+j$ even)
with vertices $(r_{i-1,j}, v_{i-1})$, $(r_{i,j-1}, v_{i})$, $(r_{i, j+1},v_{i})$, 
$(r_{i+1,j}, v_{i+1})$. To simplify the notation, we introduce the values of $U_\D$ at the vertices of $\Diamond_{i,j}$ and 
the corresponding Riemann problems by 
$$
U_S:= U_{i-1, j}, \qquad U_W:= U_{i, j-1}, \qquad U_E:= U_{i, j+1}, \qquad U_N := U_{i+1,j},
$$ 
$$
\Uh _W:= \Uh_{i, j-1}, \qquad \Uh _E:= \Uh_{i,j+1}, \qquad \Uh_N:= \Uh_{i + 1,j},
$$  
in terms of which, the strength $\str_*( \Diamond_{i,j})$ of the waves entering the diamond is defined
as
$$
\str_*( \Diamond_{i,j}) := |\str(\Uh_W, U_S)| + |\str(U_S, \Uh_E)|,
$$
whereas the strength $\str^*( \Diamond_{i,j})$ of the waves leaving is 
$$
\str^*( \Diamond_{i,j}) := |\str(U_W, \Uh_N)| + |\str(\Uh_N, U_E)|.
$$
Let $J$ be a spacelike mesh curve, that is a polygonal curve connecting the vertices $(r_{i,j+1}, v_i)$ of 
different diamonds, where $i+j$ is even. We say that waves $(U_{i-1,j}, \Uh_{i,j+1})$ cross the curve $J$ 
if $J$ connects $(r_{i-1,j}, v_{i-1})$ to $(r_{i,j+1}, v_i)$ and similarly for $(\Uh_{i,j-1}, U_{i-1,j})$. 
The total variation $L(J)$ of $J$ is defined as 
$$
L(J):= \sum |\str(U_{i-1,j}, \Uh_{i,j+1})| + |\str(\Uh _{i,j-1}, U_{i-1,j})|, 
$$
where the sum is taken over all the waves crossing $J$. Furthermore, we say that a curve $J_2$ is an immediate successor 
of the curve $J_1$ if they connect all the same vertices except for one and if $J_2$ lies in the future of $J_1$. For the 
difference of their total variation, we have the following result. 

\begin{lemma}[Global total variation estimate] 
\label{diamond}
Let $J_1, J_2$ be two spacelike curves such that $J_2$ is an immediate successor of $J_1$ and let $\Diamond_{i,j}$ be the diamond
limited by these two curves. There exists a uniform constant $C_2$ such that  
$$
L(J_2) - L(J_1) \leq C_2 \, \Delta v \, \str_*(\Diamond_{i,j}),
$$
in which $\Delta v$ denotes the time step length. 
\end{lemma}
 
From this lemma, it is immediate to derive, by induction in time and for all spacelike curve $J$, the uniform bound
$L(J) \leq C_3 e^{C_2 (v_* - v_0)} L(J_0)$, which is equivalent to a uniform total variation on the approximate solutions up to the
time $v_*$.

\begin{proof}
By definition, we have 
$$
\aligned
L(J_2)-L(J_1) &= |\str(U_W, \Uh_N)| + |\str(\Uh_N, U_E)| - |\str(\Uh_W, U_S)| - |\str(U_S, \Uh_E)| 
\\
&= \str^*( \Diamond_{i,j}) -  \str_*( \Diamond_{i,j}).  
\endaligned
$$ 
Observe that $|\str(U_W, \Uh_N)| + |\str(\Uh_N, U_E)| = |\str(U_W, U_E)|$ since $\Uh_N$ is just 
one of the states in the solution of the Riemann problem for $U_W, U_E$. Hence, we can write 
$
L(J_2)-L(J_1) = X_1 + X_2,
$
where
$$
\aligned
&X_1 : = |\str(\Uh_W, \Uh_E)| - |\str(\Uh_W, U_S)| - |\str(U_S, \Uh _E)|,
\\
&X_2 : = |\str(U_W, U_E)| - |\str(\Uh_W, \Uh_E)|.
\endaligned
$$
By the interaction estimate established in Lemma~\ref{lem94} concerning the Euler system in a flat geometry and in
 Eddington-Finkelstein coordinates, we have 
$
X_1 \leq 0. 
$
The term $X_2$ accounts for the effect of the source-terms and geometric terms in the Euler equations. Using that $U_\D$ is uniformly
bounded (for the interval $v \in[v_0, v_*]$ under consideration) and for some constant $C$ we obtain 
$$
\aligned
X_2 &\leq C \, |\str(\Uh_W, \Uh_E)| \, \big( |U_W - \Uh_W| + |U_E - \Uh _E| \big) + C \, \big| (U_W - U_E) - (\Uh _W - \Uh _E) \big|
\\
    &\leq C \Delta v \, |\str(\Uh_W, \Uh_E)| \, \big(
\sup_{v \in[v_{i-1}, v_i]}| \widehat U'_W(v)| + \sup _{v \in[v_{i-1}, v_i]}|\widehat U'_E(v)| \big) 
+ C \Delta v \, |(\Uh _W - \Uh _E)|
\\
    &\leq C \Delta v  \,|\str(\Uh_W, \Uh_E)| 
\leq C \Delta v  \, (|\str(\Uh_W, U_S)| + |\str(U_S, \Uh_E)|), 
\endaligned 
$$
in which we have denoted by $\widehat U_W(v)$ and $\widehat U_E(v)$ the solutions of the ODE associated with the vertex $W$ and $E$, 
and we have used the continuous dependence property $|(\Uh_W - \Uh_E)| = O(1)|\str(\Uh_W,\Uh_E)|$.
\end{proof}


\section*{Acknowledgments}

 The first author (AYB) was supported by an Eiffel excellence scholarship of the French Ministry of Foreign and European Affairs and
 a ``For Women in Science'' fellowship  funded by L'Or\'eal Austria.
 The second author (PLF) was supported by the Centre National de la Recherche Scientifique (CNRS) and by the Agence Nationale de la
 Recherche through the grant ANR SIMI-1-003-01 (Analysis and geometry of spacetimes with low regularity).  
The second author gratefully acknowledges financial support from the National Science Foundation under Grant No. 0932078 000 via the Mathematical Sciences Research Institute (MSRI), Berkeley, where the author spent the Fall Semester 2013 for the program on  ``Mathematical General Relativity''.



\end{document}